\documentclass{article}

\usepackage{arxiv}

\usepackage[utf8]{inputenc}
\usepackage[T1]{fontenc}
\usepackage{hyperref}
\usepackage{url}
\usepackage{booktabs}
\usepackage{amsfonts}
\usepackage{nicefrac}
\usepackage{microtype}
\usepackage{graphicx}
\usepackage{subcaption}
\usepackage{natbib}
\usepackage{doi}
\usepackage{wrapfig}
\usepackage{array}
\usepackage{xcolor}
\IfFileExists{colortbl.sty}{\usepackage{colortbl}}{}
\providecommand{\rowcolors}[3]{}
\providecommand{\columncolor}[2][]{}
\providecommand{\rowcolor}[2][]{}
\usepackage{amsmath}
\usepackage{amssymb}
\usepackage{mathtools}
\usepackage{amsthm}
\usepackage[capitalize,noabbrev]{cleveref}
\usepackage{listings}
\usepackage[section]{placeins}
\usepackage{tikz}
\usepackage{tabularx}
\IfFileExists{makecell.sty}{\usepackage{makecell}}{%
}
\IfFileExists{enumitem.sty}{\usepackage{enumitem}}{%
  \let\origitemize\itemize
  \let\endorigitemize\enditemize
  \renewenvironment{itemize}[1][]{\origitemize}{\endorigitemize}
  \let\origenumerate\enumerate
  \let\endorigenumerate\endenumerate
  \renewenvironment{enumerate}[1][]{\origenumerate}{\endorigenumerate}
}
\IfFileExists{adjustbox.sty}{\usepackage[export]{adjustbox}}{}

\usetikzlibrary{arrows.meta,positioning,calc,fit}

\newcolumntype{L}[1]{>{\raggedright\arraybackslash}p{#1}}
\newcolumntype{R}[1]{>{\raggedright\arraybackslash}p{#1}}
\newcolumntype{Y}{>{\raggedright\arraybackslash}X}
\newcolumntype{C}{>{\centering\arraybackslash}p{0.18\linewidth}}
\newcolumntype{D}{>{\centering\arraybackslash}p{0.22\linewidth}}

\definecolor{blkModel}{RGB}{255,245,204}
\definecolor{blkData}{RGB}{230,245,255}
\definecolor{blkLoop}{RGB}{240,235,255}
\definecolor{blkStep}{RGB}{226,250,236}
\definecolor{thmBox}{RGB}{255,240,240}
\definecolor{eagerBox}{RGB}{240,255,240}

\newcommand*\HLModel{\color{black}\colorbox{blkModel}}
\newcommand*\HLData {\color{black}\colorbox{blkData}}
\newcommand*\HLLoop {\color{black}\colorbox{blkLoop}}
\newcommand*\HLStep {\color{black}\colorbox{blkStep}}

\lstdefinestyle{pytorchBox}{
  language=Python,
  basicstyle=\ttfamily\scriptsize,
  frame=single, rulecolor=\color{black!35}, framerule=0.4pt, framesep=5pt,
  columns=fullflexible, keepspaces=true, showstringspaces=false, tabsize=2,
  keywordstyle=\color{blue!70!black}\bfseries,
  commentstyle=\color{gray!65!black}\itshape,
  stringstyle=\color{green!45!black},
  backgroundcolor=\color{gray!3},
  moredelim=**[is][\HLModel]{@M@}{@M@},
  moredelim=**[is][\HLData ]{@D@}{@D@},
  moredelim=**[is][\HLLoop ]{@L@}{@L@},
  moredelim=**[is][\HLStep ]{@S@}{@S@},
}

\newlength{\PanelH}
\lstdefinelanguage{Lean4}{
  sensitive=true,
  morekeywords={open,def,let,do,IO,Unit},
  morecomment=[l]{--},
  morestring=[b]"
}

\lstdefinestyle{torchleanBox}{
  language=Lean4,
  basicstyle=\ttfamily\scriptsize,
  frame=single, rulecolor=\color{black!35}, framerule=0.4pt, framesep=5pt,
  columns=fullflexible, keepspaces=true, showstringspaces=false, tabsize=2,
  keywordstyle=\color{blue!70!black}\bfseries,
  commentstyle=\color{gray!65!black}\itshape,
  stringstyle=\color{green!45!black},
  backgroundcolor=\color{gray!3},
  emph={NN,API,TorchLean,Module,instantiate,IEEE32Exec},
  emphstyle=\color{teal!65!black}\bfseries,
  moredelim=**[is][\HLModel]{@M@}{@M@},
  moredelim=**[is][\HLData ]{@D@}{@D@},
  moredelim=**[is][\HLLoop ]{@L@}{@L@},
  moredelim=**[is][\HLStep ]{@S@}{@S@},
}

\makeatletter
\IfFileExists{tcolorbox.sty}{
  \@ifundefined{NewStructureName}{%
    \providecommand{\NewStructureName}[2][]{}
    \providecommand{\UseStructureName}[1]{}
    \providecommand{\AssignStructureRole}[2]{}
  }{}
  \@ifundefined{tagstructbegin}{%
    \providecommand{\tagstructbegin}[1][]{}
    
  }{}

  \usepackage[most]{tcolorbox}
  \tcbuselibrary{theorems,breakable}
  \newtcbtheorem[number within=section]{BlueTheorem}{Theorem}{
    enhanced,
    breakable,
    colback=black!3,
    colframe=blue!60!black,
    colbacktitle=blue!60!black,
    coltitle=white,
    fonttitle=\bfseries,
    boxrule=0.6pt,
    arc=2pt,
    left=6pt,right=6pt,top=5pt,bottom=5pt,
    borderline west={2pt}{0pt}{blue!60!black},
    before upper={\sloppy\emergencystretch=1.5em},
  }{thm}
}{
  \newenvironment{tcolorbox}[1][]{%
    \par\smallskip\noindent
    \begin{center}\begin{minipage}{0.98\linewidth}\hrule\smallskip
  }{%
    \smallskip\hrule\end{minipage}\end{center}
    \par\smallskip
  }
  \newenvironment{BlueTheorem}[2]{\begin{theorem}[#1]\label{thm:#2}}{\end{theorem}}
}
\makeatother

\newcommand{\ttbr}[1]{\texttt{#1}}
\IfFileExists{dblfloatfix.sty}{\usepackage{dblfloatfix}}{}

\newcommand{\bpG}[2]{\mathrm{bp}_{G}(#1,#2)}
\newcommand{\EvalG}{\mathrm{evalVec}\,G}
\newcommand{\Tensor}{\mathrm{Tensor}}
\newcommand{\sig}{\mathrm{sig}}
\newcommand{\Fin}{\mathrm{Fin}}
\newcommand{\checkCert}{\mathrm{checkCert}}
\theoremstyle{plain}
\newtheorem{theorem}{Theorem}[section]

\theoremstyle{definition}

\theoremstyle{remark}

\lstdefinelanguage{Lean}{
  keywords={def, inductive, structure, theorem, example, let, in, if, then, else, match, with, type, namespace, import, open, using, where},
  morecomment=[l]{--},
  morestring=[b]",
}
\lstdefinestyle{leanBox}{
  language=Lean,
  basicstyle=\ttfamily\footnotesize,
  frame=single,
  columns=fullflexible,
  keepspaces=true,
  breaklines=true,
  showstringspaces=false,
  keywordstyle=\color{blue}\bfseries,
  commentstyle=\color{gray}\itshape,
  stringstyle=\color{red},
  backgroundcolor=\color{gray!3},
}

\newcommand{\Lean}{\textup{\textsc{Lean}}}
\newcommand{\TorchLean}{\textup{\textsc{TorchLean}}}
\newcommand{\llbracket}{\mathopen{[\![}}
\newcommand{\rrbracket}{\mathclose{]\!]}}
\newcommand{\huan}[1]{}

\title{\TorchLean{}: Formalizing Neural Networks in Lean}
\date{}

\author{
Robert Joseph George$^{1}$,
Jennifer Cruden$^{1}$,
Will Adkisson$^{2}$,\\
\textbf{Xiangru Zhong}$^{3}$,
\textbf{Huan Zhang}$^{3}$,
\textbf{Anima Anandkumar}$^{1}$\\
$^{1}$California Institute of Technology\\
$^{2}$Washington University in St. Louis\\
$^{3}$University of Illinois Urbana-Champaign\\
}

\renewcommand{\shorttitle}{\TorchLean{}: Formalizing Neural Networks in Lean}
\date{}

\hypersetup{
  hidelinks,
  pdftitle={TorchLean: Formalizing Neural Networks in Lean},
  pdfsubject={Semantic infrastructure for formal verification of neural network systems in Lean},
  pdfauthor={Robert Joseph George, Jennifer Cruden, Will Adkisson, Xiangru Zhong, Huan Zhang, Anima Anandkumar},
  pdfkeywords={neural networks, verification, formalization, automatic differentiation, Lean, CUDA, reinforcement learning, certificates},
}

\begin{document}
\maketitle

\begin{abstract}
Neural networks are increasingly deployed in scientific, safety critical, and mission critical pipelines, yet verification and analysis are often performed outside the programming environment that defines, runs, and exports the model. This separation creates a semantic gap between the executed network and the analyzed artifact: guarantees can depend on implicit conventions about operator semantics, tensor layouts, preprocessing, floating-point behavior, graph transformations, accelerated kernels, and externally produced certificates. We present \textsc{TorchLean}, a unified framework for formalizing, executing, and verifying neural networks in Lean 4. \textsc{TorchLean} treats learned models as executable programs and first-class mathematical objects with a shared semantics for computation, verification, and theorem proving. The framework provides a PyTorch style API for typed tensors, layers, objectives, optimizers, automatic differentiation, and graph programs, together with eager and compiled execution paths that lower to a common computation-graph representation. \textsc{TorchLean} supports exact and finite-precision tensor semantics, verified reverse-mode differentiation for supported graph programs, interval and affine bound propagation, CROWN/LiRPA style certificate checking, import and export workflows, and CUDA-backed execution through explicit FFI boundaries. It also includes formal semantic layers for modern learning systems, including attention and FlashAttention, state-space sequence models, diffusion and sampling processes, probability kernels, reinforcement-learning objectives and Markov decision processes, and self-supervised objectives such as masked autoencoding, JEPA-style predictive views, and variance/correlation-based anti-collapse losses. Together, these components provide a semantic foundation for verified machine learning, where executable neural network artifacts, verification procedures, runtime boundaries, and mathematical claims can be stated and related inside a single theorem-proving environment\footnote{Code is available at\url{https://github.com/lean-dojo/TorchLean}}.
\end{abstract}


\section{Introduction}
\label{sec:intro}

Neural networks are increasingly embedded in systems where mistakes have physical or societal cost, and the machine learning community is correspondingly placing more weight on guarantees that go beyond empirical accuracy: robustness to perturbations, satisfaction of hard safety and stability constraints, and conservative bounds on quantities derived from a network's computation. This has driven an active verification ecosystem spanning \emph{solver-based} approaches that reduce verification queries to constraint solving, such as Reluplex and Marabou for piecewise-linear networks~\citep{katz2017reluplex,katz2019marabou}, as well as scalable relaxation and abstract interpretation methods that compute certified enclosures by bound propagation \cite{gowal2018ibp,zhang2018crown,singh2019abstract}. Surveys and systems papers emphasize that such techniques are central to validating learning-enabled components in autonomy, control, and other high-assurance pipelines \cite{xiang2018verification}. Despite this progress, turning verifier outputs into dependable guarantees remains difficult in practice. The issue is not only the strength of a verifier, but also the meaning of the artifact being verified. A trained network may appear as a PyTorch \texttt{nn.Module}, an exported ONNX or FX graph, a compiler IR, a CUDA-backed runtime computation, an external CROWN certificate, or a mathematical object in a paper proof. These artifacts are connected by exporters, graph rewrites, scalar semantics, layout conventions, and runtime libraries. If these boundaries are not made explicit, a property can be proved about one object while the deployed system executes another. We highlight four concrete failure modes.

\begin{figure}[!ht]
    \centering
    \includegraphics[width=1.0\linewidth]{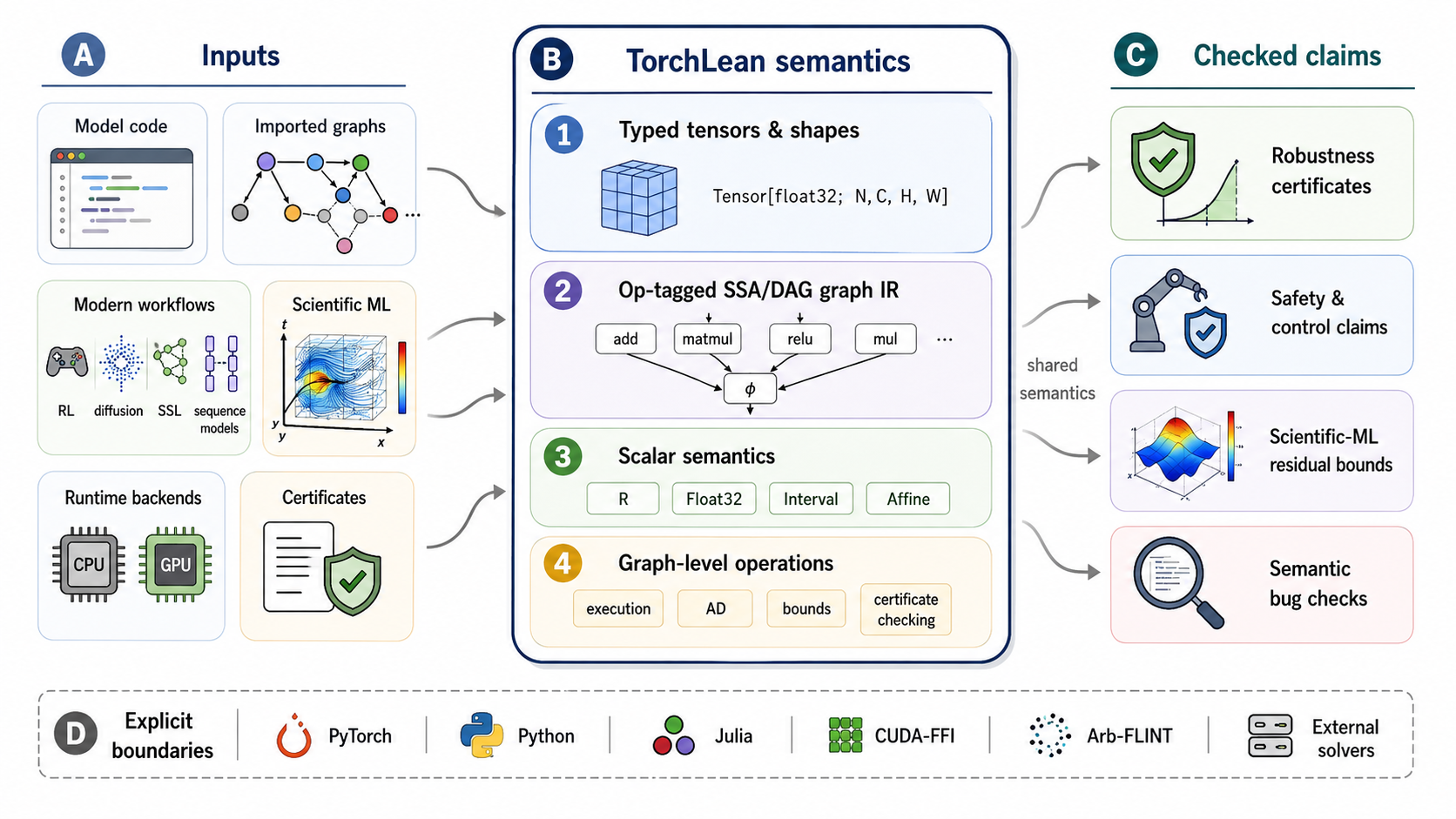}
\caption{\textbf{\TorchLean{} provides a shared semantic layer for verified machine learning.}
Model code, imported graphs, modern ML workflows, scientific ML artifacts, and certificates are funneled into a common typed tensor and operator-tagged SSA/DAG graph semantics in Lean. Execution, automatic differentiation, bound propagation, and certificate checking are all performed against this shared representation, enabling checked claims in robustness, control, scientific machine learning, and bug detection while keeping external runtimes and solvers as explicit boundaries.}
\label{fig:placeholder}
\vspace{-0.4cm}
\end{figure}
\textbf{(1) Export boundaries create semantic drift.}
In modern workflows, the trained model is a program, while verifiers typically consume an exported or recompiled surrogate, such as TorchScript, ONNX, FX, or a custom verification IR. This conversion boundary is exactly where subtle mismatches arise. Trace-based export can miss control flow and dynamic behavior, and even when export succeeds, the resulting graph may only reflect the operations executed in a particular run \cite{pytorch_onnx}. Interchange formats also version operator sets and require models to declare the opsets they rely on \cite{onnx_versioning}. Thus, the meaning of an exported network becomes entangled with tool, version, and backend assumptions that can drift across runtimes. Figure~\ref{fig:semantic-gap} summarizes this gap.

\begin{figure}[!h]
  \centering
  \vspace{-0.5cm}
  \begin{tikzpicture}[
      >=Latex,
      font=\scriptsize,
      node distance=2.5mm and 4mm,
      box/.style={draw, rounded corners=1.5pt, inner sep=2pt, minimum height=5.5mm},
      arr/.style={->, shorten <=0.5pt, shorten >=0.5pt},
      drift/.style={font=\tiny, fill=white, inner sep=0.5pt},
      container/.style={draw, rounded corners=2pt, inner sep=3pt, thick}
    ]
    \node[align=center] (std) at (0,0) {\textbf{Standard pipeline}};
    \node[box, below=3mm of std, xshift=-18mm] (P) {PyTorch};
    \node[box, right=of P] (E) {Export};
    \node[box, right=of E] (O) {ONNX};
    \node[box, right=of O] (S) {Verifier (CROWN/LiRPA)};
    \draw[arr] (P) -- (E);
    \draw[arr] (E) -- (O);
    \draw[arr] (O) -- (S);
    \node[drift, above=0.5mm of E, xshift=2.5mm] {?};
    \node[drift, above=0.5mm of O, xshift=2.5mm] {?};

    \node[align=center] (lnn) at (0,-14mm) {\textbf{\TorchLean{} pipeline}};
    \node[box, below=3mm of lnn, xshift=-12mm] (TL) {\TorchLean{} program};
    \node[box, right=12mm of TL] (IR) {Shared IR};
    \node[box, right=12mm of IR] (Sol) {Verifier (CROWN)};
    \draw[arr] (TL) -- (IR);
    \draw[arr] (Sol) -- (IR);

    \node[container, fit=(TL) (IR),
          label={[font=\tiny]above:Lean~4 (machine-checked semantics)}] (LEANBOX) {};

  \end{tikzpicture}
  \caption{Semantic gap: the standard pipeline (top) chains PyTorch $\to$ Export $\to$ ONNX $\to$ Verifier, with potential drift at each conversion step; \TorchLean{} (bottom) keeps a single shared IR as the semantic target for both execution and verification (Verifier artifacts are checked against this same meaning).}
  \label{fig:semantic-gap}
  \vspace{-0.5em}
\end{figure}
\textbf{(2) Floating-point semantics can invalidate idealized guarantees.}
Most verification methods are stated over real arithmetic or over simplified numerical models, while deployment uses IEEE-754 arithmetic with rounding, overflow, underflow, NaNs, infinities, signed zeros, mixed precision, and backend-specific algorithm choices \cite{ieee754_2019}. Several lines of work show that treating floating point as an implementation detail can invalidate verification conclusions. \citet{jia2020float} demonstrated that numerical error can be exploited to refute robustness claims, motivating conservative modeling of finite-precision effects. More recently, \citet{szasz2025nosoundness} argue that theoretical soundness for a real-valued model does not automatically imply practical soundness for deployed floating-point networks, and \citet{hwang2025fpua} show that robust approximation results for floating-point networks require explicit finite-precision semantics.

\textbf{(3) Certification often stops short of the executed workflow.}
Neural networks are brittle to worst-case perturbations and distribution shifts, motivating certified robustness and safety properties rather than empirical testing alone \cite{szegedy2013intriguing,goodfellow2014explaining}. This need is especially acute for learning-enabled components in control and autonomy, where verification must connect to the implemented artifact \cite{xiang2018verification}. It also appears in scientific machine learning. For example, physics-informed neural networks use automatic differentiation to enforce PDE residuals, so meaningful certificates may need to cover both function values and derivatives \cite{raissi2019pinn}. These settings require a pipeline in which execution, differentiation, and certification refer to the same model semantics.

\textbf{(4) Real ML bugs often live at semantic boundaries.}
Empirical studies of deep learning bugs show that many failures are not exotic theorem-proving corner cases, but ordinary semantic boundary errors: wrong tensor shapes, unintended broadcasting, dtype conversions, layout assumptions, padding conventions, train/eval state, masking mistakes, data preprocessing mismatches, and compiler or runtime transformations \cite{islam2019bugCharacteristics,humbatova2020taxonomy,shen2021dlCompilerBugs}. These bugs are especially dangerous for verification because they can preserve surface-level compatibility while changing the property actually being checked. A robustness certificate for the wrong padding convention, attention mask, batch axis, or exported graph is still a certificate, but not for the intended computation.

We address these challenges with \TorchLean{}, a unified framework for formalizing, executing, and verifying neural networks in \Lean{}~4. Lean is both a practical functional programming language and an interactive theorem prover with a small trusted kernel, extensible automation, metaprogramming, and code generation \cite{moura2021lean4,ebner2017metaprogramming}. This makes it well suited to our goal: not only to state properties of neural networks, but to implement the surrounding pipeline while making its semantic assumptions explicit. \TorchLean{} provides a PyTorch style API for typed tensors, layers, models, objectives, optimizers, automatic differentiation, and training workflows. Programs execute eagerly or lower to a shared operator-tagged SSA/DAG computation-graph IR, which serves as the common semantic target for execution, differentiation, verification, and certificate checking. Prior work has formalized neural-network reasoning in other proof assistants, including Isabelle/HOL tooling for importing models and proving safety or correctness properties \cite{bruckerstell2023isabelle}. We build in Lean because it combines proof checking with a growing ecosystem for mathematics and computer science, including \texttt{mathlib} (\citealp{mathlib2020}) and recent efforts such as \emph{CSLib} (\citealp{barrett2026cslib}). In this setting, \TorchLean{} treats the network definition as the semantic ground truth and organizes specification, runtime execution, and verification around the same graph semantics.

We summarize our contributions as follows.
\begin{itemize}[leftmargin=*, itemsep=0em]
\vspace{-0.2cm}
\item \textbf{\TorchLean{}: PyTorch style neural-network programming in \Lean{}.}
We provide a Lean based API for defining tensors, layers, models, objectives, optimizers, data loaders, and multi-epoch training loops. The runtime supports eager execution and a compiled mode that lowers programs to the shared IR. CPU execution is supported by default, while CUDA-backed dense/matmul and batched-matmul paths, reductions, broadcasts, views, convolution and pooling families, FFT/spectral kernels, normalization kernels, attention/fused-attention paths, gather/scatter utilities, positional/RoPE helpers, and selective-scan-style sequence kernels are exposed through explicit native FFI boundaries. These kernels are treated as runtime interfaces, not as silently verified code.

\item \textbf{A shared IR for execution, differentiation, and verification.}
\TorchLean{} compiles supported programs to an operator-tagged SSA/DAG computation graph with a precise denotation. The same graph object is used by the executor, compiled evaluator, reverse-mode automatic differentiation, IBP/CROWN style bound propagation, and certificate checking. For the supported fragment, we prove a graph-parametric reverse-mode theorem: backprop computes the adjoint Fr\'echet derivative of $\llbracket G \rrbracket$ for any well-typed graph $G$ satisfying local derivative-correctness hypotheses.

\item \textbf{Explicit scalar and finite-precision semantics.}
The same definitions can be instantiated over exact mathematical scalars for proofs, executable finite-precision scalars for runtime behavior, and interval or affine domains for bounds. We provide a rounding model for proofs, executable IEEE-style binary32 components, endpoint interval reasoning, and explicit external-producer boundaries, including Arb/FLINT-style oracles for rigorous transcendental enclosures when needed.

\item \textbf{Verification and certificate checking on the shared semantics.}
On the shared IR, \TorchLean{} implements IBP, CROWN/LiRPA style affine relaxations, and certificate checking. External artifacts are not accepted as guarantees by fiat; they are checked against the graph semantics consumed by Lean. We validate this pipeline on robustness, control-style safety, PINN and derivative-bound examples, VNN-COMP style ONNX/VNN-LIB slices, and regression tests designed to expose semantic boundary failures.

\item \textbf{Modern ML workflows and semantic bug checks.}
\TorchLean{} includes runnable and formalized components for supervised vision and tabular models, GPT-style attention, FlashAttention specifications, Mamba/state-space layers, diffusion and sampling processes, FNO/PINN-style scientific ML, PPO/Gymnasium rollouts, Markov-kernel MDP semantics, self-supervised objectives such as MAE and JEPA-style predictive views, PyTorch round trips, 3D vision certificates, and a bug zoo of semantic boundary checks. Python, Julia, Arb/FLINT, PyTorch, Gymnasium, CUDA, and CROWN family solvers appear as explicit producers or runtimes; Lean remains the checker for the formal artifacts that cross those boundaries.

\end{itemize}

\section{Methodology}
\label{sec:method}

\paragraph{Computation graphs and our IR (definitions).}
Modern ML systems commonly represent a model as a \emph{computation graph}: a directed graph in which \emph{nodes} are primitive operations (e.g., \texttt{MatMul}, \texttt{Conv}, \texttt{ReLU}) and \emph{edges} carry tensors (intermediate values) between operations. This is the standard representation used by deployment/exchange formats such as ONNX, where a graph is a side-effect-free computation composed of nodes that call operators and whose dataflow must admit a topological evaluation order. A graph is a \emph{directed acyclic graph (DAG)} if it has no directed cycles; equivalently, it admits a topological evaluation order.

An \emph{intermediate representation (IR)} is a compiler style, machine- and language-independent program representation designed to be a stable target for analysis and transformation. In \TorchLean{}, the IR is an operator-tagged computation graph: each node carries an explicit \emph{operator tag} identifying the primitive it denotes, such as \texttt{linear}, \texttt{relu}, \texttt{conv2d}, \texttt{attention}, or \texttt{softmax}, together with shape metadata. Execution and verification therefore interpret nodes by the same primitive semantics rather than by an implicit convention inherited from an external framework. We store graphs in \emph{static single assignment (SSA)} form: every intermediate value is defined exactly once and then referenced by subsequent nodes. SSA is a standard compiler IR discipline that simplifies dataflow reasoning and enables deterministic evaluation and induction over the node list \cite{cytron1991ssa}.

Mathematically, for a scalar semantics $\alpha$ and tensor shape $s$, write $\Tensor_\alpha(s)$ for tensors of shape $s$ with scalar entries interpreted in $\alpha$. Each operator tag $\tau$ has a typed signature
$
  \sig(\tau) = (s_1,\ldots,s_k) \to s
$
and a denotation
\[
  \llbracket \tau \rrbracket_\alpha :
  \Tensor_\alpha(s_1) \times \cdots \times \Tensor_\alpha(s_k)
  \to \Tensor_\alpha(s).
\]
A well-typed SSA/DAG graph $G$ denotes a function
\[
  \llbracket G \rrbracket_\alpha :
  \Tensor_\alpha(s_1)\times\cdots\times\Tensor_\alpha(s_m)
  \to
  \Tensor_\alpha(t_1)\times\cdots\times\Tensor_\alpha(t_r),
\]
obtained by evaluating nodes in topological order. The same syntax can be interpreted over exact reals, finite-precision execution domains, intervals, or affine relaxations:
\[
  \llbracket G \rrbracket_{\mathbb{R}},
  \qquad
  \llbracket G \rrbracket_{\texttt{Float32}},
  \qquad
  \llbracket G \rrbracket_{\texttt{Interval}},
  \qquad
  \llbracket G \rrbracket_{\texttt{Affine}}.
\]

Our goal is to eliminate the semantic gap between the \emph{model that is executed} during training or inference and the \emph{model that is analyzed} by bounds or certificates. \TorchLean{} treats the network definition as the single semantic reference point and provides multiple interpretations over the same meaning. Concretely, a model compiles to a shared operator-tagged computation-graph IR in SSA/DAG form, so the denotation is total and deterministic by evaluation in topological order. This IR is the common target for (i) training and inference, (ii) verified reverse-mode automatic differentiation, and (iii) Lean based bound propagation and certificate checking. External verifiers are optional producers rather than trusted authorities: their outputs are interpreted as \emph{certificates} checked against the Lean semantics of the shared graph, reducing the trusted computing base.

Beyond standard feed-forward networks, the same interface also supports modern workflow-level operators and objectives. Attention and FlashAttention style fused operators are specified by equivalence to masked scaled dot-product attention,
\[
  \mathrm{Attn}(Q,K,V,M)
  =
  \mathrm{softmax}\!\left(\frac{QK^\top}{\sqrt d}+M\right)V,
\]
where the mask is part of the formal operator semantics. Diffusion/sampling steps, reinforcement-learning returns and Bellman operators, state-space sequence layers, and self-supervised masking or view objectives are represented as typed specifications connected to the same graph and tensor semantics.

\subsection{Specification layer}
\label{sec:method-spec}

Mainstream ML frameworks represent tensor shapes dynamically: tensors carry runtime shapes, and incompatibilities surface late, often as runtime exceptions or subtle semantic bugs. Empirically, tensor shape faults are among the most prevalent classes of deep learning bugs and frequently lead to crashes \cite{wu2021empiricalstudytensorshape}. \TorchLean{} takes a semantic view: shape constraints are part of the \emph{type} of a tensor, so ill-shaped programs are unrepresentable. This makes theorem statements cleaner, reduces repeated ``shapes match'' side conditions, and moves many plumbing errors from testing to type checking.

\textbf{Core datatypes (shapes and tensors).}
A tensor is indexed by a scalar domain $\alpha$ and a shape $s$; shapes form an inductive tree, where common ML shapes are nested applications of \texttt{dim}. Tensors are represented structurally as total index functions. This functional view is well suited for proofs: tensor operations are defined by recursion on shape, and extensional equality reduces tensor equality to pointwise equality. The specification layer intentionally does \emph{not} commit to a storage layout such as row-major or column-major order, making it suitable as a semantic reference across execution backends.

\vspace{-0.6em}
\begin{lstlisting}[language=Lean,basicstyle=\ttfamily\scriptsize,frame=single]
/-- Shape-indexed tensors. -/
inductive Tensor (a : Type) : Shape -> Type where
  | scalar : a -> Tensor a .scalar
  | dim    : forall {n s}, (Fin n -> Tensor a s) ->
             Tensor a (.dim n s)
\end{lstlisting}
\vspace{-0.35em}

In this model,
\[
  \Tensor_\alpha(\mathtt{scalar}) = \alpha,
  \qquad
  \Tensor_\alpha(\mathtt{dim}(n,s)) = \Fin(n) \to \Tensor_\alpha(s).
\]
Thus a tensor is a total function over a finite index set. Pointwise maps, zips, reductions, reshapes, broadcasts, and masks are explicit typed operations. In particular, broadcasting is not a silent runtime convention; it is represented by a shape transformation with its own typing and semantic lemmas.

\textbf{Typed modules (compositional networks).}
At the specification level, each layer is a typed morphism between tensor spaces,
\[
  L : \Tensor_\alpha(s_{\mathrm{in}})
      \to
      \Tensor_\alpha(s_{\mathrm{out}}),
\]
so composing networks is analogous to composing \texttt{nn.Module}s in PyTorch, but with shape contracts enforced by the type checker. In \TorchLean{}, trainable parameters are carried as a shape-indexed heterogeneous list, one tensor per parameter shape, so parameter routing is type-driven. Composition is only definable when intermediate shapes match; ill-shaped networks do not typecheck.

For example, a linear layer with
$
  W : \Tensor_\alpha(d_{\mathrm{out}}\times d_{\mathrm{in}}),
  \qquad
  b : \Tensor_\alpha(d_{\mathrm{out}})
$
denotes
\[
  \mathrm{Linear}_{W,b}(x)_i
  =
  b_i + \sum_{j<d_{\mathrm{in}}} W_{ij}x_j.
\]
The same definition can be interpreted over $\mathbb{R}$ for proof level reasoning, over \texttt{Float32} or \texttt{IEEE32Exec} for execution, or over interval and affine domains for verification. Because repeated functional updates, such as SGD or Adam steps, can build long closure chains, \TorchLean{} provides \emph{materialization}, which rebuilds a tensor into an array-backed normal form with the same extensional meaning but faster evaluation.

\textbf{Modern operator contracts.}
The specification layer also records contracts for modern neural-network components. Masked attention is specified directly as a mathematical operator, with causal, padding, or arbitrary masks included in the input semantics. FlashAttention style fused or tiled implementations are related to this denotation by equality theorems at the specification level, rather than by treating a fused kernel as a new informal primitive. Sequence/state-space layers are specified as finite recurrences; diffusion and sampling components are specified by noising and denoising transition maps; reinforcement-learning components are specified by returns, temporal-difference residuals, and Bellman-style operators; and self-supervised objectives are specified by explicit view, mask, and target contracts.

\begin{figure*}[!t]
\centering
\begin{subfigure}[t]{0.485\linewidth}
\centering
\begin{lstlisting}[style=pytorchBox]
@M@import torch.nn as nn@M@

# Model + optimizer (semantics live in the runtime)
@M@model = nn.Linear(2, 1)@M@
@M@opt = torch.optim.SGD(model.parameters(), lr=0.2)@M@

# Training data
@D@x = torch.tensor([[1.,0.],[0.,1.],[1.,1.]])@D@
@D@y = torch.tensor([[2.],[-3.],[-1.]])@D@

# Training loop (autodiff + updates external)
@L@for step in range(10):@L@
@S@  opt.zero_grad()@S@
@S@  loss = nn.functional.mse_loss(model(x), y)@S@
@S@  loss.backward(); opt.step()@S@
\end{lstlisting}
\end{subfigure}\hfill
\begin{subfigure}[t]{0.485\linewidth}
\centering
\begin{lstlisting}[style=torchleanBox]
@M@open TorchLean@M@

-- Model + samples (compiled mode)
@M@def model := sequential [linear 2 1, relu]@M@
@D@def samples := [(x1, y1), (x2, y2), (x3, y3)]@D@

-- Explicit Float32 semantics for this run
def main : IO Unit := do
  let mod <- instantiate model (scalar := IEEE32Exec) .compiled
  let opt := Optim.sgd (scalar := IEEE32Exec) (lr := 0.2)
  for step in [0:10] do
    let loss := mse (mod.forward batch.x) batch.y
    backward loss
    opt.step mod.params
\end{lstlisting}
\end{subfigure}

\caption{\textbf{PyTorch vs.\ \TorchLean{}.} We highlight comparable blocks: model/setup (yellow), data (blue), loop structure (purple), and per-step updates (green). In \TorchLean{}, the same model can be executed, lowered to the shared IR, differentiated, and passed to verification.}
\label{fig:pytorch-vs-torchlean}
\vspace{-0.5cm}
\end{figure*}

\subsection{Runtime layer}
\label{sec:method-runtime}

\textbf{Design target.}
The runtime layer connects proof level definitions to runnable training and inference. The user-facing workflow mirrors PyTorch: define a module, run a forward pass, call backward, and update parameters. The key difference is that runtime execution is organized around a verifier semantic target. A \TorchLean{} program can run eagerly, recording a dynamic tape, or compile to the shared operator-tagged SSA/DAG IR consumed by differentiation, bound propagation, and certificate checking.

\textbf{Training and data path.}
The runtime includes reusable ML infrastructure rather than only isolated forward passes: CSV/NPY loaders, minibatch tensor loaders, JSON training logs, parameter import/export, checkpointing, and optimizers including SGD, momentum, Adagrad, RMSProp, Adam, AdamW, and Adadelta. A typical training loop has the usual mathematical shape
\[
  \theta_{t+1}
  =
  \mathrm{Opt}\bigl(\theta_t,\nabla_{\theta}L(\theta_t;B_t),s_t\bigr),
\]
where \(B_t\) is the minibatch, \(s_t\) is optimizer state, and the loss is computed by a \TorchLean{} program with a chosen scalar semantics. The important point is that the same model parameters \(\theta_t\), loss graph, and gradient computation can later be lowered to the verifier IR rather than reconstructed through an informal export path.

\textbf{Two execution modes, one semantic target.}
\textbf{Eager mode} executes operations immediately and records a tape in the style of define-by-run systems. Each tape entry stores the operator tag, parent references, forward values needed by the VJP rule, and an accumulator for cotangents. \textbf{Compiled mode} lowers the same program to the shared operator-tagged SSA/DAG IR and evaluates that graph. This is not \texttt{torch.compile}-style kernel optimization. Its role is semantic: it produces the graph object that theorems and checkers consume.

For a program \(p\) with parameters \(\theta\), lowering produces a graph \(G\) and parameter store \(P_\theta\). The intended correspondence is
\[
  p \Downarrow (G,P_\theta)
  \qquad\Longrightarrow\qquad
  \mathrm{eval}_{\mathrm{src}}(p,\theta,x)
  =
  \llbracket G \rrbracket_\alpha(P_\theta,x),
\]
for the supported forward fragment and scalar semantics \(\alpha\). This is the semantic bridge that allows us to develop in a PyTorch-like style while proving properties of the graph artifact.

\textbf{Reverse-mode AD.}
Reverse-mode automatic differentiation is proved once at the graph level. Let \(G\) be a well-typed SSA/DAG graph over \(\mathbb{R}\), let \(\llbracket G\rrbracket\) be its denotation, and let \(\bar y\) be an output cotangent. Backpropagation computes the adjoint derivative:
\[
  \mathrm{backprop}(G,x,\bar y)
  =
  \bigl(D\llbracket G\rrbracket(x)\bigr)^{\!*}\bar y .
\]
Operationally, this is implemented by a reverse topological sweep:
\[
  \bar v_i
  \;{+}{=}\;
  \mathrm{VJP}_{\tau_j,i}
  \bigl(v_{i_1},\ldots,v_{i_k},\bar v_j\bigr)
  \qquad
  \text{for every edge } i\to j .
\]
Accumulation is essential: if a value is consumed by multiple downstream nodes, the corresponding cotangent is the sum of all downstream contributions. The graph theorem follows by composing local derivative-correctness lemmas for primitive operators.

\begin{BlueTheorem}{Reverse-mode AD correctness}{rev-correct}
Let \(G\) be a well-typed SSA/DAG computation graph over \(\mathbb{R}\) with input context \(\Gamma\). Assume
\ttbr{GraphFDeriv\allowbreak Correct}\,\(G\). Define
\(\bpG{x}{seed}:=\mathrm{backpropVec}(G,x,seed)\). Then for any input \(x\) and cotangent \(seed\),
\[
  \bpG{x}{seed}
  =
  \bigl(D\,(\EvalG)(x)\bigr)^\top seed .
\]
For non-smooth primitives, such as ReLU at \(0\), we use the pointwise variant
\ttbr{GraphFDeriv\allowbreak CorrectAt} under explicit side conditions; the executable convention is deterministic at kink points.
\end{BlueTheorem}

For classification losses and embedding-style models, the runtime separates differentiable tensor values from non-differentiable integer data. Labels, token ids, and row indices are provided through a separate index channel. Gradients flow through floating-point tensor values, such as an embedding table, but not through the integer selectors:
\[
  W \mapsto W[\mathrm{idx}]
  \quad\text{is differentiable in }W,
  \qquad
  \mathrm{idx}\text{ is a read-only selector.}
\]
This avoids mixed-dtype differentiable graphs while still supporting common ML patterns such as cross-entropy with integer labels and embedding lookup.

\textbf{CUDA and Lean FFI.}
\TorchLean{} also provides an optional native runtime path through Lean's foreign-function interface. In the default build, CPU implementations and stubs keep the repository portable. With CUDA enabled, Lean external declarations call C/CUDA wrappers that allocate or receive device buffers, dispatch kernels, and return runtime handles or copied tensors to Lean. The CUDA surface includes dense matrix operations, cuBLAS-backed matmuls and batched matmuls, deterministic reductions, broadcasts, views, reshapes, convolution/pooling and transposed-convolution kernels, FFT/spectral kernels for FNO-style examples, softmax/log-softmax, normalization kernels, attention-oriented and fused-attention kernels, gather/scatter utilities, positional/RoPE helpers, and selective-scan-style sequence kernels.

A typical boundary has the following shape: Lean exposes a typed wrapper, while the implementation lives in C/CUDA. The CUDA path is exposed through Lean \texttt{@[extern]} declarations; the detailed FFI surface and native wrapper contracts are described in Appendix~\ref{app:trust}. A representative boundary is:

\begin{lstlisting}[style=leanBox,
caption={Representative CUDA runtime boundary through Lean FFI.},
label={lst:cuda-ffi-boundary}]
@[extern "torchlean_cuda_buffer_bmm"]
opaque bmm
  (A B : Buffer) (batch m n p : UInt32) : Buffer

@[extern "torchlean_cuda_buffer_flash_attention_fwd"]
opaque flashAttentionFwd
  (Q K V mask : Buffer)
  (hasMask batch n d : UInt32) (scale : Float) : Buffer
\end{lstlisting}
The CUDA path is deliberately not hidden inside the logic. Native kernels are not Lean theorems and are not part of Lean's trusted kernel. Each native call is attached to a formally specified operator tag \(\tau\), whose meaning is already given by the shared semantics. The intended conformance statement is
$
  \mathrm{native}_{\tau}(x)
  \approx
  \llbracket \tau \rrbracket_{\texttt{Float32}}(x),
$
under explicit preconditions about shape, dtype, layout, device memory, aliasing, rounding mode, reduction order, and determinism. When this agreement is tested rather than proved, it is recorded as a runtime conformance assumption. Thus CUDA provides an execution backend, not a second semantics.

This distinction matters for kernels whose hardware behavior may depend on scheduling or library choices. For example, parallel reductions and some pooling/attention backward passes may use nondeterministic accumulation orders; cuBLAS and custom CUDA kernels may select implementation variants; and fused attention kernels may change memory layout and evaluation schedule. In \TorchLean{}, the theorem-facing claim is stated against the graph/operator semantics, while the CUDA implementation is required to refine that semantics under a declared deployment configuration.

\textbf{Interop as producer/checker plumbing.}
Runtime support also includes PyTorch state-dict import/export, PyTorch graph-capture import into the \texttt{torchlean.ir.v1} dialect, IR-to-PyTorch code generation, JSON artifact loaders, Julia subprocess wrappers for small certificate producers, a Gymnasium bridge for RL rollouts, PPO-style rollout helpers, and an Arb/FLINT oracle path for validated numeric enclosures. These integrations are not hidden inside theorem statements. External tools produce weights, traces, rollouts, candidate bounds, or enclosures; Lean either checks the property-bearing artifact or records a named oracle/conformance assumption.

The runtime layer therefore has two roles. It makes \TorchLean{} usable as an ML programming environment, and it preserves the semantic path from executable artifacts to checked claims. The price is that high performance backends and external producers remain explicit boundaries. This is intentional: the paper's formal claims are about the Lean graph semantics and the artifacts Lean checks, not about unverified native code by default.

\subsection{Floating Point Semantics}
\label{sec:method-floats}

Many verification claims hinge on subtle numerical behavior such as rounding, overflow/underflow, NaN/Inf propagation, signed zeros, non-associativity, mixed precision, and library-level conventions for operations such as \texttt{min}, \texttt{max}, pooling, reductions, and transcendental functions. IEEE~754 standardizes floating-point formats, rounding rules, and exceptional values \cite{ieee754_2019}. In Lean, however, the built-in runtime floating-point types are \emph{opaque to the kernel}: floating-point values are not encoded in the logic, so the kernel cannot compute with or reason about them without additional assumptions. Consequently, \TorchLean{} separates fast execution from proved semantics and makes the trust boundary explicit.

For \textbf{explicit numeric semantics}, \TorchLean{} makes numerical assumptions first-class by instantiating the same model over multiple scalar domains~$\alpha$, each serving a distinct role. We use $\alpha=\mathbb{R}$ for clean reference semantics in analytic reasoning and verified differentiation; enclosure domains such as $\alpha=\texttt{Interval}$ and affine domains for sound region-wise bounds in verification pipelines; and explicit finite-precision semantics for execution. For executable Float32, we provide \texttt{IEEE32Exec}, a Lean-defined bit-level model of IEEE-754 binary32, including signed zeros, subnormals, NaNs/Infs, and rounding behavior. In parallel, we provide proof level rounding-on-$\mathbb{R}$ models for compositional error envelopes, where each primitive is specified as ``compute in $\mathbb{R}$ then round'' with lemmas that bound and compose rounding error through whole graphs, in the spirit of verified floating-point libraries such as Flocq~\citep{boldo2011flocq}.

Lean's runtime \texttt{Float} and \texttt{Float32}, CUDA kernels, cuBLAS/cuFFT, and vendor libraries remain available for fast execution, but they are treated as explicitly trusted or validated backends due to kernel opacity and native runtime boundaries. In addition, \TorchLean{} builds an interval/enclosure layer spanning both proof level and executable semantics. We implement Float32 endpoint interval arithmetic on top of \texttt{IEEE32Exec} and prove operation-level soundness theorems that computed endpoint boxes conservatively enclose the corresponding real or extended-real interpretations. For rigorous transcendental bounds, we optionally integrate Arb/FLINT via an explicit oracle boundary, using it as a certificate/enclosure generator while keeping the trusted computing base explicit.

\begin{table}[t]
  \centering
  \vspace{-0.8cm}
\caption{Trust and scalar semantics in \TorchLean{}, including explicit native runtime and validated-numerics boundaries.}
  \label{tab:trust-semantics}
  \setlength{\tabcolsep}{3pt}
  \renewcommand{\arraystretch}{1.05}
  \begin{tabular}{@{}p{0.24\linewidth}p{0.32\linewidth}p{0.40\linewidth}@{}}
    \toprule
    \textbf{Mode} & \textbf{Purpose} & \textbf{Semantics and trust} \\
    \midrule
    $\mathbb{R}$ & Proofs, reference Autograd & Exact real arithmetic; reference semantics for theorem statements. \\
    Interval / affine & IBP, CROWN/LiRPA bounds & Sound enclosure domains with proved transfer rules where available. \\
    FP32\,/\,NF & Rounding-aware theorems & Round-on-$\mathbb{R}$ proof model with compositional error envelopes. \\
    \texttt{IEEE32Exec} & Executable Float32 model & Lean-defined binary32 semantics for core ops; hardware matching is separate. \\
    Arb/FLINT oracle & Rigorous transcendentals & External validated ball/interval enclosures; explicit non-kernel oracle boundary. \\
    Lean \texttt{Float32} & Fast execution & Runtime implementation opaque to the kernel; treated as a deployment assumption or validated backend. \\
    CUDA/native FFI & Accelerated examples & Optional Lean FFI-backed execution for selected kernels; governed by native primitive agreement contracts. \\
    \bottomrule
  \end{tabular}
  \vspace{-0.7cm}
\end{table}

\textbf{FP32/NF $\leftrightarrow$ \texttt{IEEE32Exec}: internal refinement.}
For the IEEE-754 core arithmetic implemented in \texttt{IEEE32Exec}, we prove an \emph{internal} refinement to the FP32 round-on-$\mathbb{R}$ model on finite executions, excluding NaN/Inf and overflow. We establish per-operator theorems of the form
$
  \texttt{toReal}(\mathrm{op}(x,y))
  =
  \mathrm{fp32Round}(\cdots)
$
and lift them to a compositional bridge showing that real-valued evaluation of \texttt{IEEE32Exec} expressions agrees with the corresponding FP32/NF specification on the finite path.

\subsection{Verification layer}
\label{sec:method-verification}

\textbf{Verification is a statement about semantics.}
The verification layer expresses goals as Lean theorems about the denotation of a compiled graph $\llbracket G \rrbracket$: robustness margins, output bounds, invariance and safety constraints, Lyapunov or barrier inequalities, PINN residual bounds, ODE corridor conditions, spline certificate obligations, and dataset-backed certified-accuracy claims. The verifier is not the claim. It is a mechanism for producing intermediate bounds, relaxations, witnesses, or constraints which are then checked and used to discharge a semantic property of $\llbracket G \rrbracket$. All enclosure theorems are stated against a specified execution semantics; relating results to hardware \texttt{Float32} or CUDA execution is an explicit refinement or validation assumption at the deployment boundary.

\textbf{Native bound propagation over the shared IR.}
Because the computation-graph IR is operator-tagged and typed, bound propagation operates on the same object used by execution, eliminating any separate export semantics to trust. Formally, bounds live in an abstract domain $\mathcal{D}$ with a concretization map $\gamma$ to sets of concrete values. A local abstract transformer $\widehat{\tau}$ for an operator $\tau$ is sound when
\[
  x_i \in \gamma(a_i)\ \text{for all } i
  \quad\Longrightarrow\quad
  \llbracket \tau \rrbracket(x_1,\ldots,x_k)
  \in
  \gamma\bigl(\widehat{\tau}(a_1,\ldots,a_k)\bigr).
\]
Soundness for a full graph follows by induction over the same topological node order that defines $\llbracket G \rrbracket$.

We include a sound core of interval bound propagation (IBP) \cite{gowal2018ibp}, which propagates node wise enclosures $[l_i,u_i]$ forward through the graph using per-operator transfer rules stated against the same operator denotations that define $\llbracket G \rrbracket$. Building on this, \TorchLean{} implements CROWN/LiRPA style affine propagation \cite{zhang2018crown,xu2020automatic}, including forward affine propagation and objective-dependent backward components. Affine relaxations capture correlations that intervals miss and typically yield substantially tighter output bounds through compositions of linear layers and monotone activations. Both IBP and affine propagation are anchored to the same IR denotation, rather than to an external verifier-side model.

\textbf{Certificates and a reduced trusted computing base.}
When stronger tightness is required than the native engines provide, \TorchLean{} adopts a certificate/checker architecture: external solvers act as producers of artifacts, while Lean checks the finite data needed to conclude a theorem about $\llbracket G \rrbracket$ \cite{necula1997pcc}. The checker validates shapes, operator tags, topological order, parent bounds, local inequalities, phase constraints, and final property obligations. The intended theorem shape is
\[
  \checkCert(G,\Phi,C)=\texttt{true}
  \quad\Longrightarrow\quad
  \Phi(\llbracket G \rrbracket).
\]
In this style, we support an $\alpha/\beta$-CROWN certificate dialect in which the certificate supplies IBP pre-activation boxes, per-node affine bounds, $\alpha$ parameters for unstable-ReLU lower relaxations, and optional $\beta$ phase vectors encoding active/inactive constraints consistent with the IBP intervals \cite{xu2021fast,wang2021betacrown}. The checker replays the same per-node step semantics in Lean and accepts only if the provided bounds match the recomputed or checked obligations under the declared canonicalization policy.

\textbf{Broader verification workflows.}
The current verification layer also includes graph-CROWN/LiRPA entry points, executable certificate checkers, robustness workflows, VNN-COMP style ONNX/VNN-LIB slices, PINN graph builders and derivative-residual bound helpers, ODE corridor certificates, spline and piecewise-polynomial certificate checks, 3D geometry certificates, Lyapunov/CROWN oracle-backed workflows, and regression tests for graph rewrite or backend drift. These workflows differ in how much is proved internally and how much is delegated to external producers, but they follow the same discipline: each accepted claim is tied to a graph, tensor, or certificate semantics stated in Lean.

\begin{figure*}[!t]
\centering
\begin{tikzpicture}[
    >=Latex,
    node distance=5mm and 18mm,
    every node/.style={font=\scriptsize},
    scale=0.96,
    transform shape,
    module/.style={
        draw,
        rounded corners,
        align=center,
        inner sep=6pt,
        minimum width=3.1cm,
        minimum height=0.92cm
    },
    core/.style={
        draw,
        rounded corners,
        align=center,
        inner sep=4pt,
        fill=blue!10,
        minimum width=3.1cm,
        minimum height=0.72cm
    },
    block/.style={
        draw,
        rounded corners,
        align=center,
        inner sep=4pt,
        minimum width=3.1cm,
        minimum height=0.72cm
    },
    layer/.style={block, fill=green!10},
    model/.style={block, fill=yellow!10},
    runtime/.style={block, fill=orange!10},
    verify/.style={block, fill=red!10},
    source/.style={block, fill=cyan!10},
    ffi/.style={block, fill=orange!18},
    ir/.style={
        draw,
        rounded corners,
        align=center,
        inner sep=5pt,
        fill=purple!10,
        dashed,
        minimum height=0.82cm
    },
    lab/.style={
        font=\tiny,
        fill=white,
        inner sep=1pt
    }
]

\node (a1) {$\alpha=\mathbb{R}$ (proofs)};
\node[right=of a1] (a2)
  {$\alpha=\texttt{Rat}/\texttt{Float32}/\texttt{IEEEExec}$ (execution)};
\node[right=of a2] (a3)
  {$\alpha=\texttt{Interval}/\texttt{FP32}$ (bounds)};

\node[module, fill=blue!20, below=2mm of a1] (spec)
  {\textbf{Specification}\\\textit{mathematical semantics}};
\node[module, fill=orange!20, below=2mm of a2] (rt)
  {\textbf{Runtime}\\\textit{execution and training}};
\node[module, fill=red!20, below=2mm of a3] (vf)
  {\textbf{Verification}\\\textit{bounds and certificates}};

\node[core, below=4mm of spec] (spec_core)
  {Core\\Tensors, shapes\\Ops over $\alpha$};

\node[layer, below=3mm of spec_core] (spec_layers)
  {Layers\\Linear, conv, attention\\Normalization, etc.};

\node[model, below=3mm of spec_layers] (spec_models)
  {Models\\MLP, CNN, Transformer\\RNN/LSTM/GRU, SSM, diffusion};

\draw[->] (spec) -- (spec_core);
\draw[->] (spec_core) -- (spec_layers);
\draw[->] (spec_layers) -- (spec_models);

\node[runtime, below=4mm of rt] (autograd)
  {Autograd\\Reverse-mode AD\\Tape and graph};

\node[runtime, below=3mm of autograd] (torchlean)
  {\textbf{\TorchLean}\\PyTorch style API\\Eager and compiled};

\node[runtime, below=3mm of torchlean] (optim)
  {Training and workflows\\Optimizers, RL, SSL\\Sampling, import/export};

\draw[->] (rt) -- (autograd);
\draw[->] (autograd) -- (torchlean);
\draw[->] (torchlean) -- (optim);

\node[verify, below=4mm of vf] (ibp)
  {IBP\\Interval bounds\\Sound propagation};

\node[verify, below=3mm of ibp] (crown)
  {\textbf{CROWN/LiRPA}\\Affine relaxations\\Dual certificates};

\node[verify, below=3mm of crown] (cert)
  {Certificate checking\\Native checkers + CLI\\external artifacts optional};

\draw[->] (vf) -- (ibp);
\draw[->] (ibp) -- (crown);
\draw[->] (crown) -- (cert);

\node[source, below left=7mm and 6mm of optim] (lowering)
  {Optional import / lowering\\Python, Julia, PyTorch\\ONNX, FX};

\node[ffi, below right=7mm and 6mm of optim] (native)
  {Optional execution backend\\CUDA, cuBLAS, cuFFT\\custom kernels};

\node[
    draw,
    dashed,
    rounded corners,
    fit=(lowering)(native),
    inner sep=4pt,
    label={[font=\tiny]above:optional interfaces}
] (optbox) {};

\node[ir, below=18mm of optim, minimum width=0.82\textwidth] (irnode)
  {\textbf{Operator tagged shared IR}\\
   \textit{SSA/DAG computation graph}\\
   \textit{single semantic target for execution, verification, and lowering}};

\draw[->, dashed]
  (spec_models.south) to[out=-95, in=115]
  node[lab, pos=0.34] {compile}
  (irnode.west);

\draw[->, dashed]
  (torchlean.south) to[out=-88, in=90]
  node[lab, pos=0.55] {build / compile}
  (irnode.north);

\draw[->, dashed]
  (lowering.south) to[out=-90, in=120]
  node[lab, pos=0.52] {lower / import}
  ([xshift=-26mm]irnode.north);

\draw[->, dashed]
  (cert.south) to[out=-95, in=0]
  node[lab, pos=0.30] {verify / check}
  (irnode.east);

\draw[->, dashed]
  ([xshift=28mm]irnode.north) to[out=60, in=-90]
  node[lab, pos=0.55] {execute via FFI}
  (native.south);

\end{tikzpicture}
\caption{
Comprehensive architecture of \TorchLean{}. The specification, runtime, and verification layers share a single operator-tagged SSA/DAG computation-graph IR as the semantic target. Models are defined once, executed or trained by the runtime, and verified through CROWN bounds and certificate checking against the same graph semantics.
}
\label{fig:torchlean_architecture}
\vspace{-0.8em}
\end{figure*}
\vspace{-0.3cm}
\section{Results}
\label{sec:results}
We evaluate \TorchLean{} as a semantic ML systems stack. A model definition can be executed for training or inference, lowered to a typed operator-tagged graph, analyzed by native bound and differentiation passes, and validated by certificate checkers against one mathematical denotation. The goal is not to outperform specialized training frameworks or solver-optimized verifiers, but to make the object being trained, differentiated, bounded, checked, and accelerated semantically visible. The results are organized around three components. First, we show that the framework supports recognizable modern ML workflows rather than only toy verifier networks. Second, we keep the VNN-COMP style mini-suite and the original three case studies: certified robustness, neural-controller verification, and PINN residual bounds. Third, we add semantic bug examples showing how typed shapes, masking contracts, and explicit numerical semantics catch failures that can otherwise survive ordinary execution or export.

\begin{figure*}[t]
  \centering
  \includegraphics[width=0.99\textwidth]{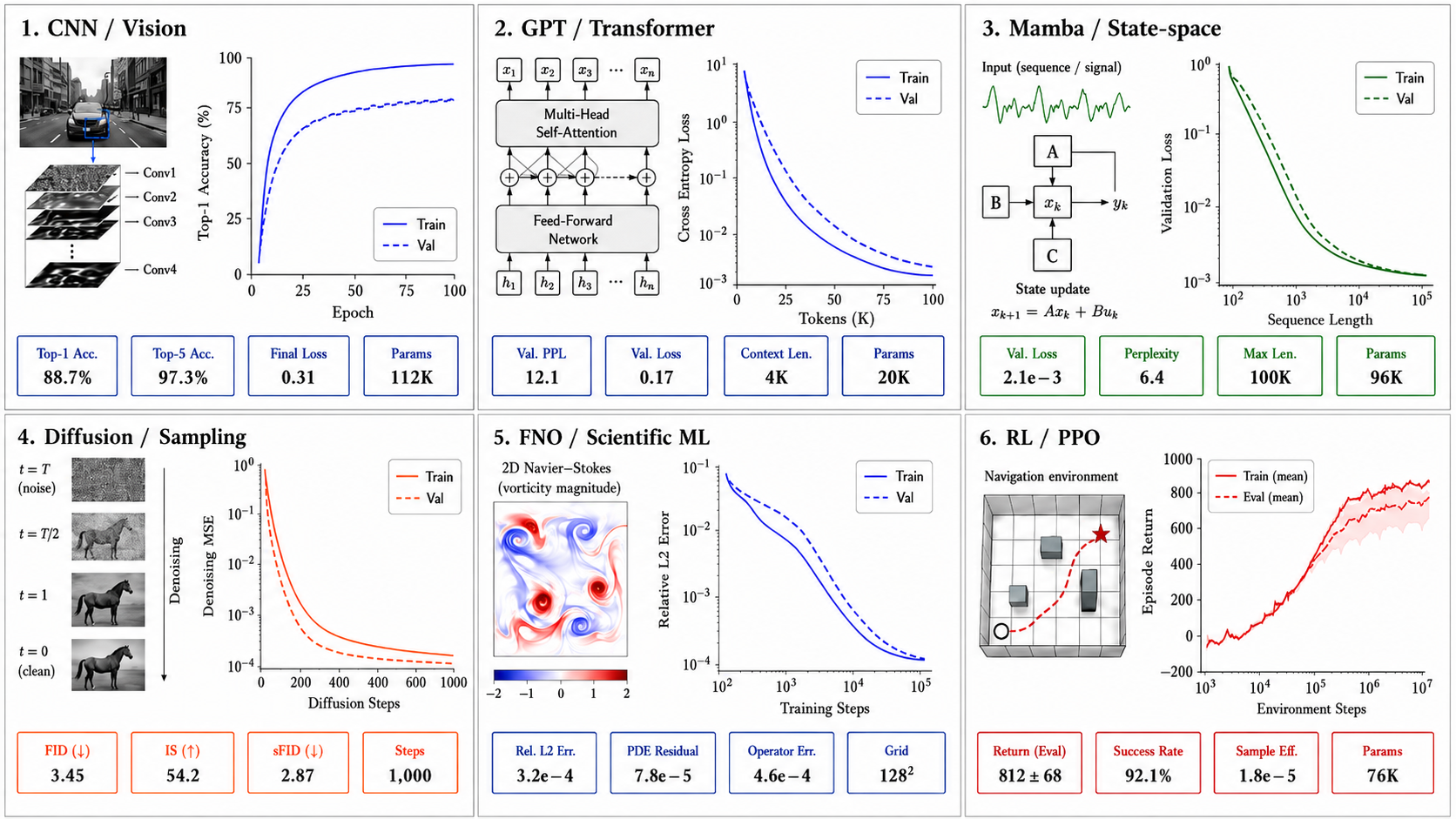}
  \caption{\textbf{Representative model families exercised in \TorchLean{}.}
  The examples cover vision, transformer-style language models, Mamba/state-space sequence models, diffusion/sampling, Fourier-neural-operator/scientific ML workflows, and reinforcement learning. The plotted curves illustrate executable training and evaluation workflows across these model families. The formal role of the examples is that each can be connected to the typed tensor, runtime, graph-lowering, finite-precision, and checking infrastructure described in the methodology.}
  \label{fig:architecture-results-grid}
  \vspace{-0.5cm}
\end{figure*}
\textbf{Model families and modern workflow coverage.}
Figure~\ref{fig:architecture-results-grid} summarizes the current executable surface. The examples exercise ordinary ML structure: minibatch loaders, multiple epochs, optimizer state, saved parameters, loss curves, sampling loops, rollout tensors, imported weights, and generated certificates. They also exercise the newer parts of the codebase: CUDA/native FFI paths for dense/matmul and batched-matmul kernels, reductions, broadcasts, views, convolution/pooling and transposed-convolution kernels, FFT/spectral operators, softmax/log-softmax, normalization, attention/fused attention, gather/scatter, positional/RoPE helpers, selective-scan-style sequence kernels, and shape operations; PyTorch/ONNX/VNN-LIB interop; Arb/FLINT validated-numerics calls; Gymnasium/PPO-style RL bridges; and semantic regression tests.

These examples are not leaderboard training results; they show that verification-oriented semantics can be attached to workflows that look like recognizable ML programs. A CNN or FNO example stresses tensor shapes, convolutional/spectral operators, finite-precision execution, and optional CUDA kernels. A GPT-style model stresses tokenization, causal masking, generation, and attention semantics. A Mamba/state-space example stresses recurrence, scan-like execution, and non-anticipation. Diffusion and self-supervised examples stress masking, views, schedules, latent variables, and reconstruction or prediction losses. RL examples stress rollouts, returns, advantages, policy/value artifacts, and simulator boundaries.

\begin{wraptable}[10]{r}{0.44\linewidth}
\vspace{-1.2em}
\centering
\caption{VNN-COMP style MNIST-FC slice. Checked rows are Lean-verified; CUDA is an untrusted GPU producer.}
\vspace{-0.6em}
\scriptsize
\setlength{\tabcolsep}{3.5pt}
\renewcommand{\arraystretch}{1.03}
\begin{tabular}{@{}llcc@{}}
\toprule
\textbf{Path} & \textbf{Method} & \textbf{Safe/$n$} & \textbf{Time} \\
\midrule
Lean & IBP & 0/30 & 6.05s \\
Lean & CROWN-Obj & 6/30 & 22.60s \\
Lean & CROWN+$\alpha$ & 6/30 & 14.21s \\
\midrule
CPU producer & $\alpha$-CROWN & 13/30 & 8.39s \\
CUDA producer & $\alpha$-CROWN & 13/30 & 7.67s \\
\bottomrule
\label{tab:vnncomp-mini}
\end{tabular}
\vspace{-1.2em}
\end{wraptable}

For these model families, the meaningful result is not only that an example runs, but that the surrounding development exposes a specification surface. FlashAttention style blocked or fused implementations are related to this denotation at the specification level; CUDA kernels remain execution backends behind Lean FFI boundaries. State-space models are specified by finite recurrences and scan contracts, diffusion models by noising and denoising transitions, RL by rollout algebra, returns, advantages, finite MDP views, Markov-kernel MDP semantics, and Bellman-style operators, and self-supervised learning by finite mask/view/target objectives.

\vspace{-0.37cm}
\paragraph{Verification artifacts and producer/checker boundaries.}
\label{sec:results-verification}
The verification stack is evaluated as a checker-oriented system. Native IBP and CROWN/LiRPA style passes operate over the shared IR. External solvers, Python/Julia scripts, Arb/FLINT, CUDA kernels, and imported ONNX or VNN-LIB files are treated as producers or runtime boundaries. Lean checks the shape, graph, interval, affine, and schema obligations that are sufficient for the theorem being claimed, or records a named assumption when the producer itself is outside the kernel. Thus a successful verification result is a statement about the Lean-side graph semantics, not an implicit claim that every external runtime or exporter has been verified. We adopt the VNN-COMP convention that each benchmark instance is an ONNX network paired with a VNN-LIB property specification \cite{vnncompSite,vnnlibSite}. A lightweight Python export step converts ONNX/VNN-LIB into compact JSON bundles consumed by a Lean runner, which checks a \emph{sufficient UNSAT} condition by replaying IBP/CROWN style bounds against the shared IR semantics (fast runtime \texttt{Float}). Table~\ref{tab:vnncomp-mini} reports both fully checked TorchLean results and CUDA-backed producer results on the same MNIST-FC slice. The checked rows are Lean-verified claims about the shared IR. The CUDA row uses an $\alpha$-CROWN GPU producer on an NVIDIA A100 to generate stronger candidate bounds; those GPU results become trusted TorchLean claims only when the exported artifacts are replayed by the Lean checker.

\label{sec:demos}
\textbf{Certified Robustness.}
We certify an $\ell_\infty$ margin condition for a digits linear classifier (\texttt{sklearn} digits, 8$\times$8$\rightarrow$64 features; 64$\rightarrow$10; $\varepsilon=0.02$). The model is nominally correct on $349/360$ test inputs and certified robust on $318/360$ by a Lean checker that replays the exported bound artifact and verifies the standard logit-margin inequality (Appendix~\ref{app:verification}). The checker is lightweight (0.032\,ms average, 0.057\,ms maximum per example), but the point is semantic rather than competitive: robustness is attached to the checked graph, not to an unchecked training script or exporter.

\textbf{Neural controller.}
We consider a two stage controller-verification workflow: external training/search proposes a neural feedback controller $u(x)$ together with a Lyapunov candidate $V(x)$, and Lean then certifies region-based safety/stability by checking CROWN/LiRPA enclosures for Lyapunov inequalities over input regions (Appendix~\ref{app:verification}). This setup mirrors recent \emph{two stage} stabilizing-controller pipelines: Stage~1 learns an initial region of attraction using Zubov-inspired boundary sampling, and Stage~2 fixes this region and iteratively refines the networks to eliminate counterexamples (CEGIS) discovered within it \cite{li2025two}. We bound $V(x)$ and $\dot V(x)=\nabla V(x)\cdot f(x,u(x))$ on a region, using verified autograd to compute $\nabla V$, and discharge the resulting Lyapunov/safety conditions as Lean-checked statements about the shared semantic model. Figure~\ref{fig:controller-three-settings} compares three execution settings for this workflow: (i) Python-only (Stage~1+2 in PyTorch Float32, with optional Lean checking), (ii) all-\TorchLean{}, and (iii) hybrid (Stage~1 in PyTorch, Stage~2 and post-checking in \TorchLean{}). Holding Stage~1 weights fixed and using a common Stage~2 baseline (width 100; 10 PGD candidates; 1 PGD step), Python Stage~2 takes $\approx 0.015$s and yields 9/10 positive-loss candidates, while \TorchLean{} Stage~2 under explicit IEEE-754 Float32 semantics (\texttt{IEEE32Exec}) takes $\approx 1212$s and yields 7/10 positives; \TorchLean{} additionally computes a native CROWN enclosure for the scalar loss over a small input box (e.g., an upper bound $\approx 0.1869$).

\begin{figure}[!t]
    \centering
    \includegraphics[width=0.65\linewidth]{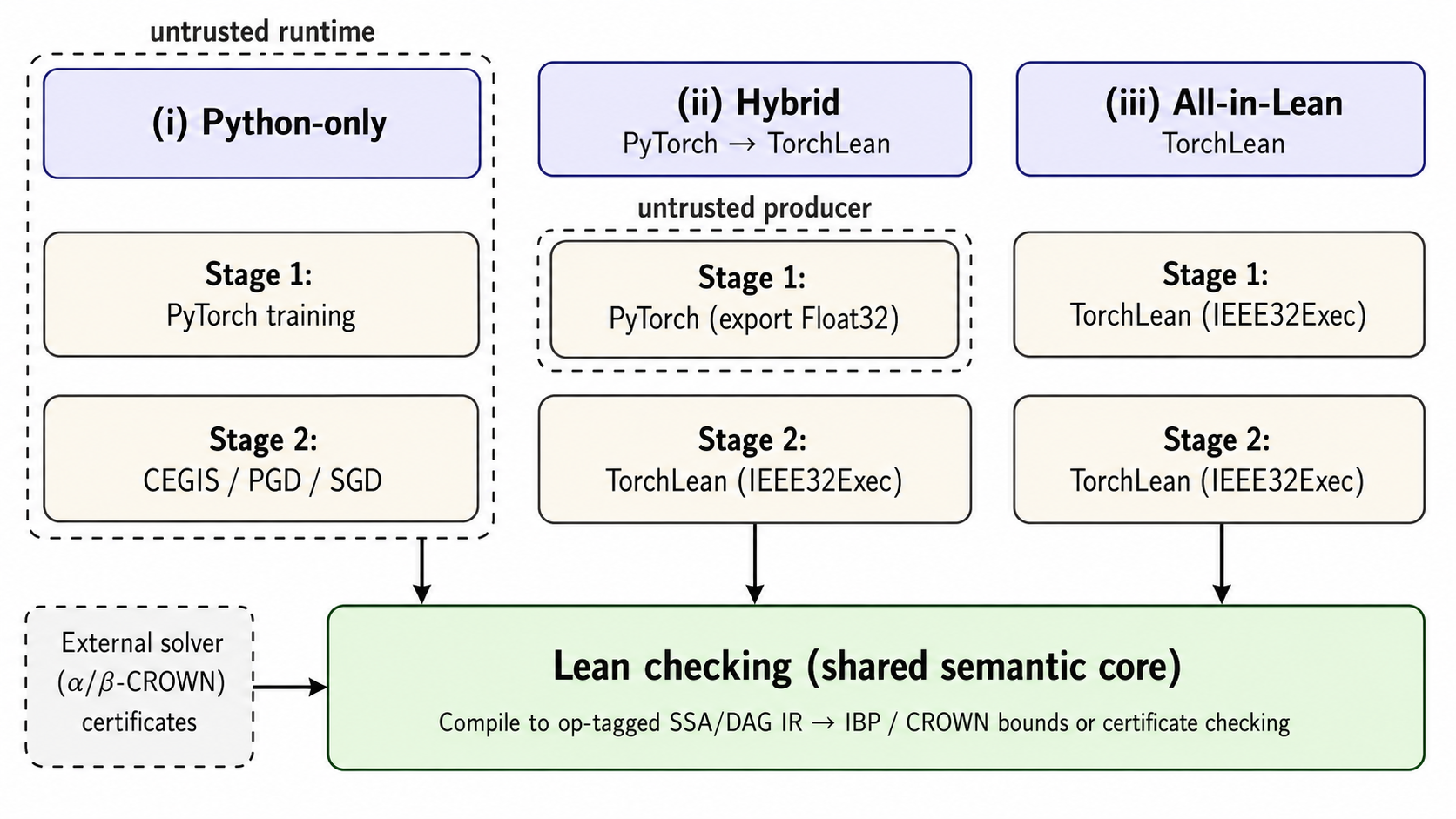}
\caption{Neural-controller workflow in three execution settings. Certification claims are checked against the shared operator-tagged SSA/DAG IR semantics in Lean; external optimizers and runtime systems act as untrusted producers.} \label{fig:controller-three-settings}
\vspace{-1.0cm}
\end{figure}

\textbf{PINN residual bounds.}
We demonstrate full verification of physics-informed neural networks trained to satisfy PDEs. The workflow begins with training a PINN model in Python (e.g., viscous Burgers equation $u_t + u \cdot u_x - \nu u_{xx} = 0$) using PyTorch autograd to compute derivatives for the PDE residual loss. After training, we export weights to JSON and load them into Lean. For verification, we compute bounds on $u(x)$ via IBP/CROWN, and compute interval enclosures for $u'(x)$ and $u''(x)$ via dedicated first/second-derivative bound-propagation passes on the same operator-tagged graph (covering the smooth ops used in the PINN demos, e.g.\ \texttt{tanh} and linear layers), rather than by recursively differentiating the backward graph. These bounds are then combined to certify that the PDE residual $|\mathcal{R}(u_\theta)(x)|$ is bounded within tolerance $\varepsilon$ at verification points (Appendix~\ref{app:verification}).

\textbf{Semantic bug examples.}
Figure~\ref{fig:bug-examples} shows three representative semantic boundary failures. The first is a shape/broadcasting bug: a bias $b\in\mathbb{R}^{C}$ should be applied to logits $X\in\mathbb{R}^{B\times C}$ as $Y_{ij}=X_{ij}+b_j$, but a wrong broadcast axis can silently change the function while still producing a plausible tensor. The second is a temporal bug in sequence models: an incorrect causal mask or stale KV-cache entry can allow future-token information to influence the current output. The third is a numerical bug: a naive log-sum-exp or log-softmax computation can overflow under Float32, while the stable formulation remains finite. These examples are small by design. They are not meant to be benchmark wins; they isolate the kind of error that makes verification pipelines fragile. In each case, the model may execute, export, or even produce a plausible downstream artifact, while the intended semantics has changed.
\begin{figure*}[!t]
  \centering
  \includegraphics[width=0.99\textwidth]{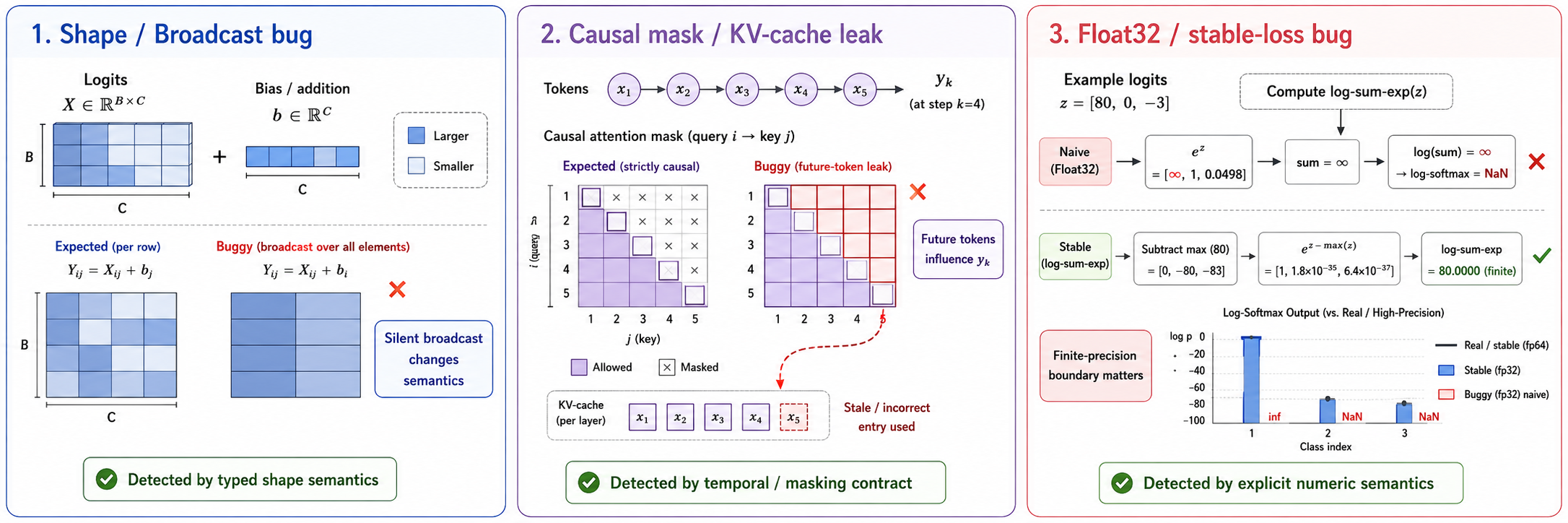}
  \caption{\textbf{Representative semantic boundary bugs.}
  Left: typed shape semantics distinguishes intended row-wise bias addition from an unintended broadcast. Middle: temporal and masking contracts expose causal-mask or KV-cache leaks in sequence models. Right: explicit finite-precision semantics separates a stable log-sum-exp/log-softmax computation from a naive Float32 implementation that can overflow or produce NaNs.}
  \label{fig:bug-examples}
  \vspace{-0.4cm}
\end{figure*}

\textbf{Numerical semantics stress tests.}
Numerical stress tests explain why we make scalar semantics explicit rather than treating floating point as an implementation detail. Endpoint evaluation alone is not a reliable enclosure mechanism for transcendentals, whereas the Arb-backed pipeline provides rigorous real enclosures (Appendix~\ref{app:trust}). For core arithmetic, directed rounding matters: a ties-to-even addition can be enclosed by \texttt{IEEE32Exec} directed endpoints but collapse under naive runtime \texttt{Float32}. Signed-zero guards can also force principled widening, for example when an interval denominator contains both real zero and signed-zero floating encodings. 

\vspace{-0.5cm}
\section{Related work}
\label{sec:related}
\vspace{-0.3cm}
\textbf{Neural network verification (solver-based and bound-propagation).}
Neural network verification is central in safety and robustness analysis \cite{kaulen20256th,li2025two,jiabstract,chen2024verification}. Broadly, the literature spans (i) \emph{solver-based} methods such as Reluplex and Marabou~\citep{katz2017reluplex,katz2019marabou}, which encode verification problems as satisfiability/constraint queries, and (ii) \emph{abstract-interpretation / relaxation} methods such as IBP \cite{gowal2018ibp,moore2009introduction} and linear bound propagation (CROWN/DeepPoly and related abstractions) \cite{dvijotham2018dual,gehr2018ai2,raghunathan2018certified,singh2018fast,singh2019abstract,wang2018efficient,wong2018provable,zhang2018crown}. These tools often operate on exported artifacts (ONNX/TorchScript/custom IRs) and therefore inherit an additional trust boundary at the export/interpretation step. This pipeline is standardized and stress-tested in VNN-COMP, where instances are packaged as an ONNX network together with a VNN-LIB property specification \cite{brix2024vnncomp,vnncompSite,vnnlibSite}, making the export/interface boundary an explicit part of the evaluation regime.

\textbf{CROWN family optimizers and certificate checking.}
CROWN family methods derive output bounds by propagating sound linear relaxations of nonlinearities and viewing the result through a dual objective; \(\alpha\)-CROWN~\citep{xu2021fast} tightens bounds by optimizing relaxation parameters, and \(\beta\)-CROWN integrates splitting/branch-and-bound to further tighten bounds~\cite{wang2021betacrown}. \TorchLean{} implements a Lean based CROWN/LiRPA core over our shared operator-tagged IR semantics (proved-sound IBP, a basic forward affine pass, and an objective-dependent backward/dual pass). {Reimplementing the full \(\alpha/\beta\)-CROWN optimization stack inside Lean including the parameter-optimization heuristics and branch-and-bound search is ongoing work.} When strongest current tightness is needed, we instead treat external solvers as untrusted producers and check their exported bounds/certificates against the same IR semantics, keeping the trusted computing base to a small checker.

\textbf{Formalization in theorem provers and Lean infrastructure.}
ITPs have been used to formalize ML-relevant mathematics (e.g., generalization bounds \cite{Bagnall_Stewart_2019}, activation-function libraries \cite{aleksandrov2023formalizingpiecewiseaffineactivation}, and network translations \cite{10.1007/978-3-031-98208-8_12}). Lean's ecosystem (notably \texttt{mathlib} (\citealp{mathlib2020})) makes these developments practical, but full pipeline neural verification remains hard because it requires tensor representations that scale to modern architectures, a sound differentiation story, and numerical semantics that are explicit enough to connect proofs to execution. \texttt{TensorLib} provides a verified tensor library for Lean~\citep{tensorlib} and, as of this writing, is under active development.

\textbf{SciLean, verified floats, and complementary developments.}
SciLean is a Lean based library for scientific computing with array abstractions and automatic differentiation \cite{scilean}. \TorchLean{} targets a different point in the design space: typed tensor computation graphs with shape-indexed tensors, a graph-parametric reverse-mode correctness theorem for SSA/DAGs, and certificate-checked verification workflows. On numerics, verified floating-point frameworks such as Flocq~\citep{boldo2011flocq} motivate the separation we use between rounding models for proofs (FP32/NF) and an executable IEEE-style kernel (\texttt{IEEE32Exec}), with an explicit trust boundary to runtime floats. HopfieldNet formalizes Hopfield/Boltzmann-style energy arguments in Lean \cite{cipollina2025hopfield}. Our Hopfield material is positioned as a complementary case study within a shared tensor/graph semantics: we mechanize the standard energy-decrease and convergence-style results in the same framework; see Appendix~\ref{app:verification-schemas}.

\vspace{-0.3cm}
\section{Discussion}
\label{sec:discussion}
\vspace{-0.3cm}
\TorchLean{} is best understood as a semantics layer for neural network systems rather than as a replacement for existing training frameworks or specialized verifiers. The main design choice is to make the network definition the semantic reference point and to connect execution, differentiation, bound propagation, certificate checking, and external artifacts to a shared operator-tagged SSA/DAG IR. This addresses a common source of fragility in verified ML workflows: the artifact that is trained, exported, optimized, or executed is often not exactly the artifact that is analyzed. In \TorchLean{}, this boundary is made explicit. Imported ONNX/VNN-LIB files, PyTorch or Julia producers, Arb/FLINT enclosures, CROWN family certificates, Gymnasium rollouts, and CUDA/native kernels are treated as producers or runtime interfaces whose artifacts must either be checked in Lean or recorded as named assumptions. This separation is also important for modern ML workflows. Verification is no longer only about small feed-forward classifiers. Attention masks, FlashAttention style fused operators, state-space scans, diffusion schedules, self-supervised views, RL returns, and scientific ML residuals all carry mathematical contracts that can be broken by small implementation choices. \TorchLean{} makes these contracts visible as typed specifications. The current system therefore supports two complementary uses: proving semantic properties of supported graph fragments, and detecting semantic boundary failures such as shape/broadcasting mistakes, causal-mask leaks, unstable losses, export drift, and certificate inconsistencies.

The main limitations are coverage and conformance. The verified fragment is necessarily smaller than the full surface of PyTorch, CUDA, and modern ML libraries. Extending \TorchLean{} requires adding operator semantics, local derivative rules, abstract transformers, finite-precision models, and certificate checkers for new primitives. Native CUDA execution is useful for scale, but it remains outside Lean's kernel; connecting it more tightly to the formal semantics requires stronger conformance testing, verified wrappers, or eventually verified kernels for important primitives. Similarly, external certificate producers are useful for tightness and performance, but the formal guarantee comes only from what the Lean checker validates. These limitations are deliberate: the system favors explicit boundaries over implicit trust.

\vspace{-0.4cm}
\section{Conclusion}
\label{sec:limits-concl}
\vspace{-0.4cm}
\TorchLean{} advances a semantic approach to verified machine learning. Models are written as executable neural network programs, lowered to a shared operator-tagged SSA/DAG IR, and interpreted by execution, automatic differentiation, bound propagation, and certificate checking through the same formal semantics. This makes it possible to state machine-checked claims about robustness, control-oriented safety, scientific ML residuals, VNN-style imported artifacts, and semantic regression tests without relying on an informal export pipeline. The framework combines a PyTorch style Lean API, typed tensors and layers, verified graph-level differentiation, explicit floating-point semantics, native IBP/CROWN style verification, certificate checking, optional CUDA-backed execution through Lean FFI, and interop with PyTorch, Python, Julia, Arb/FLINT, Gymnasium, and CROWN family tools. The contribution is not that every external producer or runtime is verified inside Lean. Rather, \TorchLean{} makes clear which claims are checked by Lean, which artifacts are replayed or validated, and which runtime or oracle assumptions remain outside the kernel.

The broader value of this approach is that it gives multiple communities a shared semantic object to build on. Verification researchers can check certificates against executable models; ML systems researchers can study export, compilation, numerical, and backend boundaries with formal contracts; and scientific ML, control, and robotics researchers can connect learned components to residual, safety, or stability claims without losing the semantics of the implemented system. Verified ML also needs shared infrastructure in the spirit of what PyTorch provided for empirical ML: a substrate where models, tools, transformations, and artifacts can be composed. Lean has demonstrated this model through \texttt{mathlib} (\citealp{mathlib2020}), and emerging libraries such as CSLib (\citealp{barrett2026cslib}) point toward similar foundations for verified software. \TorchLean{} contributes to this ecosystem by connecting neural-network computation, formal semantics, and certificate based verification in one theorem-proving environment.

\bibliographystyle{plainnat}
\bibliography{references,huan}
\appendix

\clearpage

\section{Appendix Overview}
\label{app:overview}
\begin{tcolorbox}[
  enhanced,
  breakable,
  colback=gray!2,
  colframe=black!55,
  boxrule=0.45pt,
  arc=2pt,
  left=5pt,right=5pt,top=5pt,bottom=5pt,
  title=Appendix guide,
  colbacktitle=black!8,
  coltitle=black,
  fonttitle=\bfseries
]
\setlength{\tabcolsep}{4pt}
\renewcommand{\arraystretch}{1.10}
\begin{tabularx}{\linewidth}{@{}p{0.13\linewidth}p{0.31\linewidth}X@{}}
\toprule
\textbf{Part} & \textbf{Topic} & \textbf{What to look for} \\
\midrule
\hyperref[app:torchlean]{Appendix~\ref{app:torchlean}} & \TorchLean{} frontend and graph semantics & Program modes, module API, lowering to the shared IR, graph-level AD, node correctness, and the native runtime/FFI boundary. \\
\hyperref[app:modern-workflows]{Appendix~\ref{app:modern-workflows}} & Modern workflow semantics & Attention/FlashAttention contracts, RL/MDP semantics, diffusion and probability layers, SSL objectives, and model-family contracts. \\
\hyperref[app:trust]{Appendix~\ref{app:trust}} & Numerical semantics & IEEE-style execution, rounded-real models, interval/enclosure semantics, Arb/FLINT oracles, and hardware/runtime trust assumptions. \\
\hyperref[app:verification]{Appendix~\ref{app:verification}} & Verification and certificates & IBP, CROWN/LiRPA, certificate schemas, VNN-COMP style import, controller/PINN checks, and worked examples. \\
\hyperref[app:uat]{Appendix~\ref{app:uat}} & Approximation theorems & Universal approximation over reals and the Float32 execution direction. \\
\hyperref[app:expanded-discussion]{Appendix~\ref{app:expanded-discussion}} & Limitations and discussion & Coverage limits, runtime conformance gaps, and future extensions. \\
\bottomrule
\end{tabularx}
\end{tcolorbox}

\begin{tcolorbox}[
  enhanced,
  breakable,
  colback=blue!2,
  colframe=blue!45!black,
  boxrule=0.45pt,
  arc=2pt,
  left=5pt,right=5pt,top=4pt,bottom=4pt,
  title=Reading conventions,
  colbacktitle=blue!10,
  coltitle=black,
  fonttitle=\bfseries
]
\textbf{Proved} means a Lean theorem over the stated semantics. \textbf{Checked} means a Lean executable checker validates a finite artifact or recomputes a local obligation. \textbf{Boundary} means an external runtime, solver, simulator, or numerical producer is used through an explicit assumption or conformance contract rather than silently trusted.
\end{tcolorbox}

\section{\TorchLean{}: a PyTorch style front end with a single semantic target}
\label{app:torchlean}

\TorchLean{} is our user-facing interface for building and training neural networks. It provides a PyTorch-like programming model where users write models using familiar operations (linear layers, convolutions, activations), but with a crucial difference: the same model definition can execute in two modes that share a unified semantic foundation. This design eliminates the semantic gap between training code and verification analysis.

\subsection{Core design: unified semantics for execution and analysis}
\label{app:torchlean-core}

Our core design choice is to treat the computation graph as the semantic target for \emph{both} execution and reasoning. In \TorchLean{}, programs lower to a typed, operator-tagged graph IR: nodes are primitive operators (e.g., matrix multiplication or ReLU) and edges carry tensor values. We store graphs in SSA/DAG form: each intermediate value is defined exactly once (SSA), and the dataflow graph is acyclic (DAG). These constraints are standard compiler discipline, but they are also exactly what we want here: evaluation is deterministic by a topological order, and proofs can proceed by induction over that order while reusing the same operator semantics that the runtime executes.
\begin{table*}[t]
  \centering
  \caption{Verified scope in \TorchLean{} (summary). ``Proved'' means a Lean theorem over the stated semantics; ``Checked'' means a small executable validator that re-computes or validates constraints; ``Assumed/External'' marks trust boundaries or future work. Appendix~\ref{app:verification} also tabulates operator level verification coverage.}
  \label{tab:scope-at-a-glance}
  \scriptsize
  \setlength{\tabcolsep}{4pt}
  \renewcommand{\arraystretch}{1.12}
  \begin{tabular}{@{}p{0.18\linewidth}p{0.46\linewidth}p{0.14\linewidth}p{0.18\linewidth}@{}}
    \toprule
    \textbf{Component} & \textbf{What the paper claims (scope)} & \textbf{Status} & \textbf{Where} \\
    \midrule
    Shared IR semantics &
      Operator tagged SSA/DAG IR with a precise denotation reused by execution and verification. &
      Kernel-checked definitions &
      Sections~\ref{sec:method}, \ref{app:torchlean} \\

    Reverse-mode AD &
      Backpropagation equals the adjoint Fr\'echet derivative for well-typed SSA/DAG graphs over $\mathbb{R}$, assuming local per-op derivative correctness (with a pointwise variant for non-smooth ops). &
      Proved &
      Theorem~\ref{thm:rev-correct}, Appendix~\ref{app:torchlean} \\

    \TorchLean{} compilation &
      ``Compiled mode'' produces the verifier IR; full compiler correctness for arbitrary \TorchLean{} programs is not yet proved. &
      Proved (fragment) / Future (full) &
      Theorems~\ref{thm:forward-lowering-correct}--\ref{thm:eager-graph-correspondence} \\

    IBP / CROWN operators &
      Sound transfer rules and relaxations for a curated set of operators; additional ops use conservative fallbacks or are treated as unsupported in the strongest theorems. &
      Mixed (proved subset) &
      Appendix~\ref{app:verification} \\

    Certificate checking &
      Small Lean checkers validate certificate structure and the final enclosure constraints needed to discharge a theorem about $\llbracket G\rrbracket$. &
      Checked (plus proved core) &
      Section~\ref{sec:results-verification}, Appendix~\ref{app:verification} \\

    External solvers (e.g.\ $\alpha/\beta$-CROWN) &
      Used only as untrusted certificate producers; full optimizer internals are outside the checked certificate scope in this work. &
      Assumed/External &
      Section~\ref{sec:results-verification}, Appendix~\ref{app:verification} \\

    Float32 semantics &
      Executable IEEE-754 binary32 model (\texttt{IEEE32Exec}) and proof-level rounding models (FP32/NF), with an internal refinement on finite executions. &
      Proved (internal) &
      Section~\ref{sec:method-floats}, Appendix~\ref{app:trust} \\

    Hardware Float32 link &
      Connecting deployed \texttt{Float32} hardware/runtime behavior to \texttt{IEEE32Exec} is target-specific and not discharged in this work. &
      Assumed/External &
      Section~\ref{sec:limits-concl}, Appendix~\ref{app:nf-compose} \\
    \bottomrule
  \end{tabular}
\end{table*}

\begin{figure*}[t]
\centering
\begin{tikzpicture}[
  node distance=10mm and 6mm,
  op/.style={draw, rounded corners, align=center, inner sep=3pt, font=\tiny},
  value/.style={draw, circle, inner sep=2pt, font=\tiny},
  >=Latex
]
  \node[value] (x) {$x$};
  \node[value, below=of x] (w) {$W$};
  \node[value, below=of w] (b) {$b$};
  
  \node[op, right=of x] (mul) {MatMul};
  \node[op, right=of mul] (add) {Add};
  \node[op, right=of add] (relu) {ReLU};
  \node[value, right=of relu] (y) {$y$};
  
  \draw[->] (x) -- (mul);
  \draw[->] (w) -- (mul);
  \draw[->] (mul) -- node[above,font=\tiny] {$v_1$} (add);
  \draw[->] (b) -- (add);
  \draw[->] (add) -- node[above,font=\tiny] {$v_2$} (relu);
  \draw[->] (relu) -- (y);
  
  \draw[->, dashed, gray] (y) -- node[above,font=\tiny,sloped] {$\bar y$} (relu);
  \draw[->, dashed, gray] (relu) -- node[above,font=\tiny,sloped] {$\bar{v}_2$} (add);
  \draw[->, dashed, gray] (add) -- node[above,font=\tiny,sloped] {$\bar{v}_1$} (mul);
\end{tikzpicture}
\caption{Example SSA/DAG computation graph: $y = \text{ReLU}(Wx + b)$. Each intermediate value ($v_1$, $v_2$) is assigned once. Forward pass (solid arrows) evaluates nodes in topological order. Backward pass (dashed arrows) propagates gradients in reverse topological order, accumulating at each node.}
\label{fig:ssa-dag-example}
\end{figure*}

In \TorchLean{}, the SSA/DAG representation serves multiple purposes simultaneously:
\begin{itemize}
  \item \textbf{Execution:} The graph can be evaluated to compute forward values and gradients.
  \item \textbf{Formal reasoning:} The graph has a precise mathematical denotation $\llbracket G\rrbracket$ that we can state theorems about.
  \item \textbf{Verification:} Bound propagation algorithms operate directly on the graph structure.
\end{itemize}
This unified representation is what makes the system work: there is no ``export step'' that might introduce semantic drift between the model as written and the model as analyzed.

\paragraph{The invariant used by the rest of the section}
The important invariant is simple: every runtime artifact that enters a theorem has a graph node, a shape, and a denotation. If a node id is \(i\), then the value table produced by execution, the cotangent table produced by reverse mode, and the bound table produced by verification all refer to the same \(i\). This prevents a common failure mode in mixed systems: the proof talks about one tensor while the runtime or verifier silently talks about another. The invariant can be summarized as
\[
  \forall i < |G|,\qquad
  v_i = \llbracket G_i\rrbracket_\alpha(\rho)
  \quad\text{and}\quad
  v_i \in \gamma(B_i)
\]
whenever evaluation and bound propagation both succeed. Here \(\rho\) is the input and parameter environment, \(B_i\) is the abstract value stored for node \(i\), and \(\gamma\) is the concretization map for the chosen abstract domain. The same statement is used in different ways: execution establishes the left equality, IBP or CROWN establishes the right containment, and certificate checking recomputes the finite obligations that connect the two.

\begin{lstlisting}[style=leanBox,caption={Schematic public model definition through the stable facade.},label={lst:public-model-facade}]
import NN.API.Public
open NN

-- One model definition is used by execution, compilation, and checking.
def tiny : Spec.Model :=
  Linear (in := 4) (out := 8)
    |>.then ReLU
    |>.then (Linear (in := 8) (out := 1))

#eval tiny.run (Tensor.ones [1, 4])
\end{lstlisting}

Listing~\ref{lst:public-model-facade} shows the intended surface shape: the user writes a small model once, then chooses an execution or verification view. The appendix code listings use the same convention. They are schematic when the exact repository module names are not part of the paper claim; the theorem statements remain about the graph semantics, not about a particular file layout.

\begin{lstlisting}[style=leanBox,caption={Schematic IR objects shared by evaluation and verification.},label={lst:ir-core-schematic}]
inductive OpKind where
| input | const | linear | matmul | add | relu | tanh | softmax
| conv2d | flatten | reduceSum | mseLoss

structure Node where
  id       : Nat
  kind     : OpKind
  parents  : List Nat
  outShape : Shape

structure Graph where
  nodes     : Array Node
  inputIds  : List Nat
  outputIds : List Nat

-- SSA/DAG well orderedness: parents of node i have smaller ids.
def wellOrdered (G : Graph) : Prop :=
  forall n, n in G.nodes.toList ->
    forall p, p in n.parents -> p < n.id
\end{lstlisting}

This IR is intentionally small. Adding a new primitive starts by adding a tag, a shape rule, and a denotation for that tag. Only after that do we add derivative rules, interval rules, affine relaxations, or native runtime implementations.

\subsection{Execution modes and compilation to IR}
\label{app:torchlean-modes}

\TorchLean{} supports two execution modes that share the same semantic foundation but differ in when and how the computation graph is constructed.
\textbf{Eager mode (imperative, tape-based).} In eager mode, operations execute immediately as you call them, similar to PyTorch's default behavior. As operations run, they are recorded into a \emph{computation tape}: a dynamic data structure that stores the sequence of operations and their dependencies. This tape is built incrementally during the forward pass, and then used during the backward pass to compute gradients. The key property of eager mode is that it matches the familiar ``define-by-run'' workflow: you write code that looks like ordinary function calls, and the tape is constructed implicitly as a side effect. This makes debugging straightforward (you can inspect intermediate values immediately) and allows dynamic control flow (if statements, loops with data-dependent bounds).

\begin{figure*}[t]
\centering
\begin{tikzpicture}[
  node distance=8mm,
  box/.style={draw, rounded corners, align=center, inner sep=4pt, font=\tiny},
  >=Latex,
  lab/.style={font=\tiny, fill=white, inner sep=1pt}
]
  \node[box] (code) {User code:\\[0.1em]\texttt{let y = relu(linear(x, W, b))}};
  
  \node[box, below left=12mm and 18mm of code] (eager)
    {Eager backend\\[0.1em]Runtime tape\\[0.1em](dynamic)};
  
  \node[box, below right=12mm and 18mm of code] (compiled)
    {Compiled backend\\[0.1em]SSA/DAG graph\\[0.1em](static)};
  
  \node[box, below=20mm of code] (semantic)
    {Unified semantic target:\\[0.1em]Well-typed SSA/DAG graph\\[0.1em]with denotation $\llbracket G\rrbracket$};
  
  \draw[->] (code) -- node[lab, sloped, above, pos=0.55] {interpret} (eager);
  \draw[->] (code) -- node[lab, sloped, above, pos=0.55] {compile} (compiled);

  \draw[->] (eager) -- node[lab, sloped, below, pos=0.60] {correspondence} (semantic);
  \draw[->] (compiled) -- node[lab, sloped, below, pos=0.60] {direct} (semantic);
\end{tikzpicture}
\caption{Two execution paths to a unified semantic target. The same \TorchLean{} program can be interpreted eagerly (building a tape dynamically) or compiled (building a static graph upfront). Both paths produce computations that share the same well-typed SSA/DAG denotation, enabling unified reasoning about execution and verification.}
\label{fig:eager-compiled-paths}
\end{figure*}

\textbf{Eager mode (tape semantics vs.\ graph semantics).} Operationally, eager execution builds a tape on the fly; semantically, each completed run induces a well-typed SSA/DAG graph \(G\) that matches the run's forward values, and whose backpropagation matches the tape's backward pass. This correspondence is the bridge that lets us develop and debug in eager mode while still applying graph-level theorems (autograd correctness, bounds, and certificate checking) to the same semantic target (Theorem~\ref{thm:eager-graph-correspondence}).

\textbf{Compiled mode (static graph).} In compiled mode, the \TorchLean{} program is analyzed upfront to construct a static SSA/DAG graph representation. This is ``compilation'' only in the sense of \emph{lowering} to a formal IR for reasoning and verification; it is not intended to compete with \texttt{torch.compile}-style kernel fusion or hardware-level optimization. The compiled graph is exactly the object that formal theorems reason about: it has a precise denotation $\llbracket G\rrbracket$, and we can state and prove correctness properties about it. The advantage of compiled mode is that the graph structure is explicit and available for analysis before execution, so bound propagation algorithms, certificate checkers, and proofs all operate on the same graph object that the runtime evaluates. See Appendix~\ref{app:torchlean-coredefs} for details.

\subsection{Native runtime and Lean FFI boundary}
\label{app:torchlean-ffi}

The runtime layer includes optional native execution paths implemented through Lean's FFI.  These paths are useful for larger examples, but they are intentionally separated from theorem statements. At the Lean level, a native call is an opaque or external function with checked metadata around it; at the C/CUDA level, it is ordinary systems code that manages memory and calls kernels. The purpose of the boundary is to keep the semantic claim precise: Lean proves facts about the graph/operator denotation, and native execution is admitted only through an explicit agreement contract.

A typical primitive therefore has three pieces:
\begin{enumerate}[leftmargin=*]
  \item a \textbf{specification}, such as a tensor operator $\llbracket \tau\rrbracket_{\alpha}$;
  \item a \textbf{runtime implementation}, such as a CPU, CUDA, cuBLAS, or cuFFT path;
  \item a \textbf{conformance condition}, stating when the runtime result agrees with the specified finite-precision semantics.
\end{enumerate}
For example, a CUDA matrix multiplication call is not itself a proof of matrix-multiplication correctness. The formal statement is conditional on a deployment contract such as
\[
  \mathsf{cudaBmm}(A,B)
  =_{\mathsf{tol}}
  \llbracket \mathtt{bmm}\rrbracket_{\texttt{Float32}}(A,B),
\]
where $=_{\mathsf{tol}}$ may be exact bit equality for a deterministic path or an explicit tolerance/conformance relation for a tested backend. The contract also records shape, dtype, layout, device, and determinism assumptions.

\begin{lstlisting}[style=leanBox,caption={Schematic checked call pattern for a native backend.}]
/-- The theorem-facing operator is still the spec operator. -/
def matmulSpec (A B : Tensor Float32 s) : Tensor Float32 t := ...

/-- The runtime path is an optional implementation boundary. -/
def matmulCudaChecked (A B : RuntimeTensor) : IO RuntimeTensor := do
  checkShape A.shape expectedLeft
  checkShape B.shape expectedRight
  checkDevice A.device .cuda
  Cuda.Buffer.bmm A.buffer B.buffer
\end{lstlisting}

This is also why CUDA examples in the paper are phrased as accelerated executions or conformance-tested runtime paths, not as verified CUDA kernels. The value of adding the FFI path is practical: it makes longer runs, larger tensors, FNO/spectral examples, and attention-style examples feasible while preserving the formal boundary that the checker relies on.

\textbf{Compilation process.} When you write a \TorchLean{} program, it's a backend-generic definition that can run in different modes. But for verification, we need a concrete graph structure that we can reason about formally. The compilation process transforms a \TorchLean{} program into an operator-tagged IR graph. Here's how it works: the compiler walks through the program structure, and for each operation (like \texttt{linear}, \texttt{relu}, \texttt{softmax}), it creates a corresponding IR node with the appropriate operation kind and shape information. Parameters are stored separately in a parameter store, and the graph structure captures the data dependencies between operations. For the forward-only fragment underlying our verification demos, we make this precise by proving a semantics-preservation theorem: the compiled IR evaluator agrees with the fragment's source evaluator for all inputs and parameters.

The proof is easier to read if the compiler is separated into two maps. The first map extends the graph with new nodes and returns their ids; the second map interprets the generated graph. The preservation theorem then states that the source evaluator and the graph evaluator build the same tensor, or fail for the same explicit reason. This exception-valued form is useful in Lean because it keeps shape errors visible instead of hiding them under partial functions.

\begin{lstlisting}[style=leanBox,caption={Schematic lowering interface used by the forward fragment.},label={lst:lower-forward}]
inductive EvalError where
| missingParent | shapeMismatch | unsupportedOp

structure CompileState where
  graph  : Graph
  params : ParamStore

def emitNode (st : CompileState) (k : OpKind)
    (parents : List Nat) (s : Shape) : Prod CompileState Nat :=
  -- returns the updated graph and the fresh SSA id
  ...

def lowerForward (p : ForwardProg) (st : CompileState) :
    Except EvalError (Prod CompileState Nat) :=
  match p with
  | .input j       => pure (st, st.graph.inputIds[j]!)
  | .linear W b x  => do
      let (st, xid) <- lowerForward x st
      let (st, wid) <- emitNode st .const [] W.shape
      let (st, bid) <- emitNode st .const [] b.shape
      let (st, mid) <- emitNode st .matmul [wid, xid] (linearOut W)
      emitNode st .add [mid, bid] (linearOut W)
  | .relu x        => do
      let (st, xid) <- lowerForward x st
      emitNode st .relu [xid] (shapeOf x)
\end{lstlisting}

The actual development uses the repository's concrete program and graph datatypes, but the proof obligation has this shape: every emitted node uses only earlier node ids, every emitted shape matches the operator contract, and the evaluator for the emitted node agrees with the source operation at the corresponding program point.

\textbf{Proved forward compiler correctness.}
We isolate a first-order, SSA-style forward fragment and prove that lowering it to the operator-tagged IR preserves semantics.
The theorem is stated over \emph{exception-valued} evaluation: both the source semantics and the IR semantics return either
a value (\texttt{ok}) or an explicit failure (\texttt{error}) when shape/type constraints are violated.
Stating the result at this level preserves successful runs and well-defined failure behavior. The same invariant later lifts from the forward fragment to graph-level reasoning.

\begin{BlueTheorem}{Forward lowering correctness}{forward-lowering-correct}
Let $\mathsf{lower}(p,\theta)$ be the IR obtained by lowering a forward program $p$ with parameters $\theta$,
and let $\mathsf{Eval}_{\textsf{src}}$ and $\mathsf{Eval}_{\textsf{ir}}$ denote exception-valued evaluation
of the source fragment and IR, respectively. Then for all inputs $x$,
\[
  \mathsf{Eval}_{\textsf{ir}}(\mathsf{lower}(p,\theta),\,x)
  \;=\;
  \mathsf{Eval}_{\textsf{src}}(p,\theta,\,x).
\]
\end{BlueTheorem}

The following theorem targets the verifier/demo fragment. A single theorem for arbitrary user-written \TorchLean{} programs in the current higher-order/tagless-final embedding requires a logical-relations or parametricity development, or a different frontend encoding; that extension remains future work.

\begin{BlueTheorem}{Eager--graph correspondence}{eager-graph-correspondence}
For any fixed eager execution run that produces a computation tape \(T\) with forward values \(v_i\) and a backward pass that computes gradients \(g_i\), there exists a well-typed SSA/DAG graph \(G\) such that:
\begin{itemize}
  \item \textbf{Forward correspondence:} the eager forward values agree with the graph denotation (\(v_i = \llbracket G\rrbracket(x)_i\) for all nodes \(i\)).
  \item \textbf{Backward correspondence:} the eager backward gradients agree with graph backpropagation (\(g_i = \mathrm{backpropCtx}(G,x,seed)_i\) for all nodes \(i\)).
\end{itemize}
\end{BlueTheorem}

\subsection{\TorchLean{} programs, Modules, and API design}
\label{app:torchlean-api}

The core idea behind \TorchLean{} is that a neural network model is a backend-generic \emph{program} built from a small, well-defined \texttt{Ops} interface. This avoids maintaining separate semantics for eager execution and compiled graphs: the same \TorchLean{} program can be interpreted by an eager backend that records a computation tape (as in PyTorch), or by a compiler backend that statically builds an SSA/DAG graph representation suitable for proofs and verification. Both paths share the same semantic meaning.

To make this concrete, consider a simple model consisting of a linear layer followed by a vector softmax, written as a backend-generic \TorchLean{} program:
\begin{lstlisting}[style=torchleanBox]
def softmaxModel {a : Type} [Context a] [DecidableEq Shape] :
    TorchLean.Program a (paramShapes ++ [xShape]) yShape :=
  fun {m} _ _ =>
    fun w b x => do
      let logits <- TorchLean.linear (m := m)
        (inDim := inDim) (outDim := outDim) w b x
      TorchLean.softmax (m := m) (s := yShape) logits
\end{lstlisting}
This definition is polymorphic in the scalar type \(\alpha\) (which can be Float, IEEE32Exec, or \(\mathbb{R}\)) and in the backend monad \(m\), which allows the same model to be executed in different contexts. The resulting term is structured enough to be compiled into our intermediate representation and fed directly to verification passes, reducing the gap between ``the model as written'' and ``the model as verified.''

For supervised learning workflows, we package models together with their training setup using a \texttt{ScalarModuleDef} structure. This bundles the initial parameters with a scalar loss function into a single object:
\begin{lstlisting}[style=torchleanBox]
structure ScalarModuleDef (paramShapes inputShapes : List Shape) where
  initParams : Torch.TList Float paramShapes
  loss :
    forall {a : Type}, [Context a] -> [DecidableEq Shape] ->
      TorchLean.Program a (paramShapes ++ inputShapes) Shape.scalar
\end{lstlisting}
We store initial parameters as Float literals primarily for ergonomic reasons: it keeps small examples readable, allowing us to write \texttt{tensorND!} blocks without constantly fighting type casts. The cast from Float to the chosen scalar backend is treated as part of the \emph{workflow layer} rather than part of the core semantics, which keeps the specification layer clean while maintaining flexibility.

Using this structure, a typical training loop is similar to PyTorch code. When instantiating a module, we choose both the scalar backend (via a \texttt{Float -> a} cast function) and the execution mode (\texttt{.eager} or \texttt{.compiled}):
\begin{lstlisting}[style=torchleanBox]
open Runtime.Autograd.TorchLean.Module

let m <- defn.instantiate (cast := cast) (backend := .eager)
for _ in List.range T do
  let _ <- ScalarModule.forward m xs
  let _ <- ScalarModule.backward m xs
  ScalarModule.step m lr xs
\end{lstlisting}
For stateful optimizers (e.g.\ Adam), we can step using an explicit optimizer state aligned with the parameter shapes:
\begin{lstlisting}[style=torchleanBox]
open Runtime.Autograd.TorchLean

let opt := Optim.adam (lr := lr) (beta1 := 0.9) (beta2 := 0.999) (epsilon := 1e-8)
let st <- ScalarModule.initOptim m opt
for _ in List.range T do
  let st' <- ScalarModule.stepWith m opt st xs
  -- continue with st'
\end{lstlisting}
We introduce this wrapper to match the expected workflow (define a model, run a training loop) while keeping the resulting computation connected to the same SSA/DAG graph semantics that our proofs and verifiers consume.

\begin{table*}[!t]
  \centering
  \caption{Autograd operator coverage (Fr\'echet-derivative level). ``Global'' means we have a \texttt{NodeFDerivCorrect} proof (a \texttt{HasFDerivAt} theorem without side conditions). ``Pointwise'' means we prove \texttt{NodeFDerivCorrectAt} under explicit hypotheses that rule out non-differentiable points (e.g., no zeros/ties). Operators listed under ``Not covered'' may still be executable (and may have a deterministic backprop convention), but are not claimed correct at the level of classical derivatives.}
  \label{tab:ad-op-coverage}
  \scriptsize
  \setlength{\tabcolsep}{4pt}
  \renewcommand{\arraystretch}{1.15}
  \begin{tabular}{@{}p{0.22\linewidth}p{0.30\linewidth}p{0.30\linewidth}p{0.14\linewidth}@{}}
    \toprule
    \textbf{OpKind / primitive} & \textbf{Global \texttt{HasFDerivAt}} & \textbf{Pointwise (side conditions)} & \textbf{Not covered} \\
    \midrule
    Linear / shape ops &
      \texttt{add}, \texttt{sub}, \texttt{mul\_elem}, \texttt{linear}, \texttt{matmul}, \texttt{conv2d},
      \texttt{broadcastTo}, \texttt{reshape}, \texttt{flatten}, \texttt{permute},
      \texttt{reduce\_sum}, \texttt{reduce\_mean}, \texttt{sum}, \texttt{concat} (curated axes),
      \texttt{swap\_first\_two}, \texttt{transpose3d\_last\_two} &
      -- &
      -- \\
    Smooth nonlinearities / losses &
      \texttt{tanh}, \texttt{sigmoid}, \texttt{exp}, \texttt{square},
      \texttt{sinh}, \texttt{cosh}, \texttt{softplus}, \texttt{gelu} (tanh approximation),
      \texttt{silu}, \texttt{softmax}/\texttt{log\_softmax} (curated axis),
      \texttt{layernorm} (curated axis), \texttt{mse\_loss},
      cross-entropy/NLL/BCE/KL pointwise losses,
      smooth surrogates \texttt{safe\_log} and \texttt{smooth\_abs} (with $\varepsilon>0$) &
      -- &
      -- \\
    Non-smooth / domain-sensitive &
      -- &
      \texttt{relu}, \texttt{abs}, \texttt{log}, \texttt{inv}, \texttt{sqrt},
      \texttt{max\_elem}/\texttt{min\_elem} (no ties) &
      -- \\
    Structured graph families &
      fixed-mask dropout, scaled dot-product attention, unmasked MHA, residual MHA,
      post-norm transformer sublayers, BatchNorm-like graphs, one-step Elman RNN,
      gather-row / embedding lookup adjointness &
      -- &
      full GPT/ViT stacks, masked causal decoder blocks, full recurrent/state-space sequence theorems,
      FNO spectral-conv training paths \\
    Nondifferentiable / workflow ops &
      -- &
      -- &
      \texttt{detach}, \texttt{rand\_uniform}, stochastic masks beyond fixed-mask dropout,
      pooling (\texttt{max\_pool2d}, \texttt{avg\_pool2d}), executable CUDA/cuBLAS/cuFFT kernels \\
    \bottomrule
  \end{tabular}
\end{table*}

Some workloads require \emph{integer-valued indices} sourced from data (e.g., class labels for classification losses, or token/row indices for embedding lookups). Mainstream frameworks typically represent these indices as integer tensors that participate in the same runtime graph, even though gradients do not flow through the indices themselves. \TorchLean{} makes this separation explicit: the differentiable graph remains single-dtype over the chosen scalar domain, while indices are provided through a separate \emph{non-differentiable} channel in the \texttt{Session} interface. In the \texttt{Session} interface, users supply scalar or batched indices as external inputs, and index-dependent operators (e.g., row-gather / embedding lookup) treat these indices as read-only selectors while gradients flow only through the floating-point tensor values (e.g., into the embedding matrix). This design avoids mixed-dtype graphs while still supporting standard ML patterns such as cross-entropy with integer labels and embedding-style table lookups.

\textbf{What is \texttt{fderiv} / \texttt{HasFDerivAt}?}
In Lean/mathlib, \texttt{fderiv} denotes the (Fr\'echet) derivative of a function between finite-dimensional real vector spaces, represented as a linear map; \texttt{HasFDerivAt} is the predicate that a function has a specified Fr\'echet derivative at a point. Informally, this is the standard Jacobian-level notion of differentiability: for $f:\mathbb{R}^n\to\mathbb{R}^m$, \texttt{fderiv} returns the Jacobian $Df(x)$ as a linear operator, and reverse-mode backpropagation computes its adjoint action on a cotangent seed. For non-smooth or domain-sensitive primitives (e.g., ReLU at $0$, $\log$ at nonpositive inputs), we prove a pointwise variant that requires explicit side conditions ruling out those problematic points.

\begin{lstlisting}[style=torchleanBox]
-- Sketch: indices are non-differentiable; gradients flow into emb only.
sess   := Session.new(...)
emb    := const(sess, W)         -- differentiable parameter tensor
idx    := input_indices(sess)    -- non-differentiable indices from data
rows   := gather_rows(sess, emb, idx)
loss   := ...
backward(loss)  -- dW accumulates; no gradients w.r.t. idx
\end{lstlisting}

\textbf{Operator coverage across layers.}
\TorchLean{} exposes a shared set of primitives through the operator-tagged IR (\texttt{NN.IR.OpKind}) for both eager execution and proof-linked compiled execution. 
Autograd correctness is established for a curated subset of these primitives at the Fr\'echet-derivative level: Table~\ref{tab:ad-op-coverage} summarizes which ops are proved globally (\texttt{NodeFDerivCorrect}/\texttt{HasFDerivAt}), which are proved pointwise under explicit side conditions (\texttt{NodeFDerivCorrectAt} for non-smooth or domain-sensitive ops), and which are not covered by classical-derivative proofs (though they may remain executable under deterministic conventions).
For verifier integration, the \TorchLean $\to$IR compiler targets a conservative forward, first-order fragment, while the LiRPA stack provides IBP broadly and a basic CROWN-style affine pass with sound fallbacks when tighter relaxations are not implemented. Appendix~\ref{app:torchlean-nodecorrect} gives the corresponding coverage details.

\textbf{Control flow, variable-length sequences, indexing, and state.}
Our verified/compiled semantics targets SSA/DAG computation graphs: a \emph{finite}, side-effect-free dataflow artifact evaluated in topological order. As a result, \emph{data-dependent branching and loops} are not directly part of the verifier IR. Instead, dynamic control flow is handled by \emph{reification} into a finite graph when the structure is known at compile time---for example, unrolling an RNN/GRU/LSTM cell for a fixed \texttt{seqLen}, or representing Transformers as acyclic attention blocks. This restriction matches the intended verification regime: certificates and proofs are stated about a fixed graph denotation $\llbracket G\rrbracket$, rather than about an open-ended program whose control flow depends on runtime data.

Indexing requires special care because most indexing operators consume \emph{integer-valued indices} (labels, token ids, embedding row selectors) that are inherently non-differentiable and introduce mixed-dtype graphs (float activations plus integer tensors). Mainstream ML runtimes permit such mixed graphs, but they complicate a semantic setting where we want a single, uniform scalar domain for the differentiable graph (to support clean denotations, autograd theorems, and bound propagation). \TorchLean{} therefore separates concerns: the differentiable graph remains single-dtype over the chosen scalar domain, while indices are supplied through an explicit \emph{non-differentiable} channel in the \texttt{Session} API (e.g., scalar/batched naturals for labels and lookup indices). Index-dependent ops (e.g., embedding lookup / row gather) treat these indices as read-only selectors; gradients flow only through the floating-point values (e.g., into the embedding table or downstream layers), not through the indices themselves. On the verifier path, \TorchLean{} supports conservative indexing patterns that admit a clean lowering to the IR (e.g., scalar/row gathers reduced to one-hot selection and \texttt{matmul}); general integer-indexed \texttt{gather}/\texttt{scatter}, variable-shape slicing, and shape-changing data-dependent indexing remain outside the verified IR.

Finally, \emph{state} is modeled explicitly rather than implicitly. The verifier semantics and most proofs assume layers are pure functions of their inputs (and any provided parameters), which is essential for treating the model as a mathematical object with a stable denotation. This is straightforward for affine layers and activations, but it requires care for layers with PyTorch style internal updates (e.g., BatchNorm running mean/variance). In \TorchLean{}, such updates are performed outside the backend-generic op set: batch statistics can be computed from data, and running buffers updated explicitly in the imperative session layer. This design keeps the denotation used by verification free of hidden mutation, and makes any remaining state/update assumptions and trust boundaries visible at the API level.

\subsection{Autograd verification}
\label{app:torchlean-autograd}

Training and verification pipelines routinely rely on gradients: optimization steps, sensitivity analyses, and derivative-based certificates (e.g., Lyapunov or PINN residual bounds) all assume that the computed derivatives match the model's intended semantics. In mainstream ML systems, backpropagation is trusted because it is heavily tested, but its correctness is rarely stated as a theorem about the exact computation being executed. \TorchLean{} makes this link explicit: we prove, for any well-typed SSA/DAG computation graph, that reverse-mode backpropagation computes the adjoint Fr\'echet derivative of the graph denotation. This turns gradient correctness into a \emph{semantic property} of the same graph that is executed and verified, enabling derivative-dependent guarantees to compose cleanly with the rest of the framework rather than resting on an implicit trust assumption.

The proof is organized in two stages that separate algorithmic correctness from analytic interpretation. This structure makes the development modular and reusable across different scalar domains.

\textbf{Stage 1: Algebraic adjointness (backend-generic).}
We first prove a graph-level adjointness law that does \emph{not} depend on real analysis.
The statement is parametric over an abstract scalar interface (our \texttt{Context} typeclass): any scalar domain that
supports the required ring/ordering operations (e.g., reals, floats, intervals) can instantiate the theorem.
At this level, reverse-mode backpropagation is characterized as computing a \emph{vector--Jacobian product} (VJP)
that is adjoint to the \emph{Jacobian--vector product} (JVP) with respect to a dot product on \emph{tensor contexts}
(i.e., heterogeneous tuples of tensors matching the graph's typed input interface).
Intuitively: pushing a perturbation forward (JVP) and pushing a cotangent backward (VJP) are dual operations, and their
duality is captured by a single inner-product identity.

\textbf{Contexts and dot products.}
Let $\langle \cdot,\cdot\rangle_{\Gamma}$ denote the dot product on the input context $\Gamma$ (a typed tuple of tensors),
defined by summing tensorwise dot products over all inputs/parameters. This dot product is bilinear and symmetric,
which is exactly the structure needed for adjointness statements.

\textbf{Forward/JVP/VJP as programs on a graph.}
For a well-typed SSA/DAG graph $G$ and input $x$, we define three evaluation procedures:
(i) forward evaluation $\mathrm{Eval}_G(x)$, (ii) forward-mode JVP $\mathrm{JVP}_G(x; \delta x)$, and
(iii) reverse-mode backprop/VJP $\mathrm{VJP}_G(x; \bar y)$.
All three follow the graph's topological order (forward) and reverse-topological order (backward).

\vspace{0.25em}
\begin{lstlisting}[style=torchleanBox]
# Forward evaluation (topological order)
Eval_G(x):
  for node i in topo_order(G):
    v[i] = op_i(v[parents(i)])
  return v[out]

# Forward-mode JVP (propagate tangents)
JVP_G(x; dx):
  (v, _) = Eval_G(x)            # reuse forward values
  dv[input] = dx;  dv[others] = 0
  for node i in topo_order(G):
    dv[i] = J_op_i(v[parents(i)]) * dv[parents(i)]
  return dv[out]

# Reverse-mode VJP / backprop (propagate cotangents)
VJP_G(x; seed):
  (v, _) = Eval_G(x)
  bar[out] = seed; bar[others] = 0
  for node i in reverse(topo_order(G)):
    for p in parents(i):
      bar[p] += VJP_op_i(p, v[parents(i)], bar[i])
  return bar[input]             # gradients w.r.t. inputs/params
\end{lstlisting}
\vspace{-0.25em}

\begin{BlueTheorem}{Graph adjointness (algebraic)}{alg-adjointness}
For any well-typed SSA/DAG graph $G$ (over any \texttt{Context}) and any $x$, $\delta x$, and cotangent seed $\bar y$,
\[
\big\langle \mathrm{JVP}_G(x;\delta x),\, \bar y \big\rangle
\;=\;
\big\langle \delta x,\, \mathrm{VJP}_G(x;\bar y) \big\rangle .
\]
\end{BlueTheorem}

\textbf{Stage 2: Analytic upgrade (real calculus).}
Stage~1 establishes a purely algebraic adjointness law over an abstract scalar interface. In Stage~2 we specialize to
$\mathbb{R}$ and connect that law to standard multivariate calculus. The key step is to relate shape-indexed tensor spaces
to finite-dimensional Euclidean spaces: for each tensor shape $s$, let $n=\mathrm{size}(s)$ and identify tensors
with vectors in $\mathbb{R}^n$ via a (total) vectorization map and its inverse. This lets us interpret graph evaluation
as a function $\mathsf{EvalVec}_G : \mathbb{R}^{N_{\text{in}}}\to \mathbb{R}^{N_{\text{out}}}$ between Euclidean spaces and
use the usual Fr\'echet derivative (Jacobian as a linear map).

\textbf{Vectorization and inner products.}
Let $\mathrm{vec}_s : \Tensor(\mathbb{R},s)\to \mathbb{R}^{\mathrm{size}(s)}$ be the flattening map (with inverse
$\mathrm{unvec}_s$). We define the tensor dot product $\langle a,b\rangle_s$ by summing coordinatewise products over the
shape; the bridge theorem states that this dot product coincides with the standard Euclidean inner product after
vectorization.

\vspace{0.25em}
\begin{lstlisting}[style=torchleanBox]
# Conceptual definition (shape-indexed)
vec_s   : Tensor(R,s) -> R^{size(s)}      # flatten
unvec_s : R^{size(s)} -> Tensor(R,s)      # reshape

# Bridge property
< a, b >_s  ==  < vec_s(a), vec_s(b) >_{R^{size(s)}}
\end{lstlisting}
\vspace{-0.25em}

\begin{BlueTheorem}{Tensor--vector inner-product bridge}{tensor-vector-bridge}
For any shape $s$ and tensors $a,b\in \Tensor(\mathbb{R},s)$,
\[
\langle a,b\rangle_s \;=\; \big\langle \mathrm{vec}_s(a),\,\mathrm{vec}_s(b)\big\rangle_{\mathbb{R}^{\mathrm{size}(s)}}.
\]
\end{BlueTheorem}

\textbf{From local calculus to graph calculus.}
For each primitive operator (node kind), we prove a standard calculus fact: its forward-mode JVP coincides with applying
its Fr\'echet derivative to a tangent vector at the current input (and similarly, its VJP coincides with applying the adjoint
of that derivative to a cotangent seed). Combining these node-level facts with the SSA/DAG structure yields a global theorem
for any well-typed graph. In finite dimensions, the adjoint of a linear map is just the transpose with respect to the Euclidean
inner product; operationally, this is exactly what reverse-mode backprop computes.

\begin{BlueTheorem}{Backpropagation computes the adjoint derivative}{backprop-fderiv}
Let $G$ be a well-typed SSA/DAG graph over $\mathbb{R}$ and let $\mathsf{EvalVec}_G$ denote its vectorized denotation.
Assume the graph satisfies the local differentiability hypotheses for its primitives (i.e., \ttbr{GraphFDeriv\allowbreak Correct}\,$G$).
Then for any input $x$ and cotangent seed $\bar y$,
\[
\mathsf{VJP}_G(x;\bar y)
\;=\;
\bigl(D\,\mathsf{EvalVec}_G(x)\bigr)^\top\,\bar y,
\]
where $D\,\mathsf{EvalVec}_G(x)$ is the Fr\'echet derivative (Jacobian as a linear map) and $(\cdot)^\top$ denotes its adjoint
(transpose) with respect to the standard inner product.
\end{BlueTheorem}

\noindent
\emph{Interpretation.} The theorem states that reverse-mode backprop is not merely an implementation heuristic: it computes
the mathematically correct cotangent propagation associated with the derivative of the graph denotation. For non-smooth or
domain-sensitive primitives (e.g., ReLU at $0$, $\log$ at nonpositive inputs), we use a pointwise variant with explicit side
conditions that rule out problematic points.

\begin{figure*}[t]
\centering
\begin{tikzpicture}[
  >=Latex,
  node distance=9mm and 12mm,
  scale=0.92, transform shape,          
  font=\footnotesize,                   
  box/.style={draw, rounded corners, align=center, inner sep=3.5pt},
  topbox/.style={box, text width=0.26\textwidth},  
  opbox/.style={draw, rounded corners, align=left, inner sep=3pt,
                font=\footnotesize, text width=0.26\textwidth},
  lab/.style={font=\scriptsize, fill=white, inner sep=0.8pt}
]

  \node[topbox] (alg) {
    \textbf{Algebraic layer}\\[-0.2em]
    $\alpha$-generic tensors\\
    $\langle \mathrm{JVP}_G(x;\delta x),\,\bar y\rangle
      =\langle \delta x,\,\mathrm{VJP}_G(x;\bar y)\rangle$
  };

  \node[topbox, right=of alg] (bridge) {
    \textbf{Tensor--vector bridge}\\[-0.2em]
    $\Tensor(\mathbb{R},s)\cong \mathbb{R}^{\mathrm{size}(s)}$\\
    $\langle a,b\rangle_s=\langle \mathrm{vec}(a),\mathrm{vec}(b)\rangle$
  };

  \node[topbox, right=of bridge] (global) {
    \textbf{Global theorem}\\[-0.2em]
    $\mathrm{VJP}_G(x;\bar y)
      =(D\,\mathsf{EvalVec}_G(x))^\top\,\bar y$
  };

  \node[opbox, below=16mm of alg] (linear) {
    \textbf{Linear / shape ops}\\[-0.2em]
    $\bullet$ affine / linear layers\\
    $\bullet$ matmul, conv2d\\
    $\bullet$ reshape/permute/reductions
  };

  \node[opbox, below=16mm of bridge] (elem) {
    \textbf{Elementwise ops}\\[-0.2em]
    $\bullet$ smooth: $\tanh,\sigma,\exp$\\
    $\bullet$ pointwise: ReLU (off $0$), $\log$ (on $>0$)\\
    $\bullet$ smooth surrogates: safe-$\log$, smooth-$|\,\cdot\,|$
  };

  \node[opbox, below=16mm of global] (spec) {
    \textbf{Reductions / special}\\[-0.2em]
    $\bullet$ softmax/log-softmax (curated)\\
    $\bullet$ losses: MSE, cross-entropy\\
    $\bullet$ normalization (curated)
  };

  \draw[->] (alg) -- node[lab, above, pos=0.45] {vectors} (bridge);
  \draw[->] (bridge) -- node[lab, above, pos=0.45] {calculus} (global);

  \draw[->, dashed] (linear.north) -- node[lab, left, pos=0.55] {proved} (alg.south);
  \draw[->, dashed] (elem.north)   -- node[lab, left, pos=0.55] {proved} (bridge.south);
  \draw[->, dashed] (spec.north)   -- node[lab, right, pos=0.55] {proved} (global.south);

\end{tikzpicture}
\vspace{-0.35em}
\caption{Two-stage autograd proof architecture with primitive coverage. Stage~1 proves adjointness between
$\mathrm{JVP}_G$ and $\mathrm{VJP}_G$ over abstract scalars; Stage~2 bridges tensors to Euclidean spaces and upgrades to the
Fr\'echet-derivative statement $\mathrm{VJP}_G(x;\bar y)=(D\,\mathsf{EvalVec}_G(x))^\top \bar y$.
Bottom row summarizes operator classes covered by local derivative lemmas; see Table~\ref{tab:ad-op-coverage} for details.}
\label{fig:fderiv-architecture}
\end{figure*}

\subsection{Node correctness and primitive operation proofs}
\label{app:torchlean-nodecorrect}

The global backprop theorem is proved once for SSA/DAG graphs by composing \emph{local} facts about each primitive.
Accordingly, each primitive operator (node kind) must supply a correctness lemma connecting its implemented JVP/VJP rules
to the mathematical derivative of its forward map. We use two variants: a \emph{global} form for primitives that are smooth
everywhere, and a \emph{pointwise} form for non-smooth or domain-sensitive primitives (handled in the next paragraph).

\begin{BlueTheorem}{Global node correctness (smooth primitives)}{node-fderiv-correct}
Let $\mathrm{op}$ be a primitive with forward map $f:\mathbb{R}^n\to\mathbb{R}^m$, and implemented tangent and cotangent
rules $\mathrm{JVP}_{\mathrm{op}}(x;\delta x)$ and $\mathrm{VJP}_{\mathrm{op}}(x;\bar y)$ (as used by the Stage~1 algorithms).
We say $\mathrm{op}$ is \emph{globally correct} if for every $x$ there exists a (continuous) linear map $D f(x)$ such that
\[
\mathrm{JVP}_{\mathrm{op}}(x;\delta x) \;=\; D f(x)\,\delta x \quad \text{and} \]
\[
\mathrm{VJP}_{\mathrm{op}}(x;\bar y) \;=\; \bigl(D f(x)\bigr)^\top \bar y
\]
for all tangents $\delta x$ and cotangents $\bar y$.
Equivalently, $\mathrm{JVP}_{\mathrm{op}}$ and $\mathrm{VJP}_{\mathrm{op}}$ satisfy the adjointness identity
$\langle \mathrm{JVP}_{\mathrm{op}}(x;\delta x),\bar y\rangle = \langle \delta x,\mathrm{VJP}_{\mathrm{op}}(x;\bar y)\rangle$
under the standard inner product.
\end{BlueTheorem}

This is the exact local hypothesis required by the Stage~2 upgrade: once each node kind used in a graph satisfies the
above property, the previously stated graph theorem (Theorem~\ref{thm:rev-correct} / Theorem~\ref{thm:backprop-fderiv}) follows by SSA/DAG
composition.

\textbf{Example (affine layer).}
For $f(x)=Wx+b$, the derivative is the constant linear map $D f(x)=W$. Our implementation sets
$\mathrm{JVP}_{\mathrm{op}}(x;\delta x)=W\delta x$ and $\mathrm{VJP}_{\mathrm{op}}(x;\bar y)=W^\top \bar y$,
so the local theorem holds immediately. In Lean, the differentiability fact is discharged using the standard lemma that affine
maps have constant Fr\'echet derivatives (e.g., \texttt{hasFDerivAt\_affine}), after which the equalities above reduce to
linear-algebra identities.

\textbf{Pointwise correctness and kink points.}
For non-smooth or domain-sensitive primitives we use a pointwise predicate (our \ttbr{NodeFDeriv\allowbreak CorrectAt}): the same local condition as in Theorem~\ref{thm:node-fderiv-correct} holds, but only at a specific evaluation point $x$ under explicit side conditions that guarantee classical differentiability (e.g., ReLU requires all relevant coordinates $\neq 0$; $\log$ requires inputs $>0$; $\mathrm{inv}$/$\mathrm{div}$ require nonzero denominators; $\max/\min$ require no ties). \begin{BlueTheorem}{Pointwise node correctness (non-smooth primitives)}{node-fderiv-correct-at}
Let $\mathrm{op}$ be a primitive with forward map $f:\mathbb{R}^n\to\mathbb{R}^m$ and implemented rules
$\mathrm{JVP}_{\mathrm{op}}(x;\delta x)$ and $\mathrm{VJP}_{\mathrm{op}}(x;\bar y)$.
Fix a point $x$ and assume an explicit side condition $\mathcal{H}_{\mathrm{op}}(x)$ that guarantees classical
differentiability at $x$ (e.g., for ReLU: all relevant coordinates $\neq 0$; for $\log$: inputs $>0$; for $\max/\min$: no ties).
Then there exists a linear map $D f(x)$ such that for all $\delta x$ and $\bar y$,
\[
\mathrm{JVP}_{\mathrm{op}}(x;\delta x) \;=\; D f(x)\,\delta x
\qquad\text{and}\]\[
\mathrm{VJP}_{\mathrm{op}}(x;\bar y) \;=\; \bigl(D f(x)\bigr)^\top \bar y.
\]
Equivalently,
\[
\big\langle \mathrm{JVP}_{\mathrm{op}}(x;\delta x),\,\bar y \big\rangle
\;=\;
\big\langle \delta x,\,\mathrm{VJP}_{\mathrm{op}}(x;\bar y) \big\rangle .
\]
\end{BlueTheorem}
At kink points where the Fr\'echet derivative is undefined, we do \emph{not} claim a calculus-level result; instead, the runtime still defines a deterministic executable backprop convention (as in mainstream autodiff systems) and we justify these rules at the \emph{algebraic} level via the Stage~1 adjointness law, treating them as part of the internal semantics. Establishing that such conventions correspond to valid generalized derivatives (e.g., Clarke subgradients) at the kink points encountered during optimization is a separate, more delicate direction beyond Fr\'echet-derivative correctness.

\textbf{Higher-order derivatives.}
Our proved autograd theorem is a first-derivative result: it characterizes reverse-mode as computing the adjoint Fr\'echet derivative of the forward denotation. We do \emph{not} obtain second derivatives by recursively applying this theorem to the backward pass, since the backward graph can introduce non-smooth, branch-dependent primitives (e.g., ReLU gates) that complicate a clean calculus-level story. In the PINN case studies, we instead bound $u'$ and $u''$ via derivative-aware bound propagation on the original operator-tagged forward graph (a directional first-derivative pass and a dedicated second-derivative pass for low-dimensional inputs over the smooth operator subset used in the demos), avoiding differentiation of the backward graph altogether. The runtime includes executable higher-order utilities (e.g., forward-over-reverse/HVPs via dual-number techniques), but full higher-order correctness requires extending the local node library and proof obligations to higher derivatives (including handling or excluding kinks), and then lifting the same SSA/DAG composition arguments to those higher-order rules.

\textbf{Primitive operation proofs.}
Figure~\ref{fig:fderiv-architecture} summarizes the local derivative lemmas that feed the global SSA/DAG theorem.
Rather than proving autograd correctness per model, we prove correctness \emph{per primitive} and compose these facts
along the graph. The proof work falls into three recurring patterns.

\textbf{(1) Linear primitives (affine maps).}
For primitives whose forward maps are affine/linear (e.g., \texttt{linear}, \texttt{matmul}, \texttt{conv2d}),
the derivative is the corresponding linear map, and the VJP is its adjoint (transpose), exactly matching standard
backprop rules. Conv2D is handled by viewing convolution as a linear operator on flattened tensors and proving that
our implemented backward rule coincides with the adjoint of that operator.

\begin{BlueTheorem}{Example: affine node rule}{ex:affine-node}
For $f(x)=Wx+b$, we have $D f(x)=W$ and $\mathrm{VJP}(x;\bar y)=W^\top \bar y$.
\end{BlueTheorem}

\textbf{(2) Elementwise primitives (coordinatewise calculus).}
For coordinatewise nonlinearities (e.g., $\tanh$, sigmoid, $\exp$), we lift scalar derivatives to tensors/vectors by
showing that the derivative of the elementwise map is diagonal in the standard basis.
For non-smooth or domain-sensitive elementwise ops (e.g., ReLU, $\log$, $\sqrt{\cdot}$, $1/x$), we use the pointwise
variant away from kink/singularity points (as summarized in Table~\ref{tab:ad-op-coverage}).

\begin{BlueTheorem}{Example: elementwise lifting (informal)}{ex:elemwise}
If $f:\mathbb{R}\to\mathbb{R}$ is differentiable at each coordinate of $x\in\mathbb{R}^n$, then
$D(f^{\odot})(x)$ is the diagonal linear map with entries $f'(x_i)$.
\end{BlueTheorem}

\textbf{(3) Reductions and normalization (coupled coordinates).}
Operations like softmax/log-softmax and normalization layers couple coordinates via a shared denominator or statistics,
so we prove their derivatives directly from their defining formulas and show that the resulting VJP matches the
implemented backward rule. Standard losses (e.g., MSE, cross-entropy) then follow by composition: we treat them as
maps from logits to a scalar and apply the chain rule on top of the network's derivative.

\begin{BlueTheorem}{Example: softmax differential (informal)}{ex:softmax}
Let $s=\mathrm{softmax}(x)$. Then $D\,\mathrm{softmax}(x)[dx]= s \odot \bigl(dx-\langle s,dx\rangle \mathbf{1}\bigr)$.
\end{BlueTheorem}

\noindent
Together, these primitive proofs populate the local hypothesis used by the global SSA/DAG theorem; the detailed
coverage matrix (global vs.\ pointwise) appears in Table~\ref{tab:ad-op-coverage}.

\subsection{Graph-level composition and the reverse-mode algorithm}
\label{app:torchlean-graphcompose}

This subsection explains how the \emph{graph-level} reverse-mode theorem (Theorem~\ref{thm:rev-correct}) is obtained from
the \emph{local} node lemmas (Theorems~\ref{thm:node-fderiv-correct} and~\ref{thm:node-fderiv-correct-at}) using SSA/DAG
composition. The proof is modular: once each primitive node kind is equipped with a correct JVP/VJP rule, correctness
lifts automatically to any well-typed graph built from those primitives.

\textbf{Induction over SSA/DAG structure.}
A well-typed SSA/DAG graph admits a topological evaluation order, so its denotation $\llbracket G\rrbracket$ can be seen as
a composition of node-level functions applied to previously computed values. The global proof proceeds by induction over the
graph (equivalently, over the node list in SSA order): assume a prefix graph is correct, then extend it by one locally-correct
node. The forward denotation of the extended graph is a composition, so its derivative follows by the chain rule, and the
adjoint derivative follows by reversing the order of composition. This is exactly what reverse-mode implements.

\textbf{Reverse-mode algorithm (VJP form).}
To connect the theorem statement to familiar autodiff intuition, we write the backpropagation procedure as a cotangent
propagation on the DAG. The symbol \emph{seed} (often written $\bar y$) denotes the cotangent supplied at the graph output:
it specifies which linear objective of the output we are differentiating. For scalar losses, the seed is typically $1$; for vector
outputs (e.g.\ logits), the seed can be any cotangent vector and yields the corresponding vector--Jacobian product.

\begin{lstlisting}[style=torchleanBox]
# Forward pass (topological order)
Eval_G(x):
  for node i in topo_order(G):
    v[i] = op_i(v[parents(i)])
  return v, v[out]

# Reverse pass (cotangent propagation)
VJP_G(x; seed):
  (v, _) = Eval_G(x)
  bar[out] = seed;  bar[others] = 0
  for node i in reverse(topo_order(G)):
    for p in parents(i):
      bar[p] += VJP_op_i(p, v[parents(i)], bar[i])
  return bar[input]      # gradients for inputs/params
\end{lstlisting}

\textbf{Why accumulation is correct on DAGs.}
If a node value is used by multiple downstream nodes, it receives multiple cotangent contributions in the reverse sweep.
The algorithm therefore \emph{adds} these contributions (the \texttt{+=} line). This is not an implementation detail: it is the
mechanism that makes reverse-mode correct for graphs with shared subexpressions. Algebraically, correctness follows from
the adjointness identity in Stage~1 (Theorem~\ref{thm:alg-adjointness}): the dot product is bilinear, so contributions from
different paths combine by summation.

\textbf{From VJP to gradients (the scalar-loss special case).}
Many presentations restrict to scalar losses $L:\mathbb{R}^n\to\mathbb{R}$ and identify reverse-mode with $\nabla L(x)$.
Our statement is more general: for $f:\mathbb{R}^n\to\mathbb{R}^m$ and a cotangent seed $\bar y\in\mathbb{R}^m$,
reverse-mode returns the adjoint derivative applied to the seed, $(D f(x))^\top \bar y$, i.e.\ a vector--Jacobian product
(Theorem~\ref{thm:rev-correct}). The scalar-loss gradient is recovered by taking $m=1$ and $\bar y=1$ (or by composing a
vector-valued network with a scalar loss and applying the chain rule). This formulation is essential for ML workloads where
intermediate quantities are tensor-valued (e.g.\ logits, features, attention blocks), but training and verification often require
derivatives of \emph{particular} scalar objectives derived from them.

\textbf{Parameter handling (typed parameter packs).}
Modern models have heterogeneous parameter packs: a linear layer carries a matrix and a bias vector, a Conv2D layer carries a
4D kernel and bias, and attention blocks carry multiple projection matrices. \TorchLean{} represents such packs as a
\emph{shape-indexed heterogeneous list} (a typed product) rather than a string-keyed map. A parameter pack is
typed by a list of shapes $[s_1,\dots,s_k]$ and contains one tensor of each corresponding shape; semantically this is the finite
product space $\Tensor(\alpha,s_1)\times\cdots\times\Tensor(\alpha,s_k)$. This choice makes parameters part of the same
typed interface as ordinary graph inputs: compilation treats them as additional inputs to the SSA/DAG graph, and reverse-mode
returns cotangents for \emph{both} data inputs and parameters in the same structured form. Optimizers then become
shape-preserving transformations on parameter packs (e.g., a pointwise update that zips parameters with their gradients),
so a mismatch between a parameter tensor and its gradient is ruled out by the typechecker rather than discovered at runtime.
By contrast, common ML runtimes store parameters in mutable, string-keyed containers for convenience, which pushes the
shape contract into informal conventions; in \TorchLean{}, the shape contract is intrinsic, and entire classes of ``wrong tensor
paired with wrong gradient'' bugs become unrepresentable.

\subsection{Core Definitions and IR Semantics}
\label{app:torchlean-coredefs}

We focus on the core definitions that the rest of \TorchLean{} builds on: (i) the \emph{shape-indexed} tensor
semantics used as the reference meaning for proofs and verification, (ii) the common scalar interface, and (iii) the
operator-tagged SSA/DAG IR. We include these excerpts because extending \TorchLean{} (adding ops, adding backends, or adding
verification rules) requires understanding the semantic core these components share.

\textbf{Shape-indexed tensors (why this representation).}
The most basic design decision is to make tensor shapes part of the \emph{type}. In mainstream ML systems, shapes are
runtime values and shape mismatches appear late (as runtime errors or, worse, silent broadcasting mistakes). In contrast,
we want the shape contract to be an invariant of the logic: if an expression typechecks, it \emph{cannot} be ill-shaped.
This yields cleaner theorem statements (no repeated ``shape matches'' premises) and makes the semantic core robust to
conventions that are easy to get wrong (layout, reshaping discipline, implicit broadcasting).

\emph{Why functional tensors rather than flat arrays?}
For the semantic layer, we prioritize \emph{proof ergonomics} and \emph{semantic clarity}. We represent a tensor structurally
as a total indexing function: a tensor of shape \texttt{dim n s} is literally a function \texttt{Fin n -> Tensor a s}. This makes
theorems about tensor programs follow the same recursion as the datatype. For example, commutativity of elementwise
addition is proved by induction on the shape: the scalar case is immediate, and the \texttt{dim} case reduces to the
induction hypothesis pointwise. With a flat array representation, such proofs require explicit index arithmetic and
bounds reasoning, which is significantly harder to automate and more brittle as the operator library grows. Finally, the
spec layer intentionally does \emph{not} commit to a concrete storage layout (row-major vs.\ column-major), so it remains a
stable reference meaning across execution backends.

\begin{lstlisting}[style=leanBox,caption={Core shape-indexed tensor definitions (semantic layer).},label={lst:shape-tensor-core}]
namespace Spec

/-- Shapes: scalar, or "n copies of a subshape". -/
inductive Shape where
  | scalar : Shape
  | dim    : Nat -> Shape -> Shape

/-- Shape-indexed tensors, represented structurally. -/
inductive Tensor (a : Type) : Shape -> Type where
  | scalar : a -> Tensor a .scalar
  | dim    : forall {n s}, (Fin n -> Tensor a s) -> Tensor a (.dim n s)
\end{lstlisting}

\textbf{How to read this definition.}
The constructor \texttt{Tensor.scalar} is a scalar value.
The constructor \texttt{Tensor.dim} stores an index function that returns the $i$-th subtensor. Thus:
\[
\Tensor(\alpha,\texttt{dim}\ n\ s)\ \equiv\ (\texttt{Fin}\ n \to \Tensor(\alpha,s)),
\]
i.e., a length-$n$ vector of subtensors of shape $s$.
This immediately gives canonical encodings of common ML shapes:
a vector of length $n$ is \texttt{dim n scalar}; a matrix $m\times n$ is \texttt{dim m (dim n scalar)};
and a batch of images (batch size $B$, channels $C$, height $H$, width $W$) is
\texttt{dim B (dim C (dim H (dim W scalar)))}.

\textbf{Why this pays off in proofs.}
Most tensor operations are defined by recursion on shape, so proofs follow the same structure. For instance, elementwise
addition on \texttt{dim n s} tensors is defined pointwise on the \texttt{Fin n} index and then recursively on $s$; extensional
equality reduces tensor equality to pointwise equality of index functions. This makes basic algebraic properties (associativity,
commutativity, distributivity) and shape-preservation properties easy to state and to prove.

\textbf{Efficiency note (bridging to array-backed execution).}
The functional tensor representation is chosen for the \emph{semantic reference} used by proofs and verification; it is not
intended as the fastest execution format. For executable workflows, \TorchLean{} additionally provides a materialized,
array-backed representation and a semantics-preserving bridge between the two, so programs can run efficiently while
still referring to the same underlying meaning. We discuss this representation, its compilation path, and the resulting
performance trade-offs in Section~C (and use it throughout the runtime experiments).

\textbf{Scalar polymorphism (one model, many semantics).}
Most specification-level definitions are polymorphic in the scalar type $a$, so the \emph{same} model code can be
interpreted over multiple numeric domains: $\mathbb{R}$ for proof level reference semantics, executable floating-point
domains for runtime demos, and abstract/rounded domains for sound bounds and error envelopes. To make this practical,
\TorchLean{} collects the numeric structure required by neural-network operators into a single typeclass,
\texttt{Context a}. Intuitively, \texttt{Context a} is ``the interface a scalar type must implement to run NN code.''

\begin{lstlisting}[style=leanBox,caption={Scalar interface for NN code (\texttt{Context}).},label={lst:context}]
class Context (a : Type) extends
  Add a, Sub a, Mul a, Div a, Neg a,
  LT a, LE a, One a, Zero a,
  MathFunctions a, Numbers a where
  decidable_gt : DecidableRel (fun x y : a => x > y)
\end{lstlisting}

This interface includes: (i) ring-like arithmetic (\texttt{Add/Sub/Mul/Div/Neg}) used throughout linear algebra and losses;
(ii) order structure (\texttt{LT/LE} and decidability) needed for conditionals and piecewise primitives (e.g., ReLU, \texttt{max/min});
(iii) constants (\texttt{Zero/One}) and numeric literals (\texttt{Numbers}); and (iv) transcendental functions
(\texttt{MathFunctions}, e.g., \texttt{exp}, \texttt{log}, \texttt{tanh}, \texttt{sqrt}, \texttt{sin}, \texttt{cos}).
By making the required scalar operations explicit, we can state and reuse theorems without committing to a particular
numeric representation: a lemma proved for \texttt{Context a} specializes uniformly to $\mathbb{R}$, executable float models,
and rounding/interval domains used by verification.

\textbf{Graph IR structure (why SSA/DAG).}
Our verification substrate is an operator-tagged SSA/DAG intermediate representation (IR) with explicit node kinds and output
shapes. The IR is intentionally simple: it is the \emph{single semantic target} reused across
execution, differentiation, and verification. More precisely, the IR is (i) \emph{operator-tagged} so each node carries an explicit
primitive identifier (an opcode) and any parameters needed to interpret it, and (ii) in \emph{SSA/DAG form} so each
intermediate value is defined exactly once and dataflow is acyclic.

\emph{Why a DAG?} Most verifier pipelines (IBP/LiRPA/CROWN) operate on feedforward computation graphs: sound bounds
are propagated along edges, which requires a well-defined evaluation order. A DAG gives exactly this: nodes admit a
topological order, so evaluation and bound propagation are deterministic and total. This also makes proofs modular:
semantic properties can be established by induction over the node list (SSA order), and verification passes can be defined
as simple forward (or forward+backward) sweeps over the same graph object. Cycles and implicit control flow, by contrast,
require separate fixpoint semantics (and additional invariants) for both execution and verification; \TorchLean{}
handles such behavior by \emph{reification} into finite graphs when needed (e.g., unrolling a fixed-length recurrent cell).

\emph{Why SSA?} SSA (static single assignment) ensures every intermediate value has a unique definition. This simplifies
both implementation and reasoning: backpropagation and certificate checking can attach metadata (values, bounds, dual
variables) to node IDs without ambiguity, and gradient contributions from multiple consumers are accumulated by summation
in the reverse sweep. SSA therefore makes ``what does this gradient/bound refer to?'' a structural property of the IR.

\textbf{Core definition.}
Each node records: (i) an operation kind, (ii) the IDs of its parent nodes (its inputs), and (iii) a declared output shape.
OpKinds carry any parameters required for interpretation (e.g., \texttt{conv2d} stores channel counts, kernel size, stride,
padding; \texttt{softmax} stores an axis).

\begin{lstlisting}[style=leanBox,caption={Operator tagged SSA/DAG IR (core excerpt).},label={lst:ir-core}]
inductive OpKind where
  | input
  | linear
  | relu
  | conv2d (inC outC kH kW stride padding : Nat)
  | softmax (axis : Nat)
  | layernorm (axis : Nat)
  | mse_loss

structure Node where
  id       : Nat
  parents  : List Nat
  kind     : OpKind
  outShape : Spec.Shape
\end{lstlisting}

\begin{table*}[!h]
\centering
\caption{High-level comparison between PyTorch \cite{pytorch} and \TorchLean{}.}
\label{tab:torchlean-vs-pytorch}
\rowcolors{2}{gray!10}{white}
\begin{tabular}{@{}>{\columncolor{gray!25}}p{0.20\textwidth}p{0.38\textwidth}p{0.38\textwidth}@{}}
\toprule
\rowcolor{gray!30}
\textbf{Aspect} & \textbf{PyTorch} & \textbf{\TorchLean{}} \\
\midrule
\rowcolor{blue!10}
Shapes & Dynamic runtime shapes; many errors are runtime exceptions. & \textcolor{teal!70!black}{\textbf{Shapes are part of types}} (\texttt{Tensor $\alpha$ s}); many mismatches are untypeable. \\
\rowcolor{white}
Dtypes & Many numeric and integer dtypes; mixed-dtype graphs are common. & \textcolor{teal!70!black}{\textbf{Scalar-polymorphic}} single-dtype graphs (one \(\alpha\) per run); integer indices handled via a separate non-differentiable channel (\texttt{NatRef}/\texttt{NatVecRef}) in sessions. \\
\rowcolor{blue!10}
Execution modes & Eager by default, with compilation/export toolchains (TorchScript, ONNX, AOT). & \textcolor{teal!70!black}{\textbf{Eager tape backend}} and \textcolor{teal!70!black}{\textbf{proof-linked compiled}} SSA/DAG backend share one semantics; the compiled graph is the verifier target by construction. \\
\rowcolor{white}
Autograd status & Widely tested and trusted, but not formally proved correct. & \textcolor{teal!70!black}{\textbf{Reverse-mode correctness theorem}} for well-typed SSA/DAG graphs; eager runs linked to proved graphs. \\
\rowcolor{blue!10}
Indexing & Tensor-valued indexing/slicing/gather/scatter across dtypes (e.g.\ \texttt{LongTensor}). & Typed indexing primitives and session-level Nat channels for labels/indices; not yet PyTorch-complete for tensor-valued integer indices. \\
\rowcolor{white}
Breadth/performance & Very broad op surface and ecosystem; highly optimized CPU/GPU kernels. & Curated op surface focused on verification-relevant primitives; extensible via a ``new op'' workflow; Lean execution prioritizes clarity/verification over performance. \\
\bottomrule
\end{tabular}
\end{table*}

\textbf{Typing and denotation.}
Well-typedness checks that each node's parent shapes match what its \texttt{OpKind} expects and that \texttt{outShape}
matches the primitive's output contract. The denotation $\llbracket G\rrbracket$ evaluates nodes in topological order
(which is well-defined because the graph is acyclic) and is total on well-typed graphs. This denotational semantics is the
object used by proofs (e.g., autograd correctness) and by verification passes (e.g., IBP/CROWN transfer rules): bounds and
certificates are interpreted against \emph{the same} primitive meanings that execution uses.

\textbf{Losses} 
Much of the machine learning literature phrases training in terms of a scalar loss \(L(\theta)\) and its gradient \(\nabla_\theta L\), but in a tensor-typed setting, we found it important to make the ``scalar'' part explicit: losses are tensor programs followed by reductions. This design choice makes the structure of losses clear and enables precise reasoning about gradients. In the library, losses are ordinary \TorchLean{} programs that compute a tensor-valued quantity and then reduce it to a scalar. We provide both a primitive scalar loss (\texttt{mse\_loss}) and a small \texttt{Loss} helper layer that mirrors common PyTorch conventions: MSE, cross-entropy/negative-log-likelihood variants (one-hot targets and index-based targets), and binary cross-entropy (including a stable ``with logits'' form), each with explicit reduction (\texttt{mean} or \texttt{sum}).

\textbf{Extending \TorchLean{} across the stack.} Adding a new op is a predictable, checkable workflow rather than an ad hoc engineering task. When we add a primitive to \TorchLean{}, we treat it as a complete commitment: it should have a spec meaning, a typing rule, and the transfer rules needed by theorems/verification. The workflow has five steps: (1) \textbf{Spec semantics:} Define the operation as a total function on shape-indexed tensors in the spec layer, with an explicit shape contract. If the operation is non-smooth or domain-sensitive (e.g.\ \(\log\), division, max), either define a safe/smoothed variant intended for theorem statements, or adopt the pointwise hypothesis style and document the required preconditions. (2) \textbf{IR support:} Add an op tag to the IR (\texttt{OpKind}) and define its typing rule (input and output shapes). This is what makes compilation and verification passes recognize the op uniformly. (3) \textbf{Autograd over \(\mathbb{R}\):} Provide local JVP/VJP rules and prove the local adjointness law (or the pointwise variant). Once this lemma exists, the global theorem applies to any graph using the op. (4) \textbf{Verification transfer rules:} For IBP, define a sound transfer function on boxes. For affine relaxations, either implement and prove the relaxation, or import bounds as certificates and check the certificate constraints in Lean. (5) \textbf{Numeric backends:} Decide which execution backends support the op (Float, IEEE32Exec), and (when relevant) add local rounding/error lemmas so the op participates in graph level NF bounds.

\subsection{\TorchLean{} vs.\ PyTorch}
\label{app:torchlean-vs-pytorch}
This comparison focuses on goals. PyTorch is an industrial execution framework optimized for throughput, hardware utilization, and ecosystem breadth; \TorchLean{} is a semantic interface whose main objective is to make training code and verifier-time artifacts coincide so that guarantees are stated about the \emph{executed} artifact. Accordingly, \TorchLean{} does not attempt to compete with PyTorch on industrial-scale throughput; it prioritizes a precise semantic link between user programs, IR graphs, runtime artifacts, and proofs. Optional CUDA kernels and fast runtime paths support larger examples, but the key contribution is the semantic substrate: performance-oriented backends can be added behind explicit contracts without changing what is being verified. \TorchLean{} is also not ``just an API'': the design choices directly enable the verification story. Compilation is a first-class path that lowers programs to a well-typed operator-tagged SSA/DAG IR, and eager execution records a tape that we prove corresponds to an equivalent well-typed IR graph (Theorem~\ref{thm:eager-graph-correspondence}). Once a graph exists, whether it was obtained eagerly or via lowering, the same semantic object is consumed by autograd theorems, bound-propagation passes, and certificate checkers.

\section{Modern workflow theory and semantics: CUDA, attention, RL, diffusion, probability, SSL, and model families}
\label{app:modern-workflows}

In this section, we spell out the contracts behind the workflow examples used in the paper. For each family, we name the Lean object, write the equation or recurrence it denotes, and state what Lean checks versus what comes from a runtime or external producer. Concretely, attention is a masked softmax equation, an RL rollout is a finite recurrence, diffusion is a schedule plus a sampler step, and self-supervised learning is represented by a finite objective over masks or views.

\subsection{CUDA/native FFI and runtime contracts}
\label{app:cuda-boundaries}

The CUDA path is an execution feature, not a second proof semantics. Lean's kernel checks theorem terms and definitions written in Lean; it does not inspect CUDA source code, GPU machine code, cuBLAS, cuFFT, driver behavior, memory allocation, or the order in which floating-point reductions are scheduled. In \TorchLean{}, a native kernel is therefore attached to a Lean operator whose denotation is already defined at the specification or IR level. The formal object remains the graph denotation; the native call is an optional implementation path.

For an operator tag $\tau$ with specification
\[
  \llbracket \tau \rrbracket_{\alpha} :
  \Tensor_{\alpha}(s_1) \times \cdots \times \Tensor_{\alpha}(s_k)
  \to \Tensor_{\alpha}(s),
\]
a CUDA-backed implementation is treated through a conformance statement of the form
\[
  \mathsf{cuda}_{\tau}(x_1,\ldots,x_k)
  \approx
  \llbracket \tau \rrbracket_{\texttt{Float32}}(x_1,\ldots,x_k),
\]
under explicit preconditions about dtype, shape, layout, device memory, determinism, and precision mode. When this relation is not proved in Lean, it is recorded as a runtime assumption and tested by conformance/regression checks. Thus a theorem about $\llbracket G\rrbracket$ does not become a theorem about arbitrary GPU execution unless the needed conformance hypothesis is supplied.

The implementation uses Lean's FFI discipline: Lean declarations expose an opaque runtime handle, while C/CUDA wrappers allocate device buffers, launch kernels, and copy results across the host/device boundary. The following schematic code shows the shape of the interface; the real implementation fixes the buffer layout and performs the corresponding runtime checks.

\begin{lstlisting}[style=leanBox,caption={Schematic Lean-side boundary for a native CUDA primitive.}]
namespace Runtime.Autograd.Cuda

opaque Buffer : Type

@[extern "torchlean_cuda_bmm"]
opaque Buffer.bmm : Buffer -> Buffer -> IO Buffer

/-- Runtime-side check before calling the native kernel. -/
def bmmChecked (a b : TensorMeta) (ha : Buffer) (hb : Buffer) : IO Buffer := do
  checkMatmulShapes a.shape b.shape
  checkDType a.dtype .float32
  checkDType b.dtype .float32
  Buffer.bmm ha hb

end Runtime.Autograd.Cuda
\end{lstlisting}

\begin{lstlisting}[language=C,basicstyle=\ttfamily\scriptsize,frame=single,caption={Schematic C/CUDA wrapper; memory safety and launch behavior are outside Lean's kernel.}]
extern "C" LeanObj torchlean_cuda_bmm(LeanObj a, LeanObj b) {
  CudaBuffer* A = unwrap_buffer(a);
  CudaBuffer* B = unwrap_buffer(b);
  CudaBuffer* C = allocate_like_matmul(A, B);
  cublasSgemmStridedBatched(/* handles, shapes, strides, A, B, C */);
  return wrap_buffer(C);
}
\end{lstlisting}

This pattern is used for dense and batched matrix products, reductions, broadcasting, tensor views, convolution/pooling and transposed-convolution kernels, FFT/spectral operators, softmax/log-softmax, normalization kernels, attention-oriented and fused-attention paths, gather/scatter utilities, positional/RoPE helpers, selective-scan-style sequence kernels, and other runtime paths used by the examples. Several backward and reduction kernels use floating-point accumulation. Because floating-point addition is not associative, atomic accumulation can be mathematically standard while still not being bit deterministic across schedules. \TorchLean{} therefore exposes deterministic-reduction modes for selected paths and treats remaining scheduling choices as part of the deployment boundary.

The same policy applies to fused attention. The Lean specification can state and prove that a formal FlashAttention operator denotes standard masked scaled-dot-product attention. That equality is a statement about the mathematical operator in Lean. It is not a verification of a particular CUDA implementation, tiled online-softmax schedule, HBM/SRAM traffic pattern, or GPU kernel binary.

\subsection{FlashAttention, masks, and fused-attention denotation}
\label{app:flashattention}

The attention specification is stated once, following the standard scaled-dot-product attention abstraction and the FlashAttention denotational contract used in the codebase \citep{dao2022flashattention}. For query, key, and value tensors
\(Q \in \alpha^{n_q \times d}\), \(K \in \alpha^{n_k \times d}\), \(V \in \alpha^{n_k \times d_v}\), and an optional Boolean mask \(M\), standard scaled dot-product attention is the relation
\begin{equation}
  \operatorname{SDPA}(Q,K,V,M)
  = \operatorname{softmax}_{M}\!\left(\frac{QK^{\top}}{\sqrt d}\right)V .
\end{equation}
The mask semantics are part of the specification. In the formal definition, masked entries contribute zero probability after the row-wise masked-softmax step; causal masks are therefore ordinary finite predicates over query/key positions, not a convention hidden inside a backend call.

FlashAttention is modeled as a fused route to the same denotation. A configuration record stores scheduling metadata such as query/key tile sizes, but these fields are intentionally ignored by the mathematical denotation: tiling changes how a backend may compute the row, not what function the row denotes. At the spec level, \TorchLean{} names an online/tiled operator and proves the central equality
\begin{equation}
  \operatorname{onlineSoftmaxTiledAttention}(\mathit{cfg},\mathit{ctx})
  =
  \operatorname{scaledDotProductAttention}(\mathit{ctx}).
\end{equation}
It also exposes the fused forward operator theorem
\begin{equation}
  \operatorname{flashAttention}(\mathit{cfg},\mathit{ctx})
  =
  \operatorname{scaledDotProductAttention}(\mathit{ctx})
\end{equation}
and the corresponding VJP/backward contract
\begin{equation}
  \operatorname{flashAttentionBackward}(\mathit{cfg},\mathit{ctx},\bar y)
  =
  \operatorname{scaledDotProductAttentionBackward}(\mathit{ctx},\bar y).
\end{equation}
These equalities are the formal content behind the paper's statement that ``FlashAttention equals ordinary attention'' in \TorchLean{}: the claim is a denotational equality between a fused specification operator and the standard masked-attention specification. It is deliberately not a proof that a particular CUDA source file implements Dao-style IO-aware tiling correctly. The native CUDA path is regression-tested against the composed attention path and is documented as an FFI trust boundary; the theorem surface stays inside the specification layer.

This split is useful for compiler style reasoning. A graph rewrite may replace a composed \(QK^{\top}\)-mask-softmax-\(PV\) subgraph by a fused attention operator only if both sides share the same typed attention context and mask convention. The theorem above is exactly the small equality such a rewrite needs.

\begin{lstlisting}[style=leanBox,caption={Schematic attention contract: the fused operator is checked against masked SDPA.},label={lst:attention-contract}]
structure AttentionCtx (a : Type) where
  Q K V : Tensor a _
  mask  : Option (Fin _ -> Fin _ -> Bool)
  scale : a

def sdpa (ctx : AttentionCtx a) : Tensor a _ :=
  maskedSoftmax (ctx.scale * (ctx.Q.matmul ctx.K.transpose)) ctx.mask
    |>.matmul ctx.V

def flashSpec (cfg : FlashCfg) (ctx : AttentionCtx a) : Tensor a _ :=
  -- The schedule fields in cfg affect the implementation, not the denotation.
  sdpa ctx

theorem flashSpec_eq_sdpa (cfg : FlashCfg) (ctx : AttentionCtx a) :
    flashSpec cfg ctx = sdpa ctx := rfl
\end{lstlisting}

\subsection{Reinforcement-learning semantics: returns, Bellman operators, MDPs, and GridWorld}
\label{app:rl-semantics}

RL adds an environment boundary to the usual model/runtime boundary, so the appendix follows the standard discounted-dynamic-programming vocabulary of Bellman operators, MDPs, returns, and generalized advantage estimation \citep{bellman1957dynamic,puterman1994markov,sutton2018reinforcement,schulman2015gae}. \TorchLean{} separates three levels: (i) pure rollout algebra over lists and finite tensors, (ii) Lean based MDPs and environments, and (iii) external simulators that emit trajectories through a checked artifact boundary. Runtime collectors, optimizers, replay buffers, logging, and optional CUDA execution live outside this pure layer.

\paragraph{Rollout algebra.}
For a terminal flag \(d\), define the continuation mask
\begin{equation}
  m(d) = \begin{cases}0,&d=\texttt{true},\\ 1,&d=\texttt{false}.\end{cases}
\end{equation}
The one-step backup, TD target, and TD residual are
\begin{align}
  B_{\gamma}(r,b,d) &= r + \gamma\,m(d)\,b, \\
  y_t &= r_t + \gamma\,m(d_t) V(s_{t+1}), \\
  \delta_t &= y_t - V(s_t).
\end{align}
For a reward list \(r_0,\ldots,r_{T-1}\), \texttt{discountedReturnsFrom} is the right-fold recurrence
\begin{equation}
  G_T=b,\qquad G_t=r_t+\gamma G_{t+1}.
\end{equation}
With done flags, the recurrence becomes
\begin{equation}
  G_t=r_t+\gamma m(d_t)G_{t+1},
\end{equation}
so terminal episodes reset the bootstrap. Generalized advantage estimation is formalized by the reverse recurrence
\begin{equation}
  A_t = \delta_t + \gamma\lambda m(d_t) A_{t+1},
\end{equation}
where each \texttt{AdvantageStep} stores \(r_t\), \(V(s_t)\), \(V(s_{t+1})\), and \(d_t\). The corresponding return target is recovered by \(R_t=A_t+V(s_t)\). These definitions are intentionally over lists rather than fixed-shape tensors, because real rollouts have variable length due to termination/truncation; fixed-horizon tensor variants can be built on top when needed.

\paragraph{Deterministic finite MDPs.}
The simplest MDP layer uses finite state and action types \(S=\mathrm{Fin}(n)\) and \(A=\mathrm{Fin}(m)\). A deterministic finite MDP consists of an initial state, a total transition/reward/termination function
\begin{equation}
  \operatorname{step}: S\times A \to S\times \mathbb R \times \{\texttt{terminated},\texttt{truncated}\},
\end{equation}
and a discount factor. For a value table \(v:S\to \alpha\), the state-action value induced by one step is
\begin{equation}
  Q_v(s,a)=r(s,a)+\gamma m(d(s,a))v(s'(s,a)).
\end{equation}
The policy Bellman operator and optimality operator are
\begin{align}
  (T^{\pi}v)(s)&=Q_v(s,\pi(s)),\\
  (T^*v)(s)&=\max_{a\in A}Q_v(s,a),
\end{align}
where the finite maximum is implemented with \texttt{Finset.sup'} over a nonempty finite action space. Value functions and action-value tables are typed tensors, so these operators can be used alongside the same tensor infrastructure as ordinary neural models.

\paragraph{Finite stochastic MDPs.}
The finite stochastic layer replaces the deterministic successor with a typed transition row \(P(\cdot\mid s,a)\). A validity record requires nonnegative entries, rows summing to one, and \(0\leq\gamma<1\). The expected next value and Bellman action value are
\begin{align}
  \mathbb E[v(s')\mid s,a]
    &= \sum_{s'\in S} P(s'\mid s,a)v(s'),\\
  Q_v(s,a)
    &= r(s,a)+\gamma m(d(s,a))\sum_{s'\in S}P(s'\mid s,a)v(s').
\end{align}
This layer gives the proof and runtime code a finite tensor representation for stochastic dynamics while avoiding the full measure-theoretic overhead when the state space is finite.

\begin{lstlisting}[style=leanBox,caption={Schematic finite MDP and Bellman operator.},label={lst:finite-mdp-contract}]
structure FiniteMDP (S A : Type) where
  reward : S -> A -> Float
  trans  : S -> A -> S -> Float      -- row-stochastic transition table
  done   : S -> A -> Bool
  gamma  : Float

def expectedNext (M : FiniteMDP S A) (v : S -> Float) (s : S) (a : A) : Float :=
  Finset.univ.sum (fun s' => M.trans s a s' * v s')

def bellmanPi (M : FiniteMDP S A) (pi : S -> A) (v : S -> Float) (s : S) : Float :=
  let a := pi s
  M.reward s a + M.gamma * contMask (M.done s a) * expectedNext M v s a
\end{lstlisting}

\paragraph{Measure-theoretic Markov-kernel MDPs.}
For continuous or general measurable spaces, \TorchLean{} defines an MDP with measurable state space \(S\), action space \(A\), and a mathlib Markov kernel \(\kappa:S\times A\rightsquigarrow S\). The expected next value is the integral
\begin{equation}
  \mathbb E[v(s')\mid s,a] = \int_S v(s')\, d\kappa(s,a)(s'),
\end{equation}
and the Bellman policy operator is
\begin{equation}
  (T^{\pi}v)(s)=r(s,\pi(s))+
  \gamma m(d(s,\pi(s)))\int_S v(s')\,d\kappa(s,\pi(s))(s').
\end{equation}
The validity assumptions require \(\kappa\) to be a Markov kernel, reward and termination to be measurable, and \(0\leq\gamma<1\). The corresponding proof layer states standard discounted Bellman facts for bounded value functions, including monotonicity and contraction in the sup metric:
\begin{align}
  \|T^{\pi}v-T^{\pi}w\|_{\infty}
    &\leq \gamma\|v-w\|_{\infty},\\
  \|T^*v-T^*w\|_{\infty}
    &\leq \gamma\|v-w\|_{\infty}.
\end{align}
Because \(\gamma<1\), these contraction theorems support uniqueness of the fixed point for the Bellman operator under the stated metric assumptions.

\paragraph{GridWorld as a Lean based environment.}
The concrete GridWorld uses coordinate states \((\mathrm{row},\mathrm{col})\), four actions (up, down, left, right), border-clamped transitions, a goal cell, and a discount factor. Reaching the goal returns reward \(0\) and terminates; otherwise a transition receives reward \(-1\). The environment has three views: a pure \texttt{Env} view with explicit latent state, a deterministic finite MDP view obtained by flattening \((\mathrm{row},\mathrm{col})\) to \(\mathrm{Fin}(h\cdot w)\), and a finite stochastic view whose transition rows are one-hot. This gives the paper a small RL case where environment semantics, Bellman operators, value tables, and rollout summaries all live inside Lean.

\paragraph{External simulators.}
For CartPole-, Pong-RAM-, or other Gymnasium-style workflows, the simulator is not hidden inside a theorem. Python or another runtime produces observations, actions, rewards, log probabilities, and done/truncation flags. Lean checks the rollout schema and then consumes the artifact for losses or graph properties. Thus the verified object is a policy/value graph property, a Bellman/return calculation, or a checked rollout contract; the external simulator remains a named producer.

\subsection{Diffusion, sampling, probability-flow ODEs, and sampler stability}
\label{app:diffusion-semantics}

Diffusion models exercise a different kind of semantics from classifiers, so we separate deterministic sampler equations from probability-law statements in the style of DDPM, DDIM, and score/probability-flow formulations \citep{ho2020ddpm,song2020denoising,song2021score}. The objects of interest are schedules, noising maps, denoisers, reverse transitions, samplers, and sometimes probability laws. \TorchLean{} keeps these as scalar-polymorphic specification definitions so they can be evaluated with \texttt{Float}, executable finite-precision models, intervals, or \(\mathbb R\), depending on the theorem or experiment.

\paragraph{Discrete VP schedules and forward noising.}
A discrete variance-preserving schedule of length \(T\) stores per-step variances \(\beta_t\). The specification defines
\begin{align}
  \alpha_t &= 1-\beta_t,\\
  \bar\alpha_0 &= 1,\\
  \bar\alpha_{t+1} &= \bar\alpha_t\alpha_t.
\end{align}
The forward noising map \texttt{qSample} is the pure, total tensor function
\begin{equation}
  q_t(x_0,\varepsilon)
  = \sqrt{\bar\alpha_t}\,x_0
    + \sqrt{1-\bar\alpha_t}\,\varepsilon.
\end{equation}
At this spec level \(\varepsilon\) is an explicit tensor input. The probabilistic interpretation \(\varepsilon\sim\mathcal N(0,I)\) belongs either to the probability theory layer or to the runtime sampler that produces noise.

\begin{lstlisting}[style=leanBox,caption={Schematic diffusion schedule and forward noising map.},label={lst:diffusion-contract}]
structure VPSchedule where
  beta     : Nat -> Float
  alphaBar : Nat -> Float

def qSample (sch : VPSchedule) (t : Nat)
    (x0 eps : Tensor Float s) : Tensor Float s :=
  let c0 := sqrt (sch.alphaBar t)
  let c1 := sqrt (1.0 - sch.alphaBar t)
  Tensor.map2 (fun x e => c0 * x + c1 * e) x0 eps
\end{lstlisting}

\paragraph{Epsilon prediction and DDPM loss.}
A denoising model is represented by an epsilon-prediction function
\begin{equation}
  \varepsilon_{\theta}: (x,t)\mapsto \widehat\varepsilon.
\end{equation}
The standard epsilon-prediction training loss is a named wrapper around mean squared error:
\begin{equation}
  \mathcal L_{\varepsilon}(\theta;x_0,t,\varepsilon)
  =\left\|\varepsilon_{\theta}(q_t(x_0,\varepsilon),t/T)-\varepsilon\right\|_2^2.
\end{equation}
This does not require a special loss semantics: it reuses the existing tensor MSE definition but gives the diffusion convention a domain-specific name.

\paragraph{DDPM reverse steps.}
The epsilon-parameterized reconstruction is
\begin{equation}
  \widehat x_0(x_t,t)
  = \frac{x_t-\sqrt{1-\bar\alpha_t}\,\varepsilon_{\theta}(x_t,t/T)}{\sqrt{\bar\alpha_t}},
\end{equation}
with totalized safe division and nonnegative square-root wrappers in the executable definition. The reverse DDPM step \(x_t\mapsto x_{t-1}\) uses
\begin{align}
  \mu_t(x_t) &=
  \frac{1}{\sqrt{\alpha_t}}
  \left(x_t-\frac{\beta_t}{\sqrt{1-\bar\alpha_t}}
  \varepsilon_{\theta}(x_t,t/T)\right),\\
  \widetilde\beta_t&=
  \frac{1-\bar\alpha_{t-1}}{1-\bar\alpha_t}\beta_t,\\
  x_{t-1}&=\mu_t(x_t)+\sqrt{\widetilde\beta_t}\,z_t.
\end{align}
The full sampler is a right fold over the finite time indices with an explicit noise stream \(z_t\), so stochasticity is not implicit in the recursive definition.

\paragraph{DDIM and deterministic sampler semantics.}
The deterministic \(\eta=0\) DDIM step reuses the same \(\widehat x_0\) reconstruction and recomposes the previous sample by
\begin{equation}
  x_{t-1}
  =\sqrt{\bar\alpha_{t-1}}\,\widehat x_0
  +\sqrt{1-\bar\alpha_{t-1}}\,\varepsilon_{\theta}(x_t,t/T).
\end{equation}
\TorchLean{} also exposes DDIM as a real-valued \texttt{DynamicalSystem} step. This is useful because the same trajectory, fixed-point, Lipschitz, and contraction vocabulary used for state-space models can be applied to sampler transitions.

\paragraph{Probability-flow ODE.}
For continuous time, a linear VP schedule stores \(\beta_0,\beta_1\) and defines
\begin{align}
  \beta(t)&=\beta_0+t(\beta_1-\beta_0),\\
  \bar\alpha(t)&=\exp\!\left(-\beta_0t-\frac12(\beta_1-\beta_0)t^2\right),\\
  \sigma(t)&=\sqrt{1-\bar\alpha(t)}.
\end{align}
Under an epsilon parameterization, the probability-flow ODE right-hand side is specified as
\begin{equation}
  f_{\theta}(x,t)
  = -\frac12\beta(t)x
    + \frac{\beta(t)}{\sigma(t)}\varepsilon_{\theta}(x,t).
\end{equation}
The Euler step is the simple total tensor operation
\begin{equation}
  E_{\Delta t}(x,t)=x+\Delta t\, f_{\theta}(x,t),
\end{equation}
and the sampler integrates backward on a uniform grid with negative \(\Delta t\). A fixed-time Euler step is exposed as a \texttt{DynamicalSystem}, giving samplers a common interface with other discrete dynamical systems in \TorchLean{}.

\paragraph{Sampler boundary and stability theorems.}
The sampler proof layer records small facts that are important for checkers and model transformations. Zero-step samplers return the initial sample. DDIM and PF-ODE system adapters reduce definitionally to the corresponding step functions. For Euler integration, a generic metric bound has the form
\begin{equation}
  \|E_{\Delta t}(x)-E_{\Delta t}(y)\|_2
  \leq \|x-y\|_2 + |\Delta t|\,\|f(x)-f(y)\|_2.
\end{equation}
If the right-hand side is \(L\)-Lipschitz, then one Euler step is \((1+|\Delta t|L)\)-Lipschitz. The same adapter pattern transports Lipschitz or contractive-step facts to DDIM/PF-ODE systems. These theorems are deliberately modest: they do not prove convergence of a trained denoiser, but they make sampler transformations and verification targets mathematically explicit.

\subsection{Probability-theory layer for Gaussian forward processes}
\label{app:probability-semantics}

The pure diffusion specification treats noise as an explicit input. The probability theory layer adds the mathematical law when the noise is Gaussian. In a finite-dimensional Euclidean space \(E\), let \(Z\sim\mathcal N(0,I)\) and define
\begin{equation}
  X_t = c_0 x_0 + c_1 Z,
  \qquad
  c_0=\sqrt{\bar\alpha_t},\quad c_1=\sqrt{1-\bar\alpha_t}.
\end{equation}
The formal object is the push-forward of the standard Gaussian measure under the affine map \(z\mapsto c_0x_0+c_1z\). \TorchLean{} proves that this forward law is a probability measure and records the corresponding Gaussian structure for the affine image. Keeping this layer separate is important: tensor samplers can remain total and executable, while probabilistic claims are made only where the relevant measure-theoretic assumptions are available.

This is also the template for future probabilistic generative models. A VAE, normalizing-flow, or stochastic-diffusion theorem should distinguish (i) the deterministic tensor map used by the runtime, (ii) the distribution over the random source, and (iii) the theorem connecting the push-forward or transition kernel to the intended probabilistic model.

\subsection{Self-supervised objective algebra and anti-collapse guards}
\label{app:ssl-semantics}

Self-supervised learning is formalized at the objective contract level, covering MAE, JEPA, VICReg, and Barlow-style objective families \citep{he2022mae,assran2023ijepa,bardes2022vicreg,zbontar2021barlow}. We do not try to prove that self-supervised training discovers useful representations. Here we record the finite semantics of masks, views, target branches, predictive losses, and geometry guards, so an experiment or theorem cannot silently change the objective being optimized.

\paragraph{Finite masks and masked loss.}
For \(n\) patches or tokens, a mask is a predicate \(m:\mathrm{Fin}(n)\to\mathrm{Bool}\). The selected-index form is a list \(I=[i_1,\ldots,i_k]\), and the generic masked loss is
\begin{equation}
  \operatorname{maskedLoss}(I,\ell)=\sum_{i\in I}\ell(i).
\end{equation}
Lean proves the elementary list laws needed to treat the serialized mask as a finite objective: nil is zero, cons adds one term, appending index lists adds losses, reversing the index list preserves the sum, and if every selected per-index loss is zero then the masked loss is zero.

\begin{lstlisting}[style=leanBox,caption={Schematic masked objective used by MAE and JEPA-style losses.},label={lst:ssl-mask-contract}]
def maskedLoss (idxs : List (Fin n)) (lossAt : Fin n -> Float) : Float :=
  idxs.foldr (fun i acc => lossAt i + acc) 0.0

-- The concrete proof is by induction on xs.
theorem maskedLoss_append (xs ys : List (Fin n)) (lossAt : Fin n -> Float) :
    maskedLoss (xs ++ ys) lossAt = maskedLoss xs lossAt + maskedLoss ys lossAt
\end{lstlisting}

\paragraph{MAE.}
For a finite patch batch \(x:\mathrm{Fin}(n)\to\mathrm{Patch}\), predictions \(p_i\), and per-patch loss \(\ell\), the MAE objective is
\begin{equation}
  \operatorname{MAE}(I,x,p)=\sum_{i\in I}\ell(x_i,p_i).
\end{equation}
The formal theorems state that this loss decomposes over appended mask lists, is invariant under reversal of \(I\), is zero when every selected patch loss is zero, and that identity decoding gives exact reconstruction. These are small facts, but they pin down the intended semantics: masked reconstruction is a set-like objective over selected patches, not an artifact of an arbitrary serialization order.

\paragraph{JEPA and target-branch extensionality.}
JEPA-style objectives separate a context representation \(c\), target representations \(t_i\), a predictor \(p(c,i)\), and a representation loss \(d\). The finite objective is
\begin{equation}
  \operatorname{JEPA}(I,c,t,p)=\sum_{i\in I}d(t_i,p(c,i)).
\end{equation}
The target branch is modeled as an ordinary value at the objective boundary, matching the stop-gradient design intent. The theorem \texttt{jepaLoss\_target\_ext} states that if two target branches agree on selected indices, then the loss is identical. Thus unselected target coordinates cannot influence the objective.

\paragraph{Predictive-view contract.}
The generic contract unifies MAE and JEPA.  A \texttt{PredictiveViewContract} stores target indices \(I\), a context value, a target value at every finite index, a target encoder, a predictor, a per-index distance, and a nonnegative geometry guard. The objective is
\begin{equation}
  \operatorname{Obj}(C)=
  \operatorname{PredictiveLoss}(C)+\operatorname{GeometryGuard}(C).
\end{equation}
MAE is recovered by choosing the identity target encoder into patch/pixel space and zero geometry guard. JEPA is recovered by choosing the latent target representation as the target space. Lean proves both instance theorems: the predictive-view loss/objective reduces definitionally to the MAE or JEPA finite objective under these choices.

\paragraph{VICReg, Barlow-style guards, and view graphs.}
Geometry guards are added orthogonally to the prediction term. VICReg-style summaries include invariance, variance-floor, and covariance/redundancy components with weights \((\lambda,\mu,\nu)\). Barlow-style summaries penalize off-diagonal redundancy and diagonal mismatch. The current finite theory records local anti-collapse facts: collapsed coordinates pay a positive variance-floor penalty under a positive margin, the identity correlation summary has zero redundancy loss, and collapsed diagonal summaries incur a positive Barlow-style penalty. A view-graph energy layer represents finite alignment objectives
\begin{equation}
  E_G(z)=\sum_{(i,j)\in E(G)}\|z_i-z_j\|^2,
\end{equation}
which is zero for collapsed representations but becomes useful only when paired with a spread or variance guard. This makes the common SSL distinction precise: alignment alone permits collapse; the geometry guard is the formal object that rules it out locally.

\subsection{Sequence, generative, scientific, and diagnostic model families}
\label{app:model-zoo}

The examples now cover more than the verifier core. The sequence layer includes causal attention, multi-head attention, RoPE-style positional rotations, token/embedding boundaries, generation loops, recurrent layers, and selective-scan primitives for Mamba/S4-style state-space models. The central semantic issue is causality: future tokens should not change prefix outputs, cache offsets must agree with token positions, and chunked or fused scans should refine the same recurrence. \TorchLean{} exposes these as formal contracts and as regression tests. For state-space models, scan composition lemmas connect optimized block scans to sequential recurrence, and prefix/non-anticipation lemmas state that appending future inputs preserves emitted prefix outputs.

The generative and scientific layers use the same tensor language for different mathematical roles. Diffusion examples expose schedules, denoisers, sampler states, and reverse transitions. VAE/VQ-VAE/GAN-style examples expose latent variables, diagonal-Gaussian or KL terms, finite codebooks, nearest-code quantization, reconstruction losses, and generator/discriminator objectives. FNO/PINN examples use spectral operators, derivatives, and residual constraints; spline/ODE and Arb-backed workflows use external numerical producers followed by Lean-side certificate replay. We include these examples to show that modern ML workflows can be assigned named semantic objects that are suitable for later checking; the section is not a claim about leaderboard model quality.

The diagnostic layer makes semantic boundary failures explicit. The Bug Zoo covers cases such as causal-mask leakage, KV-cache/RoPE position drift, stable-loss boundaries, tokenizer/embedding mismatches, batch-invariance bugs, normalization-state bugs, Float32 edge cases, padding conventions in native kernels, and certificate/IR import mismatches. Each bug-shaped test is small on purpose: it isolates a boundary where a conventional ML system can silently change semantics. In \TorchLean{}, the corresponding property can be stated as a mask theorem, prefix-invariance check, schema check, shape invariant, finite-precision contract, or certificate replay obligation.

\section{Numerical Semantics}
\label{app:trust}

\paragraph{Motivation and trust boundaries.}
\label{app:trust-motivation}
We make numerical semantics explicit because it is easy to prove theorems about a real-valued model and then silently
execute a Float32 implementation whose behavior differs at exactly the corner cases that matter for verification
(rounding, overflow/underflow, NaN/Inf propagation, signed zeros, and library conventions).
IEEE~754 is the de facto standard for floating-point arithmetic: it specifies binary/decimal formats, rounding rules,
and exception behavior (including NaNs/Infs and their default handling).
In Lean, however, the built-in runtime floating-point types are \emph{opaque to the kernel}:
they are intended for computation and are implemented by external runtime code rather than reducible definitions in the logic \cite{leanfloatref}.
As a result, we treat the scalar type $\alpha$ as an explicit parameter and require every theorem or executable demo to
declare which numeric semantics it is using; this turns ``what arithmetic are we reasoning about?'' into part of the
statement rather than an implicit convention.

\textbf{IEEE~754 (what it standardizes).}
IEEE~754 defines floating-point numbers as signed significands with bounded exponents (e.g., binary32/Float32),
together with \emph{rounding modes} (typically round-to-nearest-even) and \emph{exceptional values} such as $\pm\infty$
and NaNs. It also specifies default behaviors for exceptional operations (e.g., division by zero, invalid operations)
and comparison/order conventions. For verification, many ``real-analysis proofs'' do not apply
verbatim in IEEE arithmetic: operations are rounded, may overflow/underflow, and may produce non-finite values whose
propagation rules are part of the semantics.

\textbf{Motivation from Flocq.}
Our design is inspired by Flocq, a mature Coq library for reasoning about floating-point arithmetic.
Flocq cleanly separates (i) the \emph{format} (which numbers are representable: radix, exponent bounds, subnormals)
from (ii) the \emph{rounding operator} (how an exact real result is mapped to a representable value), and it provides
theorem infrastructure for compositional error bounds (\citealp{boldo2011flocq}).
We adopt this same separation because neural network verification needs both:
\emph{(a)} theorem-friendly ``round-on-$\mathbb{R}$'' models to state and compose numerical error envelopes,
and \emph{(b)} executable, bit-level semantics to make corner cases (NaNs, signed zeros, overflow) concrete when running
full demos and checkers.

\begin{tcolorbox}[
  colback=black!2, colframe=black!35, boxrule=0.6pt, arc=2pt,
  left=6pt,right=6pt,top=5pt,bottom=5pt
]
\textbf{Summary: four numeric layers in \TorchLean{}.}
We provide (i) \emph{trusted runtime floats} for speed, (ii) an \emph{executable IEEE-754 model} for explicit Float32 behavior,
(iii) \emph{proof level rounding models} for compositional error reasoning, and (iv) an \emph{external validated-numerics oracle}
for rigorous high-precision enclosures of difficult primitives (notably transcendentals).
\end{tcolorbox}

\begin{tcolorbox}[
  colback=blue!3, colframe=blue!55!black, boxrule=0.6pt, arc=2pt,
  left=6pt,right=6pt,top=5pt,bottom=5pt
]
\textbf{(A) Runtime execution: \texttt{Float}/\texttt{Float32} (fast, but opaque).}
Lean exposes \texttt{Float} (binary64) and \texttt{Float32} (binary32) as runtime types intended for efficient execution.
Their operations are implemented by external/runtime code (e.g., compiled to underlying machine/C operators) rather than
by kernel-reducible definitions, so we treat them as an \emph{explicit trust boundary} for proofs \cite{leanfloatref}.
\end{tcolorbox}

\begin{tcolorbox}[
  colback=purple!3, colframe=purple!65!black, boxrule=0.6pt, arc=2pt,
  left=6pt,right=6pt,top=5pt,bottom=5pt
]
\textbf{(B) Executable Float32 semantics: \texttt{IEEE32Exec} (bit-level IEEE).}
For full demos that require explicit Float32 behavior, we implement \texttt{IEEE32Exec}, a Lean-defined,
bit-level IEEE-754 binary32 kernel. This model includes signed zeros, subnormals, NaNs/Infs, and the IEEE rounding/exception
behavior for core arithmetic. It is \emph{executable} inside Lean and therefore gives an explicit meaning to ``Float32 execution''
independent of a particular runtime implementation. The remaining ``hardware gap'' (showing a target platform refines this exact
software model) is handled as an explicit conformance interface rather than being silently assumed.
\end{tcolorbox}

\begin{tcolorbox}[
  colback=green!3, colframe=green!55!black, boxrule=0.6pt, arc=2pt,
  left=6pt,right=6pt,top=5pt,bottom=5pt
]
\textbf{(C) Proof-relevant rounding models: \texttt{FP32}/\texttt{NF} (round after each primitive).}
For theorem statements and error analysis, we use ``round-on-$\mathbb{R}$'' models inspired by Flocq: each primitive is specified
as ``compute in $\mathbb{R}$, then round'' under a chosen format/rounding mode, enabling compositional error envelopes through
entire graphs.
\texttt{FP32} instantiates this with IEEE-style binary32 rounding (finite-only, i.e., it does not model NaN/Inf),
while \texttt{NF} generalizes the same idea to user-chosen radices/exponent functions/rounding operators.
\end{tcolorbox}
\begin{tcolorbox}[
  colback=orange!3, colframe=orange!70!black, boxrule=0.6pt, arc=2pt,
  left=6pt,right=6pt,top=5pt,bottom=5pt
]
\textbf{(D) Validated numerics oracle: Arb/FLINT (rigorous transcendental enclosures).}
For rigorous enclosures of transcendental functions (e.g., $\tanh,\exp,\log,\sqrt,\sin,\cos$) at user-chosen precision,
we integrate an Arb/FLINT backend via \texttt{python-flint} behind an explicit oracle boundary. The oracle returns certified
ball/interval enclosures at a specified working precision (in bits), which we can consume as certificate-like artifacts in
verification pipelines. This layer is \emph{not} kernel-reducible and is treated as an external trust boundary; its role is to
provide high-quality validated bounds where implementing tight transcendental enclosures natively is cumbersome, while
keeping the trusted computing base explicit.
\end{tcolorbox}

\paragraph{Why we implement \emph{both} (bit-level and round-on-$\mathbb{R}$).}
These layers serve different proof/verification needs:
\emph{bit-level execution} (\texttt{IEEE32Exec}) is ideal for making ``what happens on Float32?'' concrete in checkers and demos,
including edge cases (NaNs, signed zeros) that real analysis ignores; \emph{round-on-$\mathbb{R}$ models} (\texttt{FP32}/\texttt{NF})
are ideal for theorem statements and compositional error envelopes because they expose the rounding operator as a mathematical
object that can be bounded and composed (in the style of Flocq (\citealp{boldo2011flocq})).
Together, they let us support both the verification community (explicit executable semantics) and the floating-point proof
community (compositional rounding/error reasoning) under a single semantic umbrella.

\textbf{Lean-specific note (why opacity matters).}
Because Lean's runtime floats are not encoded in the logic, the kernel cannot reduce or reason about them without additional
axioms; in particular, floating-point operations are implemented externally (and thus are not definitionally equal to any
mathematical model inside Lean).
This is why \TorchLean{} separates ``fast execution'' from ``proved semantics'' and makes the trust boundary explicit in both
the main text and Table~\ref{tab:trust-semantics}.

\textbf{One surface name, switchable semantics.}
To keep model code uniform while making numerical assumptions explicit, \TorchLean{} exposes a single surface notion of
``Float32'' with a selectable semantic mode. This lets the same model/program be instantiated with (i) an executable
bit-level IEEE-754 semantics for full runs, or (ii) a proof level rounding model for theorem statements and
compositional error envelopes, without rewriting the model.

\begin{lstlisting}[style=leanBox,caption={Selecting Float32 semantics at the type level.},label={lst:f32-modes}]
open TorchLean.Floats

-- Executable Float32 semantics (bit-level IEEE-754 binary32):
abbrev ExecF32 : Type := F32                 -- = IEEE32Exec

-- Proof Float32 semantics (round-on-R, finite-only):
abbrev ProofF32 : Type := F32 .fp32          -- = FP32
\end{lstlisting}

\textbf{Why multiple float modes are necessary.}
No single floating-point representation serves all verification goals. IEEE~754 defines the concrete behavior of deployed
floating-point arithmetic, including rounding and exceptional values (NaNs/Infs, signed zeros, subnormals). 
For \emph{executable demos} and certificate checking where these corner cases matter, we want an explicit, runnable model of
binary32 semantics---hence \texttt{IEEE32Exec}. For \emph{theorem statements} about numerical stability and graph level error
budgets, we instead want a proof-level ``round after each primitive'' model over $\mathbb{R}$ that supports compositional
error reasoning in the style of verified floating-point libraries such as Flocq---hence \texttt{FP32}/\texttt{NF}.
Finally, Lean's built-in runtime floats (\texttt{Float}/\texttt{Float32}) are fast but opaque to the kernel and therefore live on
the explicit trust/validation side of the interface.

\textbf{Float32 mode selection.}
At the API level we package this choice behind a small mode enum, so use sites are explicit about which semantics they rely on:
\begin{lstlisting}[style=leanBox,caption={A single ``Float32'' name with explicit semantics.},label={lst:Float32-select}]
inductive Float32Mode where
  | fp32 | ieee754Exec

abbrev Float32 (mode : Float32Mode := .ieee754Exec) : Type :=
  match mode with
  | .fp32        => FP32
  | .ieee754Exec => IEEE32Exec
\end{lstlisting}

\textbf{Proof-level rounding model (\texttt{NF}).}
For graph level error bounds we use a proof level model that rounds after every primitive operation. Conceptually, an
\texttt{NF} value stores a real number together with a chosen \emph{format} (radix and exponent function) and \emph{rounding}
operator; each primitive is specified as ``compute in $\mathbb{R}$, then round,'' and local rounding/error lemmas compose
over SSA/DAG graphs.

\noindent
\emph{Practical rule of thumb.} Use \texttt{IEEE32Exec} when you want executable Float32 behavior (including corner cases)
to match an IEEE-style semantics; use \texttt{FP32}/\texttt{NF} when you want theorem statements with explicit, compositional
rounding error envelopes; use runtime \texttt{Float}/\texttt{Float32} for fast prototyping where the numeric backend is an
explicitly trusted assumption. 

\begin{lstlisting}[style=leanBox,caption={Schematic runtime conformance contract for deployed Float32.},label={lst:f32-runtime-contract}]
-- A theorem using this class is conditional on the deployment target.
class RuntimeFloat32MatchesIEEE32Exec : Prop where
  add_finite : forall x y,
    finite x -> finite y -> finite (x + y) ->
    toIEEE32 (runtimeAdd x y) = IEEE32Exec.add (toIEEE32 x) (toIEEE32 y)
  mul_finite : forall x y,
    finite x -> finite y -> finite (x * y) ->
    toIEEE32 (runtimeMul x y) = IEEE32Exec.mul (toIEEE32 x) (toIEEE32 y)
  relu_exact : forall x,
    toIEEE32 (runtimeRelu x) = IEEE32Exec.max (toIEEE32 x) IEEE32Exec.posZero
\end{lstlisting}

We use this form intentionally. A graph theorem can be unconditional about \texttt{IEEE32Exec}, while a deployed-runtime theorem carries a visible assumption such as \texttt{RuntimeFloat32MatchesIEEE32Exec}. The assumption records the target configuration: rounding mode, denormal policy, fused operations, reduction order, and library implementation choices.

\subsection{Whole-graph NF bounds and hardware soundness}
\label{app:nf-compose}

This section explains two complementary pieces of our numerical story:
(i) how we \emph{compose} per-operator rounding/error lemmas into graph level error budgets in a ``round after each primitive'' model (NF),
and (ii) how we connect an \emph{executable} bit-level Float32 semantics to that proof model on the finite (non-NaN/non-Inf) path, and what
remains to relate either of them to \emph{hardware} execution.

\paragraph{Whole-graph NF bounds: local errors compose over SSA/DAG.}
NF is a proof level ``round after each primitive'' semantics: each primitive is specified as ``compute in $\mathbb{R}$, then round'' under a
chosen format/rounding operator. We relate NF values to real semantics via an explicit error relation $\,\approx_\varepsilon\,$, and we prove
graph level bounds by induction in SSA/topological order (and analogously for backward sweeps when needed).

\begin{tcolorbox}[
  colback=black!2, colframe=blue!55!black, boxrule=0.6pt, arc=2pt,
  left=6pt,right=6pt,top=5pt,bottom=5pt
]
\textbf{Local-to-global error propagation (NF idea).}
Define a scalar approximation relation
\[
x_{\mathrm{nf}} \approx_\varepsilon x_{\mathbb{R}}
\quad\Longleftrightarrow\quad
|x_{\mathrm{nf}}-x_{\mathbb{R}}|\le \varepsilon,
\]
and lift it to tensors (pointwise or via a norm such as $\ell_\infty$).
Each primitive operator $\mathrm{op}$ comes with a lemma of the form:
\[
\text{if } u_j \approx_{\varepsilon_j} v_j \text{ for all inputs } j,\ \text{ then }\ 
\mathrm{op}_{\mathrm{nf}}(u)\approx_{\varepsilon'} \mathrm{op}_{\mathbb{R}}(v),
\]
where the derived $\varepsilon'$ is computed from operator-specific stability bounds.
For linear/affine primitives, $\varepsilon'$ follows from standard Lipschitz constants; for nonlinear primitives, it follows from local
derivative bounds or global Lipschitz envelopes (on the region of interest).
Whole-graph bounds are proved by induction over the SSA/DAG order: the induction hypothesis provides bounds for parent nodes, and the node lemma
yields the bound for the current node.
\end{tcolorbox}

\textbf{Executable IEEE model $\rightarrow$ rounding-on-$\mathbb{R}$ model (internal refinement).}
NF/FP32-style reasoning is theorem-friendly because rounding is an explicit mathematical operator on $\mathbb{R}$, but it does not capture
IEEE corner-case behavior (NaNs/Infs, signed zeros, subnormals) directly. Conversely, a bit-level IEEE model is executable and makes those
corner cases concrete, but it is harder to use for compositional error proofs. We therefore establish an \emph{internal refinement} on the
\emph{finite} path: for the core Float32 arithmetic primitives (addition/subtraction/multiplication/division/sqrt and $\min/\max$),
we prove that interpreting the executable bit-level result as a real number agrees with applying the corresponding round-to-Float32 operator
to the real arithmetic result, under explicit side conditions that rule out NaN/Inf and overflow-to-Inf. \begin{BlueTheorem}{Internal refinement on the finite path}{ieee-to-round-refinement}
Let $\mathrm{ExecOp}_{32}$ be a core Float32 primitive implemented by the executable bit-level IEEE model
(e.g., $+,-,\times,/, \sqrt{\cdot}$, $\min/\max$), and let $\mathrm{RealOp}$ be the corresponding real operator. (For $\min/\max$ on finite values, $\mathrm{ExecOp}_{32}(x,y)$ returns exactly one of its inputs, so the $\mathrm{Round}_{32}$ wrapper below is vacuous; we state the refinement uniformly for all listed primitives.)
Let $\mathrm{toReal}$ interpret a finite Float32 value as a real number, and let $\mathrm{Round}_{32}$ be the IEEE-style
round-to-Float32 operator on $\mathbb{R}$.
Define $r(x,y)\;:=\;\mathrm{RealOp}\bigl(\mathrm{toReal}(x),\,\mathrm{toReal}(y)\bigr)$.
Assume the evaluation stays \emph{finite} (no NaN/Inf and no overflow-to-Inf) on the considered inputs.
Then for all inputs $x,y$ in the finite domain,
\[
\mathrm{toReal}\!\bigl(\mathrm{ExecOp}_{32}(x,y)\bigr)
\;=\;
\mathrm{Round}_{32}\!\bigl(r(x,y)\bigr).
\]
Moreover, these per-primitive refinements compose: for any expression $e$ built from the core primitives,
interpreting the executable result with $\mathrm{toReal}$ agrees with evaluating the corresponding real expression and
rounding at each primitive, under the same finiteness assumptions.
\end{BlueTheorem}
This bridge is necessary because it lets us
run full demos under an explicit IEEE-754 semantics while still reusing the cleaner FP32/NF error-envelope lemmas whenever the execution
stays in the finite regime.

\paragraph{Hardware soundness: what remains and a pragmatic deployment path.}
IEEE~754 specifies formats, rounding modes, and exception behavior, but real deployments can diverge through compilation and platform choices
(e.g., flush-to-zero/denormals-are-zero, reassociation/``fast-math'', FMA contraction, or reduction order).
To claim that \emph{hardware} \texttt{Float32} execution inherits our theorems, one must additionally (a) fix a target semantics contract
(rounding mode, denormal policy such as FTZ/DAZ, contraction/reassociation policy, and reduction ordering), and (b) relate the compiled runtime's
observable Float32 results to the chosen executable model (or to the rounding-on-$\mathbb{R}$ model) for the operations actually used. This remains target-level future work.

\textbf{Transcendentals and IEEE-754 (what is and is not specified).}
IEEE~754 precisely specifies formats and core arithmetic (e.g., add/sub/mul/div, sqrt, comparisons, NaN/Inf behavior), but
\emph{elementary/transcendental functions} (e.g., \texttt{exp}, \texttt{log}, \texttt{sin}, \texttt{cos}, \texttt{tanh}) are largely
outside the standard's required semantics and are typically provided by system \texttt{libm} implementations.
\footnote{IEEE's own background material explicitly notes that many programs rely on library elementary functions and that the
standard does not specify them.}
As a result, even when two platforms are ``IEEE compliant'' for core arithmetic, their transcendental results can differ across
OS/compiler/libm versions or under different optimization flags. A deeper reason is the \emph{table-maker's dilemma}:
deciding the correctly-rounded result for an elementary function can require substantially more precision than the target
format in worst cases, which makes full bit-exact specification and implementation nontrivial in practice~\citep{correctlyrounded_survey,zimm_talk}.

\textbf{Our deliberate split: executable IEEE core vs.\ explicit transcendental policy.}
Accordingly, \TorchLean{} separates concerns.
For core IEEE arithmetic, \texttt{IEEE32Exec} provides a Lean-defined, bit-level model of binary32 behavior (including signed zeros,
subnormals, NaN/Inf propagation, and rounding), so the meaning of ``Float32 execution'' is explicit inside the prover.
For transcendental functions, we make the policy explicit rather than pretending it is uniquely determined by IEEE:
\vspace{-0.1cm}
\begin{itemize}[leftmargin=*, itemsep=0em]
\item \emph{Deterministic in-kernel implementations for common ML primitives.}
For functions heavily used in ML pipelines (notably \texttt{exp}/\texttt{log} and hyperbolic functions used in activations and normalizers),
we provide deterministic implementations with a fixed rounding/approximation policy so that full executions are reproducible under
\texttt{IEEE32Exec}.

\item \emph{Explicit delegation when necessary.}
For functions outside the verified kernel surface (e.g., full trigonometric stacks), we may delegate to Lean runtime \texttt{Float}
(or a chosen library implementation) and then round back to binary32, treating that choice as an \emph{explicit trust boundary}
rather than an implicit semantic fact.
\end{itemize}

\textbf{Proof models cover transcendentals via ``round the real function.''}
In parallel, the FP32/NF proof models include transcendentals definitionally as ``apply the real function, then round,'' which
supports theorem statements in terms of explicit error envelopes and interval enclosures under stated hypotheses. This mirrors
the classical verified-float approach: instead of depending on a particular \texttt{libm} implementation, the semantics is a mathematical
rounding operator applied to the real function, enabling compositional reasoning~\citep{correctlyrounded_survey,flocq_theos}.

\textbf{NaN payload caveat (why hardware conformance is subtle).}
IEEE~754 specifies NaN propagation at a high level (operations produce NaN when given NaN inputs), but the \emph{choice of NaN payload}
(and some signaling-vs-quiet details) is not fully uniform across real implementations \cite{ieee754_2019}.
Our executable kernel fixes a deterministic policy (including a specific quieting/selection rule), and we prove properties about that
policy. Consequently, any claim that a concrete compiler/runtime/hardware Float32 implementation refines \texttt{IEEE32Exec} must assume
or establish compatibility with this concrete NaN policy (or adopt a quotienting notion of observational equivalence that treats NaN
payload differences as irrelevant for the target theorems).

\section{CROWN verification: bounds, duality, and certificates}
\label{app:verification}

We expand Section~\ref{sec:method-verification} with implementation details for bound propagation and certificate checking. 

\subsection{IBP, CROWN/LiRPA, and certificate checking}
\label{app:verification-ibp-crown}

\textbf{Overview.} Our verifier operates on the shared operator-tagged SSA/DAG IR and proves properties of the graph denotation
$\llbracket G\rrbracket$ by establishing \emph{sound enclosures} for intermediate values and/or outputs.
We support two complementary verification modes:
\emph{(i) native bound propagation} implemented in Lean for demo-scale graphs, and
\emph{(ii) certificate checking}, where an external verifier produces a bound/certificate that Lean checks against
the same IR semantics. 
\begin{tcolorbox}[
  colback=black!2, colframe=black!40, boxrule=0.6pt, arc=2pt,
  left=6pt,right=6pt,top=5pt,bottom=5pt
]
\textbf{Acronyms.}
\begin{itemize}[leftmargin=*, itemsep=0.2em]
\item \textbf{IBP} (Interval Bound Propagation): propagate interval boxes $[l,u]$ forward through the network.
\item \textbf{CROWN}: propagate \emph{linear} upper/lower bounds through nonlinearities to obtain tighter output bounds.
\item \textbf{LiRPA} (Linear Relaxation-based Perturbation Analysis): a general framework that derives/propagates
linear relaxations (including CROWN- and DeepPoly-style rules) over general computational graphs.
\item \textbf{$\alpha$-CROWN}: tightens CROWN bounds by optimizing relaxation parameters.
\item \textbf{$\beta$-CROWN}: incorporates split constraints (branch-and-bound) efficiently by encoding splits with
optimizable parameters.
\end{itemize}
\end{tcolorbox}

\paragraph{IBP (interval bound propagation).}
Interval Bound Propagation (IBP) is the simplest sound enclosure method for neural computation graphs. Starting from an
\emph{input region}---typically an axis-aligned box $[l_0,u_0]$ (often used to over-approximate an $\ell_\infty$ ball)---IBP
computes, for every node $i$ in the graph, an interval enclosure $[l_i,u_i]$ such that the true node value satisfies
$v_i(x)\in[l_i,u_i]$ for all admissible inputs $x\in[l_0,u_0]$. The algorithm is a single forward sweep in topological order:
each primitive provides an \emph{interval transfer rule} that maps parent bounds to a sound output bound, maintaining the
invariant ``parents sound $\Rightarrow$ child sound.'' IBP is attractive because it is fast, compositional, and easy to make
formally sound, but it can be loose because plain intervals do not track correlations between coordinates (and thus can
over-approximate significantly after repeated mixing through linear layers).

\begin{tcolorbox}[
  colback=blue!3, colframe=blue!60!black, boxrule=0.6pt, arc=2pt,
  left=6pt,right=6pt,top=5pt,bottom=5pt
]
\textbf{Example IBP transfer rules (sketch).}
All rules below are \emph{sound}: they guarantee $\mathrm{op}(v)\in[l',u']$ for all $v\in[l,u]$.
\begin{itemize}[leftmargin=*, itemsep=0.15em]
\item \textbf{Affine / linear} $y = Wx + b$ (includes \texttt{linear} and flattened \texttt{matmul} cases). 
Write $W=W^+ - W^-$ with $W^+,W^-\ge 0$ elementwise. Then
\[
l_y = W^+ l_x - W^- u_x + b,\qquad
u_y = W^+ u_x - W^- l_x + b,
\]
which is the standard interval-arithmetic enclosure for linear maps.

\item \textbf{Monotone elementwise} (e.g., ReLU, $\exp$). 
Apply the function to endpoints:
\[
[l_y,u_y] = [f(l_x), f(u_x)] \quad\text{(componentwise).}
\]
For ReLU specifically, this yields $l_y=\max(l_x,0)$ and $u_y=\max(u_x,0)$.

\item \textbf{Reductions} (e.g., sum/mean over a fixed axis). 
Because reduction is linear, bounds reduce by summing/averaging endpoints along the axis:
\[
\mathrm{sum}([l,u]) = [\mathrm{sum}(l),\,\mathrm{sum}(u)]\]\[
\mathrm{mean}([l,u]) = [\mathrm{mean}(l),\,\mathrm{mean}(u)].
\]
\end{itemize}
\end{tcolorbox}

\begin{lstlisting}[style=leanBox,caption={Schematic interval objects and the linear IBP step.},label={lst:ibp-linear-step}]
structure Box (s : Shape) where
  lo : Tensor Float s
  hi : Tensor Float s

-- W = Wpos - Wneg, with Wpos,Wneg >= 0 elementwise.
def ibpLinear (W b : Tensor Float _) (x : Box inShape) : Box outShape :=
  let Wpos := Tensor.map (fun w => max w 0.0) W
  let Wneg := Tensor.map (fun w => max (-w) 0.0) W
  { lo := Wpos.matmul x.lo - Wneg.matmul x.hi + b,
    hi := Wpos.matmul x.hi - Wneg.matmul x.lo + b }

def ibpRelu (x : Box s) : Box s :=
  { lo := Tensor.map (fun z => max z 0.0) x.lo,
    hi := Tensor.map (fun z => max z 0.0) x.hi }
\end{lstlisting}

The Lean checker does not need to trust the producer's explanation of how a bound was found. It recomputes the local transfer step from parent boxes and compares the recomputed box with the claimed box. The following schematic pattern is the core of the checker.

\begin{lstlisting}[style=leanBox,caption={Schematic replay step for a node certificate.},label={lst:cert-replay-step}]
def checkNode? (G : Graph) (cert : Cert) (st : CertState) (i : Nat) :
    Except CheckError CertState := do
  let node      <- G.node? i
  let parents   <- st.lookupAll node.parents
  let expected  <- transfer node.kind parents
  let claimed   <- cert.boxFor i
  guard (sameShape expected claimed)
  guard (claimed.lo <= expected.lo + cert.tol)
  guard (expected.hi <= claimed.hi + cert.tol)
  pure (st.insert i claimed)
\end{lstlisting}

The inequalities are intentionally one-sided. A certificate may be looser than the recomputed bound and still be sound; it may not claim a lower bound above the recomputed lower bound or an upper bound below the recomputed upper bound unless the checker can prove that stronger claim.

\paragraph{IBP certificate soundness (graph dialect).}
To support certificate checking, we formulate IBP soundness for a \emph{safe} graph semantics that is explicitly partial:
node evaluation fails when required parent values/parameters are missing or when declared dimensions do not match.
We then define a per-node IBP \emph{certificate step} that deterministically recomputes each node's interval box from its parents' boxes.
The soundness proof follows the standard local-to-global pattern: under a topological order and a supported operator subset,
local semantic consistency and local certificate consistency imply a global enclosure invariant.

\begin{BlueTheorem}{IBP certificate soundness}{ibp-cert-sound}
Let $G$ be a topologically sorted computation graph over a supported operator fragment, and fix a parameter store.
Assume we are given:
\emph{(i)} a partial value semantics that computes node values by a one-step rule from parent values and parameters, and
\emph{(ii)} a one-step IBP rule that computes an interval box for each node from its parents' boxes and parameters.
Let $(v_i)_i$ be any assignment of semantic node values that is locally consistent with the one-step value rule, and let $(B_i)_i$ be any assignment of certified interval boxes that is locally consistent with the one-step IBP rule.
If each semantic input lies in its declared input box, then for every node $i$, whenever both $v_i$ and $B_i$ are defined, the value $v_i$ lies in $B_i$ componentwise.
\end{BlueTheorem}

\paragraph{Whole-graph IBP soundness for the concrete implementation.}
Beyond certificate soundness, we also connect the theorem to the concrete implementations used in the demos.
We define total ``evaluate-by-id'' and ``propagate-by-id'' procedures, implemented by recursion on node id, and prove that they satisfy the local-consistency premises of Theorem~\ref{thm:ibp-cert-sound}. As a result, the computed IBP boxes enclose the computed semantic values whenever both procedures return values:

\begin{BlueTheorem}{Whole-graph IBP soundness for \ttbr{runIBP?}}{runibp-full pipeline}
Let \ttbr{evalGraphRec} be the total value evaluator for $G$ and \ttbr{runIBP?} the corresponding total IBP propagation.
Under the same structural assumptions (topological order, supported operators, and input enclosure), for every node $i$, whenever both procedures return a value $v_i$ and a box $B_i$, we have $v_i \in B_i$ componentwise.
\end{BlueTheorem}

\noindent\emph{Proof idea (both theorems).}
The enclosure invariant is preserved one node at a time: assuming all parent enclosures hold, the operator-specific IBP transfer rule yields an enclosure for the current node consistent with its value semantics.
Topological order ensures that, at node $i$, all parent facts are already established, so the global result follows by induction over node IDs/topological order.

\paragraph{CROWN and LiRPA (linear relaxations).}
CROWN tightens IBP by tracking \emph{affine bounds} rather than pure intervals: instead of only maintaining
$v_i(x)\in[l_i,u_i]$ for each node, it maintains linear forms that bound each node as a function of the input,
e.g.\ $a^\top x + b \le v_i(x)\le c^\top x + d$ over the admissible input region. These affine bounds are obtained by
replacing nonlinear primitives with \emph{sound linear envelopes} on the current pre-activation interval and propagating
the resulting affine forms through linear operators. For piecewise-linear activations such as ReLU, the envelope is given by
valid upper/lower lines on $[l,u]$ (secant/tangent choices in the unstable regime $l<0<u$); for smooth activations (e.g.,
$\tanh$, sigmoid), one uses tangent/secant bounds or other valid global/region-restricted linear relaxations, yielding a
strictly tighter enclosure whenever correlations between coordinates matter.

LiRPA is the unifying viewpoint: it treats bound propagation as computing
and composing such linear relaxations over \emph{general computational graphs} (not just simple feedforward chains),
subsuming CROWN- and DeepPoly-style rules. Modern LiRPA implementations typically expose both (i) a \emph{forward} pass
that produces node wise affine bounds and (ii) an \emph{objective-dependent} backward/dual pass that tightens the bound
for a specific linear objective on the output (e.g., a robustness margin objective), since the best relaxation choices can
depend on the downstream objective.
In \TorchLean{}, our Lean based core mirrors this structure: we provide a proved-sound IBP layer and a basic CROWN/LiRPA
affine engine over the shared operator-tagged IR, with conservative fallbacks (e.g., deriving constant affine bounds from IBP)
when a specialized relaxation is not yet implemented, preserving soundness at the cost of tightness.

\begin{tcolorbox}[
  colback=green!3, colframe=green!55!black, boxrule=0.6pt, arc=2pt,
  left=6pt,right=6pt,top=5pt,bottom=5pt
]
\textbf{Example: ReLU linear relaxation (one neuron).}
Let $z\in[l,u]$ and $y=\mathrm{ReLU}(z)$.
A standard sound relaxation chooses linear bounds
\[
  y \ge \lambda_L z + \mu_L,\qquad
  y \le \lambda_U z + \mu_U,
\]
with coefficients chosen so the lines bound ReLU on $[l,u]$ (e.g., $\lambda_U=\frac{u}{u-l}$, $\mu_U=-\frac{ul}{u-l}$
for the secant upper bound when $l<0<u$; and a valid lower bound selected from the admissible family in this regime).
Propagating these envelopes yields affine bounds on downstream nodes and, ultimately, output enclosures.
\end{tcolorbox}

\begin{lstlisting}[style=leanBox,caption={Schematic ReLU relaxation coefficients used by CROWN.},label={lst:crown-relu-relax}]
structure ReluRelax where
  lowerSlope     : Array Float
  lowerIntercept : Array Float
  upperSlope     : Array Float
  upperIntercept : Array Float

def reluRelax (lo hi : Tensor Float s) : ReluRelax :=
  Id.run do
    let mut ls := #[]; let mut li := #[]
    let mut us := #[]; let mut ui := #[]
    for k in [:lo.numel] do
      let l := lo.data[k]!; let u := hi.data[k]!
      if 0.0 <= l then
        ls := ls.push 1.0; li := li.push 0.0
        us := us.push 1.0; ui := ui.push 0.0
      else if u <= 0.0 then
        ls := ls.push 0.0; li := li.push 0.0
        us := us.push 0.0; ui := ui.push 0.0
      else
        let secant := u / (u - l)
        let offset := -(l * u) / (u - l)
        -- Upper line is the secant. Lower line is chosen from
        -- a sound family, possibly optimized by alpha-CROWN.
        us := us.push secant; ui := ui.push offset
        ls := ls.push 0.0;    li := li.push 0.0
    return { lowerSlope := ls, lowerIntercept := li,
             upperSlope := us, upperIntercept := ui }
\end{lstlisting}

\paragraph{$\alpha/\beta$-CROWN as certificate-checked bound semantics (graph dialect).}
We implement an $\alpha/\beta$-CROWN certificate interface for the operator-tagged graph verifier dialect. Fix a compiled graph $G$ (nodes $i$ in topological order) and an input box $B$. CROWN family bounds at each node are represented as \emph{affine enclosures} over inputs: for a scalar node value $z_i$ we enclose it by two affine functions of the input $x$,
\[
\begin{aligned}
z_i(x) &\in [L_i(x),U_i(x)],\\
L_i(x) &= \langle a_i^{L},x\rangle+b_i^{L},\\
U_i(x) &= \langle a_i^{U},x\rangle+b_i^{U}.
\end{aligned}
\]
(For tensor nodes, the same form is applied componentwise.) Bounds are computed by a per-node step rule
\[
\begin{aligned}
\texttt{CrownStepNode?} :\;& S\times i\ \to \\
& \textsf{Option}(\textsf{AffBounds}),
\end{aligned}
\]
where $S$ denotes the checker state (context plus parent bounds) and \textsf{AffBounds} abbreviates the pair of affine maps above. This step extends $\alpha$-CROWN by optionally incorporating $\beta$ phase information at ReLU nodes: phase-fixed units use exact linear behavior (slope $0$ or $1$), while unstable units fall back to standard CROWN/$\alpha$-CROWN relaxations.

\paragraph{Certificate contents and the role of $\alpha$ and $\beta$.}
An $\alpha/\beta$-CROWN certificate provides, for each node $i$: (i) an IBP pre-activation interval $[l_i,u_i]$ (a local box enclosure); (ii) affine bounds (lower and upper affine maps of the input); (iii) $\alpha$ parameters for the unstable-ReLU \emph{lower} relaxation; and (iv) an optional $\beta$ phase vector for ReLU units. Intuitively, $\alpha$ selects a member of a sound lower-envelope family in the unstable regime ($l<0<u$), while $\beta$ encodes phase constraints (active/inactive) that, when consistent with IBP, permit \emph{exact} linear behavior on those units.

\paragraph{$\beta$ phases and phase consistency.}
For a ReLU node with pre-activation $z$ and post-activation $y=\mathrm{ReLU}(z)$, we interpret $\beta\in\{-1,0,1\}$ as
\[
\beta \equiv
\begin{cases}
-1 & \text{inactive }(z\le 0),\\
\ \ 0 & \text{unstable (no constraint)},\\
\ \ 1 & \text{active }(0\le z).
\end{cases}
\]
The step checks \emph{phase consistency} against the IBP pre-activation interval $[l,u]$:
\[
\begin{aligned}
\textsf{Consistent}(l,u,\beta)\;:\!\iff\;& (\beta=-1 \Rightarrow u \le 0)\\
&\wedge\;(\beta=1 \Rightarrow 0 \le l).
\end{aligned}
\]
If consistent, inactive/active phases use exact linearization (slope $0$ or $1$). Unstable units ($\beta=0$) fall back to the usual CROWN upper relaxation and $\alpha$-CROWN lower relaxation.

\paragraph{Phase-dependent ReLU relaxations.}
Write a linear bound as a pair $(s,b)$ representing the function $s z+b$. Define phase-dependent relaxations by
\[
\begin{aligned}
\overline{r}_{\beta}(l,u) &=
\begin{cases}
(0,0) & \text{if }\beta=-1\ \text{and}\ u\le 0,\\
(1,0) & \text{if }\beta=1\ \text{and}\ 0\le l,\\
\overline{r}(l,u) & \text{if }\beta=0,
\end{cases}\\[0.25em]
\underline{r}_{\alpha,\beta}(l,u,\alpha) &=
\begin{cases}
(0,0) & \text{if }\beta=-1\ \text{and}\ u\le 0,\\
(1,0) & \text{if }\beta=1\ \text{and}\ 0\le l,\\
\underline{r}_{\alpha}(l,u,\alpha) & \text{if }\beta=0.
\end{cases}
\end{aligned}
\]
Here $\overline{r}(l,u)$ is the standard CROWN triangular \emph{upper} relaxation (secant in the crossing case), and $\underline{r}_{\alpha}(l,u,\alpha)$ is the $\alpha$-CROWN \emph{lower} relaxation family (with $\alpha\in[0,1]$ used only when $l<0<u$). In the inactive/active cases the relaxations reduce to the exact affine graphs of ReLU on the corresponding half-line.

\begin{tcolorbox}[
  colback=green!3, colframe=green!55!black, boxrule=0.6pt, arc=2pt,
  left=6pt,right=6pt,top=5pt,bottom=5pt
]
\small
\textbf{Operator-level soundness (ReLU, $\beta$-aware).}
Let $x\in[l,u]$ and assume $\textsf{Consistent}(l,u,\beta)$. Let $(\overline{s},\overline{b})=\overline{r}_{\beta}(l,u)$ and
$(\underline{s},\underline{b})=\underline{r}_{\alpha,\beta}(l,u,\alpha)$ (with $0\le \alpha\le 1$), the scalar relaxations satisfy:
\[
\operatorname{ReLU}(x)\ \le\ \overline{s}\,x+\overline{b} \text{ and } \underline{s}\,x+\underline{b}\ \le\ \operatorname{ReLU}(x).
\]
\end{tcolorbox}
\paragraph{Certificate checking: replay-based producer/checker design.}
Certificate checking is the key mechanism for obtaining tight bounds without
enlarging the trusted computing base to include a complex optimizer.
An external verifier or solver acts purely as an \emph{untrusted producer}:
it searches for a certificate (bounds, affine envelopes, split structure,
dual variables) and serializes the result as a compact JSON artifact
containing per-node IBP boxes, affine bound coefficients, and optional
per-node $\alpha$ and $\beta$ data.

Lean then acts as the \emph{trusted checker}: it parses the artifact,
canonicalizes all numeric data to a fixed float grid, and
\emph{recomputes} each node's bound via the same step semantics that
defines $\llbracket G \rrbracket$, producing
$B_i^{\mathrm{Lean}} := \texttt{CrownStepNode?}(\ldots, i)
\in \textsf{Option}(\textsf{AffBounds})$ for each node $i$ in
topological order.
The checker accepts only if (i)~provided bounds match the recomputed
bounds after canonicalization, (ii)~parent bounds appear in topological
order, and (iii)~shapes and operator tags are consistent with the IR\@.
The external optimizer is never trusted: regardless of what search or
heuristics it uses internally, only the final artifact crosses the trust
boundary, and every claim is independently replayed in Lean.
This reduces the trusted computing base to the IR denotation plus the
small checker.

\smallskip
\noindent\textit{What the checker validates at each node.}\enskip
For each node $i$, the checker performs three checks in sequence.
First, it validates the \emph{schema}: node $i$ exists in the graph $G$,
its declared input and output shapes match the IR typing, and all
required parent nodes have already been processed.
Second, it validates \emph{local soundness}: the certified bound
$B_i$ is implied by the parent bounds and the \texttt{OpKind}-specific
enclosure rule, confirmed by comparing $B_i$ against the recomputed
$B_i^{\mathrm{Lean}}$ after float-grid canonicalization.
Third, after all nodes are processed, it validates \emph{goal reduction}:
the target property (e.g., a robustness margin, a Lyapunov decrease
condition, or a PINN residual bound) follows from the certified output
enclosure by a small, explicit Lean argument.
If any check fails, the certificate is rejected and the property is
reported as unverified; the system never silently accepts a malformed
artifact.

\smallskip
\noindent\textit{Certificate families.}\enskip
Certificates fall into two families:
\begin{itemize}[leftmargin=*, itemsep=0.2em]
  \item \textbf{Node-wise enclosure certificates} supply $[l_i, u_i]$
    and optionally affine forms $L_i(x), U_i(x)$ per node, sufficient
    to discharge an output-level property directly from the propagated
    enclosure. IBP and CROWN/LiRPA both fall here; the difference is
    whether bounds are intervals or affine functions of the input.
  \item \textbf{Branch-and-bound leaf certificates} partition the input
    region $B$ into leaves, each with its own node wise certificate.
    The checker validates coverage, per-leaf soundness, and that the
    target property holds on every leaf; global validity follows by
    case analysis. This handles properties no single relaxation can
    certify, at the cost of a larger artifact and checking time linear
    in the number of leaves.
\end{itemize}

\begin{tcolorbox}[
  colback=green!3, colframe=green!55!black, boxrule=0.6pt, arc=2pt,
  left=8pt, right=8pt, top=6pt, bottom=6pt,
  before upper={\parskip=4pt},
  breakable=false
]
\small
\textbf{One ReLU neuron under $\alpha/\beta$-CROWN.}
\medskip

Let $z \in [l, u]$ be the pre-activation and $y = \mathrm{ReLU}(z)$.
The phase $\beta \in \{-1, 0, 1\}$ encodes what the IBP interval
tells us: provably off, provably on, or uncertain.
\medskip
\noindent\textit{Inactive} ($\beta{=}{-1}$, $u \le 0$).\enskip
The interval lies entirely in the non-positive half-line, so
$\mathrm{ReLU}(z) = 0$ everywhere on $[l,u]$.
Both bounds are the exact zero line $(s^U, b^U) = (s^L, b^L) = (0,0)$;
no relaxation gap is introduced.
\medskip
\noindent\textit{Active} ($\beta{=}1$, $0 \le l$).\enskip
The interval lies entirely in the non-negative half-line, so
$\mathrm{ReLU}(z) = z$ exactly.
Both bounds are the identity $(s^U, b^U) = (s^L, b^L) = (1,0)$;
again, zero gap and exact propagation.

\medskip
\noindent\textit{Unstable} ($\beta{=}0$, $l < 0 < u$).\enskip
The interval straddles zero and a linear relaxation is unavoidable.
The upper bound is the CROWN secant connecting $(l,0)$ to $(u,u)$,
tight at both endpoints:
\[
  s^U = \frac{u}{u-l}, \qquad b^U = \frac{-ul}{u-l}.
\]
The lower bound is the $\alpha$-CROWN family
$(s^L, b^L) = (\alpha, 0)$ with $\alpha \in [0,1]$ chosen by the
external optimizer to tighten the downstream objective.
$\alpha{=}0$ gives the trivial bound $y \ge 0$; $\alpha{=}1$ gives
the identity $y \ge z$, which is tightest when $z$ is likely positive.
Intermediate values trade local tightness against the global
verification goal; in practice $\alpha$ is optimized by gradient
ascent on the final bound, treating each neuron's slope independently.
\medskip

\noindent\textit{Propagation.}\enskip
Affine forms compose exactly through linear layers using the
$W = W^+ - W^-$ decomposition, where $W^+_{ij} = \max(W_{ij},0)$
and $W^-_{ij} = \max(-W_{ij},0)$ route lower and upper input bounds
to the correct output bound without case splits.
Inactive and active neurons pass forms through without loss.
Only unstable neurons accumulate gap,
which shrinks as IBP intervals tighten. After processing all nodes
in topological order, the output node carries certified bounds
$L_{\mathrm{out}}(x), U_{\mathrm{out}}(x)$ over the full input
box $B$. The checker then discharges the target property---typically
$L_{\mathrm{out},y}(x) > U_{\mathrm{out},k}(x)$ for all $k \ne y$---via
a small explicit Lean argument whose only premises are the computed
scalar bounds, keeping the trusted kernel minimal.
\end{tcolorbox}

\begin{tcolorbox}[
  colback=red!3, colframe=red!60!black, boxrule=0.6pt, arc=2pt,
  left=6pt,right=6pt,top=5pt,bottom=5pt
]
\textbf{What the Lean checker validates (conceptual checklist).}
Given a certificate and the IR graph:
\begin{enumerate}[leftmargin=*, itemsep=0.15em]
\item \textbf{Schema + typing:} node IDs exist; shapes match; per-node bound vectors have the right dimension.
\item \textbf{Local soundness per node:} for each node $i$, the certified bound is implied by parent bounds and the
\texttt{OpKind}-specific enclosure rule (interval rule or affine relaxation rule).
\item \textbf{Goal reduction:} the final property (e.g., robustness margin, Lyapunov decrease) is derived from the certified
output enclosure by a small, explicit argument in Lean (e.g., margin check or inequality chaining).
\item \textbf{Splitting (if present):} each leaf region is within the root region; the claimed per-leaf bound is sound; and the
set of leaves covers (or over-approximates) the target region according to the certificate's declared semantics.
\end{enumerate}
\end{tcolorbox}

\paragraph{A concrete example (robustness margin).}
For a classifier with logits $z(x)\in\mathbb{R}^K$ and a target label $y$, a standard sufficient condition for certified
$\ell_\infty$ robustness on a region $R$ is:
\[
  \underline{z_y} \;>\; \max_{k\ne y}\overline{z_k},
\]
where $(\underline{z_k},\overline{z_k})$ are sound bounds on each logit over $R$.
A certificate can therefore supply logit bounds (from IBP/CROWN/LiRPA, optionally with splitting), and the Lean checker
discharges robustness by checking the margin inequality.

\paragraph{Certificate schemas (illustrative).}
Certificate-driven checking hinges on a stable schema: the artifact names the graph and supplies node-indexed data.
At a high level, a node wise enclosure certificate looks like:
\begin{lstlisting}[style=leanBox]
{
  "graph_id": "...",
  "input_region": { ... },
  "bounds": {
    "node_id": { "lo": [...], "hi": [...] },
    ...
  }
}
\end{lstlisting}
More advanced artifacts may include affine coefficients, objectives, and branch-and-bound trees. The checker described here focuses
on validating the enclosure constraints needed to discharge the target theorem, rather than replaying the optimizer's full
search/parameter-optimization internally. Planned extensions (e.g., richer dual-feasibility checks for full solver internals)
are described in Appendix~\ref{app:verification-schemas}.

\subsection{VNN-COMP / VNN-LIB interface (ONNX + VNNLIB)}
\label{app:vnncomp}

VNN-COMP is an annual community competition designed to enable fair, objective comparison of neural network verification tools by standardizing interfaces, benchmarks, and evaluation pipelines \cite{brix2024vnncomp,vnncompSite}. VNN-COMP instances are packaged in standardized formats: networks are provided as ONNX models and properties are specified in VNN-LIB, an SMT-LIB-style language that defines both syntax and semantics for satisfiability queries over neural networks \cite{brix2024vnncomp,vnnlibSite}. This regime is directly relevant to semantic drift: most verifiers consume exported ONNX artifacts and interpret operator semantics outside the training framework, and ONNX operator meaning is tied to explicit operator-set (opset) versioning. Our goal here is not to re-implement the entire ONNX/VNN-LIB toolchain inside Lean, but to make the conversion boundary explicit and keep the trusted core inside Lean.

Accordingly, we adopt a producer/checker workflow. A lightweight Python export step reads the ONNX network and the VNN-LIB specification and emits a compact JSON bundle containing the network structure/weights (in a typed graph form), the input region (typically an axis-aligned box), and the property matrices/constraints extracted from VNN-LIB. Lean then compiles this bundle into our operator-tagged SSA/DAG IR and checks a \emph{sufficient UNSAT} condition by replaying IBP/CROWN style bounds against the shared IR semantics. For these suites we run the checker under the fast runtime \texttt{Float} backend (binary64), since the objective is to demonstrate a semantic benchmarking interface rather than Float32 deployment conformance.

The check we perform follows the standard sound-but-incomplete pattern used by bound-propagation verifiers: if the propagated output enclosure implies that the VNN-LIB predicate cannot hold on the entire input region, we return \texttt{safe} (UNSAT proved); otherwise we return \texttt{unknown}. Soundness is semantic: whenever Lean reports \texttt{safe}, the conclusion is derived from the Lean denotation of the compiled IR, so the trusted computing base is the IR semantics plus the small checker, with the export step treated as an explicit, auditable boundary.

Table~\ref{tab:vnncomp-mini} reports a small MNIST-FC slice under this interface. Lean IBP is conservative, while objective-dependent CROWN refutes a subset of properties inside Lean. Importing optimized ReLU $\alpha$ slopes primarily improves runtime rather than refutation count, suggesting that further tightness typically requires richer artifacts such as improved intermediate bounds and/or splitting, consistent with broader lessons emphasized in VNN-COMP reports \cite{brix2024vnncomp}.

This interface follows the same certificate/checker style as the rest of \TorchLean{}. The Lean runner replays native IBP/CROWN style bounds on the shared IR; stronger untrusted producers (e.g., optimized $\alpha/\beta$-CROWN bounds, split certificates, or other solver artifacts) can be treated as certificate generators whose outputs are validated against the IR semantics. This gives a concrete path toward theorem-prover participation in standardized verification benchmarks by making semantics and checking the reference point, rather than re-implementing every optimizer heuristic in Lean \cite{brix2024vnncomp,vnncompSite}.

\subsection{Case Studies}
\label{app:verification-schemas}

\textbf{Certified robustness workflow (classifier).}
We illustrate the full verification pipeline on \emph{local} certified robustness for a classifier.
Let $f:\mathbb{R}^d\to\mathbb{R}^K$ denote the network's \emph{logit} function (so the predicted label is
$\arg\max_k f_k(x)$), and fix a test input $x_0$ with nominal label $y=\arg\max_k f_k(x_0)$.
Given a perturbation radius $\varepsilon$, we define the region of interest
$R=\{x:\|x-x_0\|_\infty\le \varepsilon\}$ (or an axis-aligned box that over-approximates it). Using IBP or CROWN/LiRPA,
we compute sound bounds $(\underline{z}_k,\overline{z}_k)$ such that
$f_k(x)\in[\underline{z}_k,\overline{z}_k]$ for all $x\in R$. IBP provides a fast interval enclosure baseline, while
CROWN/LiRPA tightens these bounds by propagating sound linear relaxations of nonlinearities and (optionally) using
objective-dependent back-substitution.

\begin{BlueTheorem}{Logit-margin robustness certificate}{margin-cert}
Let $f:\mathbb{R}^d\to\mathbb{R}^K$ be a classifier and $R$ an input region.
Assume we have sound logit bounds for all classes: for each $k$ and all $x\in R$,
$f_k(x)\in[\underline{z}_k,\overline{z}_k]$.
If
\[
\underline{z}_y \;>\; \max_{k\ne y}\overline{z}_k,
\]
then $\arg\max_k f_k(x)=y$ for all $x\in R$ (i.e., the prediction is certified $\ell_\infty$-robust on $R$).
\end{BlueTheorem}

\textbf{How the certificate is checked.}
The certificate checking step is deliberately small and semantic: the checker validates that the reported bounds are a
sound enclosure of the IR semantics on $R$ (shape-consistent, op-consistent, and region-consistent), and then applies the
margin lemma above to conclude label invariance. This is exactly the ``infinite-to-finite'' reduction that makes bound
propagation useful: instead of enumerating all $x\in R$, we certify robustness by a finite set of inequalities on the
computed output enclosure. (This margin-style condition is standard in certified robustness pipelines.)

\textbf{Structure of the example.}
\begin{enumerate}[leftmargin=*, itemsep=0.25em]
\item \emph{Model definition:} define the classifier once in \TorchLean{}.
\item \emph{Lowering:} lower the same program to the shared operator-tagged SSA/DAG IR (the semantic target).
\item \emph{Bounding:} run IBP or CROWN/LiRPA on the IR to obtain logit enclosures over $R$.
\item \emph{Checking:} verify the enclosure constraints and discharge robustness via Lemma~\ref{thm:margin-cert}.
\end{enumerate}
The final claim is a Lean theorem about the denotation of the shared IR, not a post-hoc interpretation
of a tool output: external verifiers may \emph{produce} bounds, but Lean \emph{checks} that those bounds imply the semantic
property of interest on the executed artifact.

\textbf{Physics-Informed Neural Networks (PINNs).}
PINNs enforce physics by penalizing a PDE residual built from a neural field $u_\theta$ and its derivatives. We fix
(i) a residual operator $\mathcal{R}$ (e.g., for Burgers/heat-type equations), (ii) a domain $\Omega$, and (iii) a trained network $u_\theta$,
and aim to certify a uniform residual bound
$\sup_{x\in\Omega} |\mathcal{R}(u_\theta)(x)| \le \varepsilon$.
This is the standard PINN correctness goal: the learned model approximately satisfies the governing PDE across the domain.
The verification step is derivative-dependent, so we combine (a) a proved first-order reverse-mode result that anchors the meaning of
$\nabla u_\theta$ to the graph semantics (under the usual smoothness/pointwise side conditions), with (b) \emph{derivative-aware bound propagation}
on the shared operator-tagged graph that produces interval enclosures for $u_\theta$ and the specific derivatives appearing in $\mathcal{R}$
(including second derivatives for the low-dimensional smooth-operator subset used in the demos). We do \emph{not} obtain $u''$
by differentiating the backward-pass graph; instead, we bound the required derivatives directly via specialized bound-propagation passes on the
forward graph, and then combine these enclosures to conclude the residual inequality.

\begin{BlueTheorem}{Residual certificate pattern (PINNs)}{pinn-residual-cert}
Let $\Omega$ be a domain and suppose we have sound interval enclosures on $\Omega$ for every term that appears in the residual
$\mathcal{R}(u_\theta)$ (e.g., bounds on $u_\theta$, $\partial u_\theta/\partial t$, $\partial u_\theta/\partial x$, $\partial^2 u_\theta/\partial x^2$).
If these enclosures imply $|\mathcal{R}(u_\theta)(x)|\le \varepsilon$ for all $x\in\Omega$, then the PDE residual of $u_\theta$ is certified within
tolerance $\varepsilon$ on $\Omega$.
\end{BlueTheorem}

\begin{lstlisting}[style=leanBox,caption={Schematic residual checker for a one-dimensional Burgers-style PINN.},label={lst:pinn-residual-checker}]
structure DerivBounds where
  u    : Box scalar
  ut   : Box scalar
  ux   : Box scalar
  uxx  : Box scalar

def burgersResidual (nu : Float) (b : DerivBounds) : Box scalar :=
  -- R(u) = u_t + u * u_x - nu * u_xx
  b.ut + intervalMul b.u b.ux - scaleBox nu b.uxx

def checkResidual? (nu eps : Float) (b : DerivBounds) : Bool :=
  let r := burgersResidual nu b
  r.lo >= -eps && r.hi <= eps
\end{lstlisting}

The checker separates two obligations. First, the derivative bound passes produce sound boxes for \(u\), \(u_t\), \(u_x\), and \(u_{xx}\) over the verification region. Second, ordinary interval arithmetic combines those boxes according to the PDE residual formula. This separation keeps the theorem statement short: any future residual operator can reuse the same pattern once its derivative boxes and algebraic interval rules are available.

\textbf{Neural controller (Lyapunov-style safety/stability).}
We also evaluate a two stage controller-verification workflow common in learning-enabled control: Stage~1 (training/search) proposes a feedback
controller $u(x)$ together with a Lyapunov candidate $V(x)$; Stage~2 (certification) proves region-based inequalities that imply safety/stability
(e.g., $V(x)\ge 0$ and $\dot V(x)\le -\rho(\|x\|)$ on a region). This aligns with recent neural-control verification pipelines that use bound
propagation and $\alpha/\beta$-CROWN-style tooling to certify Lyapunov conditions.
In our setting, we compute enclosures for $V(x)$ and
$\dot V(x)=\nabla V(x)\cdot f(x,u(x))$ over the region (with $\nabla V$ grounded by the same autograd semantics as execution), and the final claim
is discharged by checking the resulting inequalities as a theorem about the shared IR denotation.

\begin{BlueTheorem}{Lyapunov certificate pattern (region-based)}{lyapunov-cert}
Let $V$ be a candidate Lyapunov function for a closed-loop system $\dot x=f(x,u(x))$ on a region $R$.
If certified bounds over $R$ establish $V(x)\ge 0$ (and $V(x)=0$ only at the equilibrium) and $\dot V(x)\le 0$ (or a strict decrease condition)
for all $x\in R$, then the corresponding safety/stability property holds on $R$ under the chosen Lyapunov criterion.
\end{BlueTheorem}

\begin{lstlisting}[style=leanBox,caption={Schematic Lyapunov inequality check over a certified region.},label={lst:lyapunov-checker}]
structure LyapBounds where
  Vlo   : Float
  Vhi   : Float
  VdotLo : Float
  VdotHi : Float

def checkLyapunovBox (b : LyapBounds) (rho : Float) : Bool :=
  -- Typical sufficient obligations on a box away from the equilibrium.
  (0.0 <= b.Vlo) && (b.VdotHi <= -rho)

def checkPartition (boxes : List LyapBounds) (rho : Float) : Bool :=
  boxes.all (fun b => checkLyapunovBox b rho)
\end{lstlisting}

This is the finite reduction used by the controller case study. The analytic Lyapunov criterion quantifies over a region; the certificate partitions the region into finitely many boxes and provides bounds for \(V\) and \(\dot V\) on each box. Lean checks the per-box inequalities and the coverage assumptions recorded in the certificate schema.

\textbf{Hopfield networks and global dynamics (complementary case study).}
Hopfield networks are recurrent, energy-based models whose core guarantees are \emph{global dynamical} statements: rather than certifying a local property around one input, we prove properties of entire \emph{trajectories} generated by repeatedly applying an update rule (e.g., asynchronous coordinate updates). Classically, one defines an energy/Lyapunov functional $E(x)$ and proves that each update step does not increase $E$; because the state space is finite, monotone energy then implies convergence to a fixed point (an attractor) under standard symmetry and zero-diagonal conditions. 

In the recent paper, Formalized Hopfield Networks and Boltzmann Machines \cite{cipollina2025hopfield}, the authors develop a Lean~4 formalization aimed
squarely at such global properties like convergence for deterministic Hopfield dynamics and ergodicity for stochastic Boltzmann machines. They emphasize that many ML formalizations mirror modern execution frameworks (sequential layers or DAG-style
computational graphs), which is convenient for feedforward computation but does not directly express recurrent update semantics; because their
focus is convergence/ergodicity of \emph{recurrent} models, they adopt a graph-based dynamical-systems perspective (directed graphs/Markov
kernels) integrated with mathlib/PhysLean~\citep{mathlib2020}.
Their discussion highlights a real tradeoff: representations optimized for feedforward dataflow can require extra
machinery (unfolding or fixed-point semantics) to model recurrence faithfully.

Our treatment is complementary: we express the Hopfield update operator and energy functional \emph{inside the same typed tensor/program semantics} used throughout \TorchLean{}, and we prove the standard monotone-energy and convergence-style results by reasoning about an explicit state-transition map (iteration) plus a Lyapunov decrease argument. These results show that an SSA/DAG semantic core for \emph{computation} can still support proofs about \emph{global} dynamical behavior.

\begin{BlueTheorem}{Hopfield energy decrease and convergence (informal template)}{hopfield-energy}
Let $x\in\{-1,+1\}^n$ be the network state, with symmetric weights $W=W^\top$ and zero diagonal, and thresholds $\theta$.
Define the energy
\[
E(x)\;=\;-\tfrac12\, x^\top W x\;+\;\theta^\top x.
\]
Consider asynchronous updates that flip a single coordinate $x_i$ according to the sign of its local field
$(Wx-\theta)_i$ (with a fixed tie-handling convention).
Then each update step satisfies $E(x^{t+1}) \le E(x^t)$, and hence any trajectory produced by repeatedly applying the update
rule reaches a fixed point in finitely many steps (a local minimum of $E$ under the update dynamics).
\end{BlueTheorem}

\section{Universal Approximation Theorems}
\label{app:uat}

\paragraph{Universal approximation (real-valued).}
A universal approximation theorem (UAT) formalizes the idea that a simple network class can approximate \emph{any} continuous
function on a compact domain to arbitrary accuracy. Classic results show that single-hidden-layer networks are dense in
$C(K)$ for compact $K\subseteq\mathbb{R}^d$ under mild activation assumptions \cite{cybenko1989approximation,hornik1991approximation};
more recent work characterizes approximation rates for ReLU networks \cite{yarotsky2017relu}.
In \TorchLean{}, we mechanize a standard real-valued ReLU UAT on compact domains and then specialize it to boxes such as $[-1,1]^d$.

\begin{BlueTheorem}{ReLU universal approximation on compact sets}{relu-uat-compact}
For every compact $K \subseteq \mathbb{R}^d$, every continuous $f:K\to\mathbb{R}$, and every $\varepsilon>0$,
there exist a width $m$ and a single-hidden-layer ReLU network $N_\theta$ such that
\[
  \sup_{x\in K}\,|N_\theta(x)-f(x)| < \varepsilon.
\]
\end{BlueTheorem}

\subsection{Float32-execution Soundness}
\label{app:f32soundness}

\paragraph{Why the Float32 setting is harder.}
Real-valued UATs reason about exact arithmetic, whereas deployed models run under finite-precision float semantics.
In this setting, even the ``target function'' must be interpreted carefully: the executed program computes a \emph{rounded}
function, and abstract-interpretation semantics (e.g., intervals) must account for rounding and finite-range effects.
Recent work proves a floating-point analog of \emph{interval universal approximation} (IUA), showing that floating-point networks
can capture the \emph{direct image map} of a suitably rounded target function under interval semantics \cite{hwang2025fpua}.
In contrast, our goal in \TorchLean{} is to build a compositional, checkable foundation that supports practical verification
workflows under explicit Float32 semantics, without claiming the full ``exact-hull direct image'' property for arbitrary networks.

\textbf{Our claims.}
We do \emph{not} claim a full IUA theorem of the form ``the interval semantics returns the exact output-range hull for every input box''
for arbitrary networks under Float32 execution. Instead, we mechanize three ingredients that are sufficient for
approximation-style arguments and verification pipelines:
\begin{itemize}[leftmargin=*, itemsep=0.25em]
\item \textbf{Executable Float32 semantics in the prover.} An executable IEEE-754 binary32 model (\texttt{IEEE32Exec})
that gives a precise internal meaning to ``Float32 execution'' (including corner cases) for full demos and checking.
\item \textbf{Sound interval evaluation on a supported fragment.} A theorem-ready specification of interval semantics and a proved-sound
interval evaluator for a small ReLU-MLP fragment, enabling certified enclosures needed by robustness/PINN/controller checks.
\item \textbf{Error decomposition for Float32-execution approximation.} Theorem templates that bound execution error by an explicit sum of terms:
(real approximation error) $+$ (parameter/representation error, e.g.\ quantization) $+$ (IEEE rounding/execution error),
under explicit finiteness/no-NaN/no-Inf hypotheses.
\end{itemize}

\textbf{Note:} Where \cite{hwang2025fpua} targets a floating-point IUA result that exactly matches the interval \emph{direct image map} of a rounded
target function, our development emphasizes compositional error bounds and sound enclosures that integrate cleanly with certificate checking
and full verification workflows, without requiring the full exact-hull IUA guarantee for all networks.

\paragraph{Notation for the theorems below.}
Let $\mathbb{F}_{32}$ denote IEEE-style Float32 values (excluding NaNs/Infs when we explicitly assume \emph{finite execution}),
and let $\mathrm{toReal}:\mathbb{F}_{32}\to\mathbb{R}$ interpret a finite float as a real number.
A \emph{box} $B$ is an axis-aligned product of Float32 intervals (a valid element of the interval domain), and
$\gamma(B)\subseteq \mathbb{F}_{32}^d$ denotes its \emph{concretization} (the set of Float32 points represented by $B$).
We write $\nu:\mathbb{F}_{32}^d\to\mathbb{F}_{32}$ for point semantics (Float32 execution) and
$\nu^\sharp:\textsf{Box}_{32}^d\to \textsf{Box}_{32}$ for interval semantics (abstract interpretation over boxes).

\begin{BlueTheorem}{Exact-hull specification (Float32 interval semantics)}{iua-exact-hull-shape}
Fix $d\in\mathbb{N}$ and a NaN-free target function $\hat f:\mathbb{F}_{32}^d\to\mathbb{F}_{32}$.
Assume a point semantics $\nu$ and an interval semantics $\nu^\sharp$ such that $\nu^\sharp(B)$ is a box output for any valid input box $B$.
We say $\nu^\sharp$ returns the \emph{exact hull} of $\hat f$ on $B$ if, for every valid box $B\subseteq [-1,1]^d$, there exist witnesses
$m,M\in\mathbb{F}_{32}$ satisfying:
\begin{enumerate}[leftmargin=*, itemsep=0.2em]
\item (\textbf{Witnessed extrema}) $m$ and $M$ witness the range of $\hat f$ on the concretization:
      $m \le \hat f(x) \le M$ for all $x\in\gamma(B)$, and there exist $x_m,x_M\in\gamma(B)$ such that
      $\hat f(x_m)=m$ and $\hat f(x_M)=M$.
\item (\textbf{Hull equality}) the concretization of the abstract output box is exactly the interval hull:
\[
  \gamma(\nu^\sharp(B)) \;=\; \{\,y\in\mathbb{F}_{32}\mid m\le y \wedge y\le M\,\}.
\]
\end{enumerate}
\end{BlueTheorem}

\noindent
\emph{Why witnesses?}
This avoids committing to a particular implementation of float \texttt{min}/\texttt{max} at ties or NaN boundaries: the statement only requires
existence of range witnesses and equality to the induced hull, which is the robust ``Eq.(14)-style'' specification used in float-IUA work.

\begin{BlueTheorem}{From real approximation to Float32-execution error (1D hinge template)}{ieee-approx-1d}
Let $f:[a,b]\to\mathbb{R}$ be a real target function and let $h_{\mathbb{R}}$ be a hinge-style 1D network (or piecewise-linear ReLU template)
parameterized by $\theta$.
Assume a real approximation bound
\[
  \sup_{x\in[a,b]}\,|f(x)-h_{\mathbb{R}}(x;\theta)| \le \varepsilon_{\textsf{approx}}.
\]
Let $h_{32}$ denote the corresponding Float32-executed program (same structure/parameters, executed in Float32) and assume
\emph{finite execution} on $[a,b]$ (no NaN/Inf and no overflow-to-Inf along the evaluation).
Then for every real input $x\in[a,b]$ and any Float32 encoding $x_{32}$ of $x$,
\[
  \bigl| f(x) - \mathrm{toReal}(h_{32}(x_{32};\theta)) \bigr|
  \;\le\;
  \varepsilon_{\textsf{approx}} \;+\; \varepsilon_{\textsf{exec}}(x),
\]
where $\varepsilon_{\textsf{exec}}(x)$ is an explicit rounding/execution envelope for this template on $[a,b]$
(e.g., a compositional bound obtained by summing per-primitive rounding terms under the finiteness hypothesis).
\end{BlueTheorem}

\label{app:mlp-worked}
\begin{BlueTheorem}{Three-term Float32-execution error decomposition (single-hidden-layer ReLU MLP)}{f32-three-term}
Let $f:K\to\mathbb{R}$ be a target function on a domain $K\subseteq\mathbb{R}^d$ and let $N_{\mathbb{R}}(\cdot;\theta_{\mathbb{R}})$ be a
real single-hidden-layer ReLU MLP. Let $N_{32}(\cdot;\theta_{32})$ be its Float32-executed counterpart with Float32 parameters $\theta_{32}$, and let
$\theta_{\mathbb{R}}' := \mathrm{toReal}(\theta_{32})$ be the real interpretation of those float parameters.
Assume:
\begin{enumerate}[leftmargin=*, itemsep=0.2em]
\item (\textbf{Real approximation}) \[\sup_{x\in K} |f(x)-N_{\mathbb{R}}(x;\theta_{\mathbb{R}})| \le \varepsilon_{\textsf{UAT}}\]
\item (\textbf{Parameter/representation gap}) \[\sup_{x\in K} |N_{\mathbb{R}}(x;\theta_{\mathbb{R}})-N_{\mathbb{R}}(x;\theta_{\mathbb{R}}')|
      \le \varepsilon_{\textsf{param}}\]
\item (\textbf{Float32 execution gap}) for all admissible inputs (under a finiteness hypothesis),
      \[|N_{\mathbb{R}}(x;\theta_{\mathbb{R}}')-\mathrm{toReal}(N_{32}(x_{32};\theta_{32}))| \le \varepsilon_{\textsf{exec}}\]
\end{enumerate}
Then for all $x\in K$ (with any Float32 encoding $x_{32}$), where any input-encoding/quantization error can be bounded separately and added to the right-hand side (or folded into $\varepsilon_{\textsf{exec}}$),
\[
  \bigl| f(x) - \mathrm{toReal}(N_{32}(x_{32};\theta_{32})) \bigr|
  \;\le\;
  \varepsilon_{\textsf{UAT}} + \varepsilon_{\textsf{param}} + \varepsilon_{\textsf{exec}}
\]
\end{BlueTheorem}

\vspace{-0.5cm}
\section{Limitations and Discussion}
\label{app:expanded-discussion}

\paragraph{Limitations and near-term roadmap.}
\textbf{Execution vs.\ training at scale.}
\TorchLean{} is semantic infrastructure, not a throughput-optimized training stack. Lean execution and the optional CUDA backend
prioritize a machine-checked semantic link between (i) the program users write, (ii) the IR graph we reason about, and (iii) the artifacts we
check, rather than matching industrial training throughput. The CUDA path accelerates selected kernels and examples behind explicit native runtime contracts; large-scale training can also be performed externally and imported as weights/structure, with Lean used for semantic checking, certificate validation, and proof-carrying artifacts. Scaling the native backend, kernel fusion, and internal bound optimization remains an engineering extension of the same semantic design.

\textbf{Verification scope and certificates.}
Our native verification layer covers a curated IR fragment with a proved-sound IBP core and a CROWN/LiRPA style affine engine, together with
certificate checking infrastructure for externally produced bound artifacts. We support an $\alpha/\beta$-CROWN certificate
dialect for the graph-based verifier: certificates may supply IBP pre-activation boxes, affine bounds, $\alpha$ parameters for unstable-ReLU
lower relaxations, and optional $\beta$ phase vectors that are checked for consistency and then replayed by the Lean step semantics. When
strongest current tightness is required beyond this interface, we rely on external optimizers as \emph{untrusted producers} and check their
exported bounds/leaf certificates against the shared IR semantics, keeping the trusted computing base to the Lean checker plus the IR
denotation. Extending the checker to validate additional solver families (e.g., richer dual-feasibility conditions, cutting-plane certificates,
or SDP-based relaxations) is left to future extensions; Appendix~\ref{app:verification} and
Appendix~\ref{app:verification-schemas} document the present schemas and extension points.

\textbf{Float32 trust boundary (the ``hardware gap'').}
We separate (i) real-valued reference semantics, (ii) proof level rounding-on-$\mathbb{R}$ models (FP32/NF) for compositional error
envelopes, and (iii) an executable bit-level Float32 kernel (\texttt{IEEE32Exec}). What remains is a \emph{target-specific} refinement from a
deployed \texttt{Float32} toolchain to the executable model: one must fix a deployment configuration (rounding mode, denormal policy, contraction/
reassociation, and reduction order constraints) and then discharge a conformance interface (by proof or explicit validation). Appendix~\ref{app:nf-compose}
summarizes the internal refinements we prove and the remaining target-level obligations.

\textbf{Deployment and code generation.}
Inside Lean, we close the training and verification semantic gap by making the operator-tagged IR the single semantic target. Deploying a verified model on an
embedded target still introduces a translation step: either run a small IR interpreter on-device, or generate C/Rust code from the IR and connect
its behavior back to the Lean denotation. Lean's compilation pipeline and runtime already support efficient compiled code, and there are emerging
toolchains that bridge verified Lean reasoning with real systems languages (e.g., Rust-to-Lean verification pipelines), suggesting concrete paths
toward deployment pipelines.

\end{document}